\documentclass[11pt,letterpaper]{article}
\usepackage{microtype}

\usepackage[utf8]{inputenc}

\usepackage{a4wide}
\usepackage{fullpage}
\usepackage{latexsym}
\usepackage{amssymb,amsmath,amsthm}
\usepackage{graphicx}
\usepackage{url}
\usepackage{hyperref}
\usepackage{xcolor}
\usepackage{enumerate}
\usepackage[capitalize,noabbrev]{cleveref}
\usepackage{multirow}
\usepackage{authblk}
\usepackage{xspace}

\newtheorem{theorem}{Theorem}
\newtheorem{proposition}[theorem]{Proposition}
\newtheorem{observation}[theorem]{Observation}
\newtheorem{lemma}[theorem]{Lemma}
\newtheorem{corollary}[theorem]{Corollary}

\newtheorem*{claim*}{Claim}

\newtheorem*{tswappable}{Theorem~\ref{t:swappable}}

\theoremstyle{definition}

\newtheorem*{problem*}{Problem}

\theoremstyle{remark}

\makeatletter
\newcommand{\ProofEndBox}{{\ifhmode\unskip\nobreak\hfil\penalty50 \else
          \leavevmode\fi\quad\vadjust{}\nobreak\hfill$\Box$
            \finalhyphendemerits=0 \par}}
\makeatother

\newcommand{\R}{{\mathbb{R}}}

\newcommand{\D}{{\mathcal{D}}}

\newcommand{\Z}{{\mathbb{Z}}}

\newcommand\RR{\mathcal{R}}

\newcommand{\Imin}{\ensuremath{\rm{I}^-}\xspace}
\newcommand{\Ipl}{\ensuremath{\rm{I}^+}\xspace}
\newcommand{\I}{\ensuremath{\rm{I}}\xspace}
\newcommand{\II}{\ensuremath{\rm{II}}\xspace}
\newcommand{\III}{\ensuremath{\rm{III}}\xspace}
\newcommand{\IImin}{\ensuremath{\rm{II}^-}\xspace}
\newcommand{\IIpl}{{\ensuremath{\rm{II}^+}}\xspace}

\newcommand{\epsvec}{\protect\overrightarrow{\varepsilon}}
\newcommand{\phivec}{\protect\overrightarrow{\varphi}}

\newcommand{\sctt}[1]{{\sc \texttt{#1}}}

\newcounter{sideremark}

\DeclareMathOperator{\interior}{int}
\DeclareMathOperator{\crr}{cr}
\DeclareMathOperator{\defe}{def}

\title{Parameterized complexity of untangling knots\thanks{Part of this work has
been done when C.~L.-D. visited Charles University, whose stay was partially
supported by the GA\v{C}R grant 19-04113Y.  M.~T. is supported by the GA\v{C}R grant 19-04113Y.}}
\date{}

\author[1]{Cl\'ement Legrand-Duchesne}
\author[2]{Ashutosh Rai}
\author[3]{Martin Tancer}

\affil[1]{\small Univ. Bordeaux, CNRS, Bordeaux INP, LaBRI, UMR 5800, F-33400 Talence, France}
\affil[2]{\small Department of Mathematics, IIT Delhi, Hauz Khas, New Delhi,
110016, India}
\affil[3]{\small Department of Applied Mathematics, Faculty of Mathematics and
Physics, Charles University, Malostransk\'{e} n\'{a}m.
25, 118~00~~Praha~1, Czech Republic}

\begin{document}
\maketitle
\begin{abstract}
  Deciding whether a diagram of a knot can be untangled with a given number of
  moves (as a part of the input) is known to be NP-complete. In this paper we
  determine the parameterized complexity of this problem with respect to a
  natural parameter called defect. Roughly speaking, it measures the efficiency
  of the moves used in the shortest untangling sequence of Reidemeister moves.
  
  We show that the \IImin moves in a shortest untangling sequence
  can be essentially performed greedily. Using that, we show that this problem belongs to W[P]
  when parameterized by the defect. We also show that this problem is W[P]-hard
  by a reduction from \sctt{Minimum axiom set}.
\end{abstract}

\section{Introduction}
A classical and extensively studied question in algorithmic knot theory is to
determine whether a given diagram of a knot is actually a diagram of an
unknot. This question is known as the \emph{unknot recognition problem}.  The
first algorithm for this problem was given by Haken~\cite{haken61}.  Currently,
it is known that the unknot recognition problem belongs to NP$\, \cap\, $co-NP
but no polynomial time algorithm is known. See~\cite{hass-lagarias-pippenger99}
for the NP-membership and~\cite{lackenby16arxiv} for co-NP-membership
(co-NP-membership modulo Generalized Riemann Hypothesis was previously
established in~\cite{kuperberg14}). In addition, a quasi-polynomial time
algorithm for unknot recognition has been recently announced by
Lackenby~\cite{lackenby21_announced}.

One possible path for attacking the unknot recognition problem is via
Reidemeister moves (see~\cref{f:rm}): if $D$ is a diagram of an unknot, then $D$
can be untangled to a diagram $U$ with no crossing by a finite number of
Reidemeister moves.  In addition, Lackenby~\cite{lackenby15} provided a
polynomial bound (in the number of crossings of $D$) on the required number of
Reidemeister moves. This is an alternative way to show that the unknot
recognition problem belongs to NP, because it is sufficient to guess the
required Reidemeister moves for unknotting.

However, if we slightly change our viewpoint, de~Mesmay, Rieck, Sedgwick, and
Tancer~\cite{demesmay-rieck-sedgwick-tancer21} showed that it is NP-hard to
count the number of required Reidemeister moves exactly. (An anologous result
for links has been shown to be NP-hard slightly earlier by~Koenig and
Tsvietkova~\cite{koenig-tsvietkova21}.) More precisely, it is shown
in~\cite{demesmay-rieck-sedgwick-tancer21} that given a digram $D$ and a
parameter $k$ as input, it is NP-hard to decide whether $D$ can be untangled
using at most $k$ Reidemeiser moves. For more background on unknotting and
unlinking problems, we also refer to Lackenby's survey~\cite{lackenby17}.

The main aim of this paper is to extend the line of research started
in~\cite{demesmay-rieck-sedgwick-tancer21} by determining the parameterized
complexity of untangling knots via Reidemeister moves. On the one hand, it is
easy to see that if we consider parameterization by the number of Reidemeister
moves, then the problem is in FPT (class of \emph{fixed parameter tractable}
problems).  This happens because of a somewhat trivial reason: if a diagram $D$
can be untangled with at most $k$ moves, then $D$ contains at most $2k$
crossings, thus we can assume that the size of $D$ is (polynomially) bounded by
$k$.  On the other hand, we also consider parameterization with an arguably much
more natural parameter called the defect (used
in~\cite{demesmay-rieck-sedgwick-tancer21}). This parameterization is also
relevant from the point of view of \emph{above guarantee parameterization}
introduced by Mahajan and Raman~\cite{MahajanR99}. Here we show that the problem
is W[P]-complete with respect to this parameter. This is the core of the paper.

In order to state our results more precisely, we need a few preliminaries on
diagrams and Reidemeister moves. For purposes of this part of the introduction,
we also assume that the reader is at least briefly familiar with complexity
classes FPT and W[P]. Otherwise we refer to Subsection~\ref{ss:parameterized_complexity} 
where we briefly overview these classes, and
to the references in this subsection.

\paragraph{Diagrams and Reidemeister moves.}
A \emph{diagram} of a knot is a piecewise linear map $D\colon S^1 \to \R^2$ in
general position; for such a map, every point in $\R^2$ has at most two
preimages, and there are finitely many points in $\R^2$ with exactly two
preimages (called \emph{crossings}). Locally at crossing two arcs cross each
other transversely, and the diagram contains the information of which arc passes
`over' and which `under'. This we usually depict by interrupting the arc that
passes under.  (A diagram usually arises as a composition of a (piecewise
linear) knot $\kappa \colon S^1 \to \R^3$ and a generic projection $\pi \colon
\R^3 \to \R^2$ which also induces `over' and `under' information.) We usually
identify a diagram $D$ with its image in $\R^2$ together with the information
about underpasses/overpasses at crossings.  Diagrams are always considered up-to
isotopy. The unique diagram without crossings is denoted $U$ (untangled).
See~\cref{f:d_example} for an example of a diagram.
\begin{figure}
\begin{center}
\includegraphics{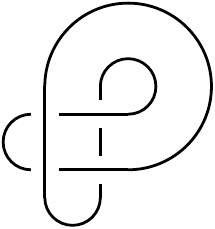}
  \caption{An example of a diagram.}
  \label{f:d_example}
\end{center}
\end{figure}

Let $D$ be a diagram of a knot. \emph{Reidemeister moves} are local modifications of
the diagram depicted at~\cref{f:rm}. We distinguish Reidemeister moves of
types $\I$, $\II$, and $\III$ as depicted in the figure. In addition, for
types $\I$ and $\II$, we distinguish whether the moves remove crossings (types
$\Imin$ and $\IImin$) or whether they introduce new crossings (types $\Ipl$ and
$\IIpl$). 

\begin{figure}
\begin{center}
  \includegraphics{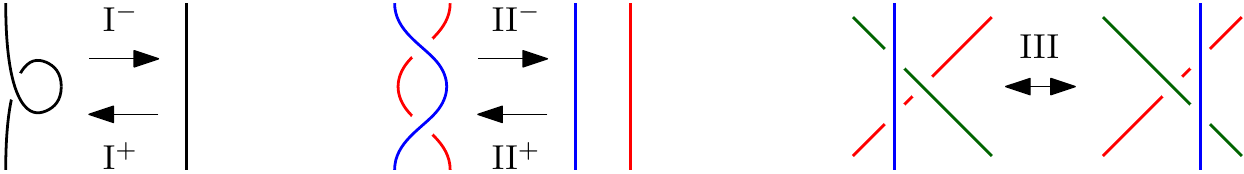}
\end{center}
\caption{Reidemeister moves}
  \label{f:rm}
\end{figure}

A diagram $D$ is a \emph{diagram of an unknot} if it can be transformed to the
untangled diagram $U$ by a finite sequence of Reidemeister moves. (This is well
known to be equivalent to stating that the lift of the diagram to $\R^3$,
keeping the underpasses/overpasses is ambient isotopic to the unknot, that is
standardly embedded $S^1$ in $\R^3$.)   The diagram on~\cref{f:d_example}
is a diagram of an unknot. 
\bigskip

Now we discuss different parameterizations of untangling (un)knots via
Reidemeister moves in more detail.

\paragraph{Parameterization via number of Reidemeister moves.}
In the first case, we consider the following parameterized
problem.

\begin{problem*}[\sctt{Unknotting via number of moves}] \

  \begin{tabular}{ll}
  \sctt{Input} & A diagram $D$ of a knot.\\
  \sctt{Parameter} & $k$.\\
  \sctt{Question} & Can $D$ be untangled to an unknot using at most $k$
  Reidemeister moves?\\
\end{tabular}
\end{problem*}

\begin{observation}
  \sctt{Unknotting via number of moves} belongs to FPT.
\end{observation}

\begin{proof}[Sketch]
  Each Reidemeister move removes at most $2$ crossings. Thus, if $D$ contains at
  least $2k + 1$ crossings, then we immediately answer \sctt{No}. This can be
  determined in time $O(|D|)$, where $|D|$ stands for the size of the encoding
  of $D$.
  
  It remains to resolve the case when $D$ contains at most $2k$ crossings. In
  this case we perform the naive algorithm trying all possible sequences of
  Reidemeister moves. We will see that if $D$ has $n$ crossings, then there are
  at most $O(n^2)$ feasible Reidemeister moves---see~\cref{l:complexity_moves}
  (a more naive analysis yields $O(n^3)$). Each of them can be performed in time
  $O(n \log n)$---see~\cref{l:complexity_moves} again. In addition, any diagram
  obtained after performing a Reidemeister move on $D$ has at most $n+2$
  crossings. This gives that all feasible sequences of at most $k$ Reidemeister
  moves starting with $D$ can be constructed and performed in time
  $O(k^{2k+2}\log k)$. Indeed, each such sequence can be constructed and
  performed in time $O(k^2 \log k)$ and there are at most $O(k^{2k})$ such
  sequences. If any of them yields a diagram without crossings, then we answer
  \sctt{Yes}, otherwise we answer \sctt{No}.

  Altogether, \sctt{Unknotting via number of moves} can be solved in time
  $O(k^{2k+2} \log k) + O(|D|)$, therefore, it belongs to FPT. (Of course,
  in the sketch above, we didn't really need to know the concrete bounds on the
  number of feasible Reidemeister moves and on the time for construction of all
  feasible sequences. Any computable function would have sufficed.)
\end{proof}

\paragraph{Parameterization via defect.}
As we see from the sketch above, 
parameterization in the number of Reidemeister
moves has the obvious disadvantage that once we fix $k$, the problem becomes trivial for
arbitrarily large inputs (they are obviously a \sctt{NO} instance). We also see that
if we have a diagram $D$ with $n$ crossings and want to minimize the number of
Reidemeister moves to untangle $D$, presumably the most efficient way is to
remove two crossings in each step, thus requiring at least $n/2$ steps. This
motivates the following definition of the notion of defect.

Given a diagram $D$, by an \emph{untangling} of $D$ we mean a sequence
$\D = (D_0, \dots, D_{\ell})$ such that $D = D_0$; $D_{i+1}$ is obtained from
$D_i$ by a Reidemeister move; and $D_\ell = U$ is the diagram with no
crossings. Then we define the \emph{defect} of an untangling $\D$ as above as 
\[
  \defe(\D) := 2\ell - n
\]
where $n$ is the number of crossings in $D$. Note that $\ell$ is just the
number of Reidemeister moves in the untangling. It is easy to see that
$\defe(\D) \geq 0$ and $\defe(\D) = 0$ if and only if all moves in the
untangling are $\IImin$ moves. Therefore, $\defe(\D)$ in some sense measures
the number of `extra' moves
in the untangling beyond the trivial bound. (Perhaps, a more accurate
expression for this interpretation would be $\ell - n/2 = \frac12\defe(\D)$ but
this is a minor detail and it is more convenient to work with integers.) In
addition, it is possible to get diagrams with arbitrarily large number of
crossings but with defect bounded by a constant (even for defect $0$ this is
possible). The defect also plays a key role in the reduction
in~\cite{demesmay-rieck-sedgwick-tancer21} which suggests that the hardness of
the untangling really depends on the defect.

As we have seen above, asking the question whether a diagram can be untangled
with defect at most $k$ is same as asking if it can be untangled in $k/2$ moves
above the trivial, but tight lower bound of $n/2$. This fits perfectly in the
framework of above guarantee parameterization, which was introduced by Mahajan
and Raman~\cite{MahajanR99} for \sctt{Max-Sat} and \sctt{Max-Cut} problems. In
this framework, when there is a trivial lower bound for the solution in terms
of the size of the input, parameterizing by solution size trivially gives an
FPT algorithm by either giving a trivial answer if the input is large, or
bounding size of the input by a function of the solution size. Hence, for those
problems, it makes more sense to parameterize above a tight lower bound. The
paradigm of above guarantee parameterization has been very successful in the
field of parameterized complexity and many results have been
obtained~\cite{AlonGKSY11,CrowstonFGJKRRTY14,CrowstonJMPRS13,GutinKLM11,GutinIMY12,GutinP16,MahajanRS09}.

For these reasons, we find the defect to be a more natural parameter
than the number of Reidemeister moves. Therefore, we consider the following
problem.

\begin{problem*}[\sctt{Unknotting via defect}] \

  \begin{tabular}{ll}
  \sctt{Input} & A diagram $D$ of a knot.\\
  \sctt{Parameter} & $k$.\\
  \sctt{Question} & Can $D$ be untangled with defect at most $k$?\\
\end{tabular}
  \end{problem*}

\begin{theorem}
\label{t:main}
  The problem \sctt{Unknotting via defect} is W[P]-complete.
\end{theorem}

The proof of Theorem~\ref{t:main} consists of two main steps:
W[P]-membership and W[P]-hardness. Both of them are non-trivial. 

For W[P]-membership, roughly speaking, the
idea is to guess a small enough set of \emph{special} crossings on which we
perform all possible Reidemeister moves, while we remove other crossings in a greedy
fashion. In order to succeed with such an approach we will need some powerful
and flexible enough lemmas on changing the ordering of Reidemeister moves in some
untangling by swapping them. 
In Subsection~\ref{ss:algo_membership} we provide an algorithm for
W[P]-membership but we do not prove yet that it works. Then we start filling in
the details. In Section~\ref{s:diagrams}, we explain how we represent
diagrams as an input for the algorithm. In Section~\ref{s:r_moves} we explain
how we implement Reidemeister moves with respect to the aforementioned
representation of diagrams and also introduce some combinatorial notation
for them which will be useful in further considerations. We also explain the idea of
special crossings in detail and we provide a bound on the number of required
special crossings with respect to the defect. In Section~\ref{s:algo_correct}
we prove the correctness of our algorithm modulo a result on rearranging the
moves, whose proof is postponed to Section~\ref{s:rearrange}.

For W[P]-hardness, we combine some techniques that were quite recently used in showing
parameterized hardness of problems in computational
topology~\cite{burton-lewiner-paixao-spreer16, bauer-rathod-spreer19}, along with the
tools in~\cite{demesmay-rieck-sedgwick-tancer21} for lower bounding the defect.
Namely, we use a reduction from the \sctt{Minimum axiom set} problem, which
proved to be useful in~\cite{burton-lewiner-paixao-spreer16,
  bauer-rathod-spreer19}. Roughly speaking, from an instance $I$ of the
\sctt{Minimum axiom set} problem (which we have not defined yet), we need to build
a diagram which has a small defect if and only if $I$ admits a small set of
axioms. For the “if" part, we use properties of Brunnian rings to achieve our
goal. For the “only if" part, we need to lower bound the defect of our
construction. We use the tools from~\cite{demesmay-rieck-sedgwick-tancer21} to
show that the defect (of some subinstances) is at least $1$. Then we use very
simple but powerful boosting lemma (Lemma~\ref{l:boosting}) that shows that the
defect is actually high. We give the details of
the hardness proof in Section~\ref{s:hardness}. 

We conclude this part of introduction by proving a lemma on the properties of the
defect which we will use soon after.

Given a Reidemeister move $m$, let us define the \emph{weight} of this move $w(m)$
as 
\[
w(m) =  
\left\{
        \begin{array}{ll}
                        0  & \mbox{if $m$ is a $\IImin$ move};  \\
                        1  & \mbox{if $m$ is a $\Imin$ move};  \\
                        2  & \mbox{if $m$ is a $\III$ move};  \\
                        3  & \mbox{if $m$ is a $\Ipl$ move};  \\
                        4  & \mbox{if $m$ is a $\IIpl$ move}.  \\
         \end{array}
\right.
\]
\begin{lemma}
\label{l:defect}
Let $\D$ be an untangling of a diagram $D$. Then $\defe(\D)$ equals to the sum
of the weights of the Reidemeister moves in $\D$.
\end{lemma}
Because the proof of the lemma is short, we present it immediately.
\begin{proof}
We will give a proof by induction in the number of moves used in $\D$. There is
nothing to prove if there is no move in $\D$ (in particular, $D$ is already
untangled in this case).

Otherwise, let $\D = (D_0, \dots, D_{\ell})$ with $\ell \geq 1$. Let $\D' = (D_1,
\dots, D_{\ell})$ and let $m_1$ be the Reidemeister move transforming $D_0$ to
$D_1$. Let $\ell' := \ell -1$ be the number of moves in $\D'$; $n$ be the number
of crossings in $D_0$ and $n'$ be the number of crossings in $D_1$. We get
\[
\defe(\D) - \defe(\D') = 2\ell - n - (2 \ell' - n') = 2 + (n' - n) = w(m_1)
\]
where the last equality follows from a very simple case analysis depending on
the type of $m_1$. Because $\defe(\D')$ is the sum of the weights of the
remaining moves (by induction), we get the desired conclusion.
\end{proof}

\subsection{A brief overview of the parameterized complexity classes}
\label{ss:parameterized_complexity}
Here we briefly overview the notions from parameterized complexity needed for
this paper. For further background, we refer the reader to monographs~\cite{flum-grohe06, fund_param_complexity,cygan+15}. A 
\emph{parameterized problem} 
is a language $L
\subseteq \Sigma^*\times \mathbb{N}$, where $\Sigma^*$ is the set of strings
over a finite alphabet $\Sigma$ and the input strings are of the form $(x, k)$.
Here the integer $k$ is called \emph{the parameter}. Note that
parameterized problems are sometimes equivalently defined as a subset $P$ of
$\Sigma^*$ and a polynomial-time computable function $\kappa$ called parameter.
In subsequent parts of the paper, given an input for a parameterized
problem, $n$ denotes the size of the input and $k$ denotes the value of the
parameter on this input.

\paragraph{The classes FPT and XP.} A parameterized problem belongs to the class FPT (fixed parameter tractable) if
it can be solved by a deterministic Turing machine in time $f(k) \cdot n^c$
where $c > 0$ is some constant and $f(k)$ is some computable function of $k$.
In other words, if we fix $k$, then the problem can be solved in polynomial time
while the degree of the polynomial does not depend on $k$. This is, of course,
sometimes not achievable and there is a wider class XP of problems, that can be
solved in time $O(n^{f(k)})$ by a deterministic Turing machine. The problems in
XP are still polynomial time solvable for fixed $k$, however, at the cost that
the degree of the polynomial depends on $k$.

\paragraph{The class W[P].} Somewhere in between FPT and XP there is an
interesting class W[P]. A parameterized problem belongs to the class W[P] if it
can be solved by a nondeterministic Turing machine in time $h(k)\cdot n^c$ provided
that this machine makes only $O(f(k) \log n)$ non-deterministic choices where
$f(k), h(k)$ are computable functions and $c > 0$ is some constant. 
Given an algorithm for some computational problem $\Pi$, we say that this
algorithm is a \emph{W[P]-algorithm} if it is represented by a Turing machine satisfying
the conditions above.

It is obvious that a problem in FPT belongs to W[P]---we use the same algorithm
with $0$ non-deterministic choices. On the other hand, given a W[P]-algorithm,
it can be converted to a deterministic algorithm by trying all option for each
non-deterministic choice. The running time of the new algorithm is $h(k)\cdot
n^c \cdot (c')^{O(f(k) \log n)}$ for some $c' > 0$ which can be easily
manipulated into a formula showing XP-membership.

\paragraph{FPT-reduction.} Given two parameterized problems $\Pi$ and $\Pi'$, we
say that $\Pi$ reduces to $\Pi'$ via an \emph{FPT-reduction} if there exist 
computable functions $f \colon \Sigma^* \times \mathbb{N} \to \Sigma^*$
and $g \colon \Sigma^* \times \mathbb{N} \to \mathbb{N}$ such that
\begin{itemize}
  \item  $(x,k) \in \Pi$ if and only if $(f(x,k), g(x,k)) \in \Pi'$ for every
    $(x,k) \in \Sigma^* \times \mathbb{N}$;
  \item  $g(x,k) \leq g'(k)$ for some computable function $g'$; and
\item there exist a computable function $h$ and a fixed constant $c > 0$ such
  that for all input string, $f(x,k)$ can be computed by a deterministic
    Turing-machine in $O(h(k)n^c)$ steps.
\end{itemize}
Let us remark that here we use a definition of reduction consistent
with~\cite[Definition~2.1]{flum-grohe06} or~\cite[Definition~13.1]{cygan+15}. 
On the other hand, some authors,
e.g.~\cite[Definition~20.2.1]{fund_param_complexity}, impose a stronger
requirement that would correspond to asking $g(x,k) = g'(k)$ in the second item
of our definition. Thus our hardness claims are with respect to the reduction
we define here.\footnote{This is not a principal problem. With some extra
effort it would be possible to provide a reduction consistent
with~\cite{fund_param_complexity}. But this would prolong the already technically
complicated proof.}

The classes FPT, W[P] and XP are closed under FPT-reductions. A problem $\Pi$ is said
to be \emph{C-hard} where C is a parameterized complexity class, if all problems in C
can be FPT-reduced to $\Pi$. Moreover, if $\Pi \in \hbox{C}$, we say that $\Pi$
is \emph{C-complete}.

\subsection{An algorithm for W[P]-membership}
\label{ss:algo_membership}
In this subsection we provide the algorithm used to prove W[P]-membership
in~\cref{t:main}. 

\paragraph{Brute force algorithm.} First, let us however look at a brute force
algorithm for \sctt{Unknotting via defect} which does not give W[P]-membership.
Spelling it out will be useful for explaining the next steps. We will exhibit
this algorithm as a non-deterministic algorithm, which will be useful for
comparison later on.
In the
algorithm, $D$ is a diagram, and $k$ is an integer, not necessarily positive.
Also, given a diagram $D$ and a feasible Reidemeister move $m$, then $D(m)$
denotes the diagram obtained from $D$ after performing $m$.

\medskip
\sctt{BruteForce$(D, k)$:}

\begin{enumerate}[1.]
  \item If $k < 0$, then output \sctt{No.}  If $D = U$ is a diagram
    without crossings and $k \geq 0$, then output \sctt{Yes}. In all
    other cases continue to the next step.
\item (Non-deterministic step.) Enumerate all possible Reidemeister
  moves $m_1, \dots, m_t$ in $D$ up to isotopy. Make a `guess' which
  move $m_i$ is the first to perform. Then, for such $m_i$, run
  \sctt{BruteForce$(D(m_i), k - w(m_i))$}.
\end{enumerate}

Therefore, altogether, the algorithm outputs \sctt{Yes}, if there is a sequence
of guesses in step 2 which eventually yields \sctt{Yes} in step 1.

It can be easily shown that the algorithm terminates because whenever step 2 is
performed, either $k - w(m_i) < k$, or $m_i$ is a $\IImin$ move, $k - w(m_i) =
k$, but $D(m_i)$ has fewer crossings than $D$.

It can be also easily shown by induction that this algorithm provides a correct
answer due to~\cref{l:defect}. Indeed, step 1 clearly provides a correct
answer, and regarding step 2, if $D(m_i)$ can be untangled with defect at most
$k - w(m_i)$, then~\cref{l:defect} shows that $D$ can be untangled with
defect at most $k$. Because, this way we try all possible sequences of
Reidemeister moves, the algorithm outputs \sctt{Yes} if and only if $D$ can be untangled with
defect at most $k$.

On the other hand, this algorithm (unsurprisingly) does not provide W[P]-membership as it performs too many non-deterministic guesses.

\paragraph{Naive greedy algorithm.} In order to fix the problem with the
previous algorithm, we want to reduce the number of non-deteministic steps. It
turns out that the problematic non-deterministic steps in the previous
algorithm are those where $k$ does not decrease. (Because the other steps
appear at most $k$ times.) Therefore, we want to avoid non-deterministic
steps where we perform a $\IImin$ move. The naive way is to perform such steps
greedily and hope that if $D$ untangles with defect at most $k$, there is also such a
`greedy' untangling (and therefore, we do not have to search through all
possible sequences of Reidemeister moves). This is close to be true but it does
not really work in this naive way. Anyway, we spell this naive algorithm
explicitly, so that we can easily upgrade it in the next step, though it does
not always provide the correct answer to \sctt{Unknotting via defect}.

\medskip

\sctt{NaiveGreedy$(D, k)$:}

\begin{enumerate}[1.]
\item If $k < 0$, then output \sctt{No.}
If $D = U$ is a diagram without crossings and $k \geq 0$, then output \sctt{Yes}. In all other cases continue to the next step.
\item If there is a feasible Reidemeister $\IImin$ move $m$, run
  \sctt{NaiveGreedy$(D(m), k)$} otherwise continue to the next step.
\item (Non-deterministic step.) If there is no feasible Reidemeister $\IImin$
  move, enumerate all possible Reidemeister moves $m_1, \dots, m_t$ in $D$ up
    to isotopy. Make a `guess' which $m_i$ is the first move to perform and run \sctt{NaiveGreedy$(D(m_i), k - w(m_i))$}. 
\end{enumerate}

The algorithm must terminate from the same reason why \sctt{BruteForce} terminates.

It can be shown that this algorithm is a W[P]-algorithm. However, we do not do this
here in detail as this is not our final algorithm. The key is that the step 3 is
performed at most $(k+1)$-times.  

If the algorithm outputs \sctt{Yes}, then this is clearly correct answer from
the same reason as in the case of \sctt{BruteForce}. However, as we hinted
earlier, outputting \sctt{No} need not be a correct answer to  
\sctt{Unknotting via defect}.
Indeed, there are known examples of diagrams of an unknot when untangling requires
performing a $\IIpl$ move, see for example \cite{kauffman12}.  With
\sctt{NaiveGreedy}, we would presumably undo such $\IIpl$ move immediately in
the next step, thus we would not find any untangling using the $\IIpl$ move.  We
have to upgrade the algorithm a little bit to avoid this problem (and a few
other similar problems).

\paragraph{Special greedy algorithm.}

The way to fix the problem above will be to guess in advance a certain subset
$S$ of so called special crossings. This set will be updated in each
non-deterministic step described above so that the newly introduced crossings will
become special as well. Then, in the greedy steps we will allow to perform only
those $\IImin$ moves which avoid $S$. (In particular, this means that a $\IIpl$
move cannot be undone by a $\IImin$ move avoiding $S$ in the next step.) It
will turn out that we may also need to perform $\IImin$ moves on $S$ but there
will not be too many of them, thus such moves can be considered in the
non-deterministic steps.

For description of the algorithm, we introduce the following notation. Let $D$
be a diagram of a knot and $S$ be a subset of the crossings of $D$; we will
refer to crossings in $S$ as \emph{special} crossings. Then we say that a
feasible Reidemeister $\IImin$ move $m$ is \emph{greedy} (with respect to
$S$), if it avoids $S$; that is, the crossings removed by $m$ do not belong to
$S$. On the other hand, a feasible Reidemeister move $m$ is \emph{special} (with
respect to $S$) if it is 
\begin{itemize}
\item a $\IImin$ move removing two crossings in $S$; or
\item a $\Imin$ move removing one crossings in $S$; or
\item a $\III$ move such that all three crossings affected by $m$ are special;
  or
\item a $\Ipl$ move performed on an edge with both of its endpoints in $S$; or
\item a $\IIpl$ move performed on edges with all of their endpoints in $S$.
\end{itemize}

Given a move $m$ in $D$, special or greedy with respect to $S$, by $S(m)$ we denote the
following set of crossings in $D(m)$. 
\begin{itemize}
\item If $m$ is a greedy $\IImin$ move or if $m$ is a $\III$ move, then $S(m) = S$ (under the convention that the
  three crossings affected by $m$ persist in $D(m)$).
\item If $m$ is a special $\IImin$ move or a $\Imin$ move, then $S(m)$ is obtained from $S$
  by removing the crossings removed by $m$.
\item If $m$ is a $\Ipl$ or a $\IIpl$ move, then $S(m)$ is obtained from $S$
  by adding the crossings introduced by $m$.
\end{itemize}

Now, we can describe the algorithm.

\medskip

\sctt{Special\linebreak[0]{}Greedy$(D, k)$:}
\begin{enumerate}[0.]
  \item (Non-deterministic step.) Guess a set $S$ of at most $3k$ crossings in $D$. Then run
     \sctt{Special\linebreak[0]{}Greedy$(D, S, k)$}.
\end{enumerate}

\sctt{Special\linebreak[0]{}Greedy$(D, S, k)$:}

\begin{enumerate}[1.]
\item If $k < 0$, then output \sctt{No.}
If $D = U$ is a diagram without crossings and $k \geq 0$, then output \sctt{Yes}. In all other cases continue to the next step.
\item If there is a feasible greedy Reidemeister $\IImin$ move $m$ with respect
  to $S$, run
    \sctt{Special\linebreak[0]{}Greedy$(D(m), S(m), k)$} otherwise continue to the next step.
\item (Non-deterministic step.) If there is no feasible greedy Reidemeister $\IImin$ move with respect
    to $S$, enumerate all possible special Reidemeister moves $m_1, \dots, m_t$
    in $D$ with respect to $S$ up to isotopy. If there is no such move, that
    is, if $t = 0$, then output \sctt{No}. Otherwise, make a guess which $m_i$ is
    performed first and run
    \sctt{Special\linebreak[0]{}Greedy$(D(m_i), S(m_i), k - w(m_i))$}.
\end{enumerate}

The bulk of the proof of W[P]-membership in~\cref{t:main} will be to show that
the algorithm \sctt{Special\linebreak[0]{}Greedy$(D, k)$} provides a correct answer to
\sctt{Unknotting via defect}. Of course, if the algorithm outpus
\sctt{Yes}, then this is the correct answer by similar arguments for the
previous two algorithms. Indeed, \sctt{Yes} answer corresponds to a sequence of
Reidemeister moves performed in step 2 or guessed in step 3 (no matter how we
guessed $S$ in step 0 and the role of $S$ in the intermediate steps of the run
of the algorithm is not important if we arrived at \sctt{Yes}). The defect of
the untangling given by this sequence of moves is at most $k$
by~\cref{l:defect}. On the other hand, we also need to show that if $D$
untanlges with defect at most $k$, then we can guess some such untangling while
running the algorithm. 
This is done in~\cref{s:algo_correct}.

\section{Diagrams and their combinatorial description}
\label{s:diagrams}

\paragraph{Arcs, strands, edges and faces.}
By an \emph{arc} $\alpha$ in the diagram $D$ we mean a set $D(A)$ where
$A$ is an arc in $S^1$ (i.e., a subset $A$ of $S^1$ homeomorphic to the closed interval). This definition is slightly non-standard but it will be very useful
for us. The \emph{endpoints} of $\alpha$ are the points $D(\partial A)$;
usually, these are two points but they may coincide. The \emph{interior} of $\alpha$
is $D(\interior A)$ where $\interior A$ stands for the interior of
$A$. (Note that $\alpha$ determines $A$ uniquely, thus the interior
and the endpoints are always well defined.) 

Consider a crossing $x$ in diagram $D$ and take a small disk
surrounding $x$. Then there are four short arcs in $D$ with one endpoint $x$
and the other endpoint on the boundary of the disk. We call these arcs
\emph{strands}. Each strand is marked as an overpass or underpass.

We also define an \emph{edge} as an arc $\varepsilon$ in $D$
between two crossings $x$ and $y$ (which may coincide) such that the interior
of $\varepsilon$ does not contain any crossing. In an exceptional case when $D
= U$ has
no crossings, with slight abuse of the terminology, we regard $D$ itself as an
edge. Finally, given arbitrary diagram $D \colon S^1 \to \R^2$, the connected components of $\R^2
\setminus D(S^1)$ are called \emph{faces} of $D$. 

\paragraph{Combinatorial description of diagrams.}
For purposes of our algorithms, we will need a purely combinatorial description
of diagrams up to isotopy. Such a description can be given as graphs embedded in
the plane as follows. By $U$, let us denote the unique diagram of the unknot
with no crossing. We provide a representation of a $D \neq U$ with at least one
crossing.  First, we give the set of vertices $V(D)$ of a diagram $D$ which
represent crossings in $D$.

Next, we represent edges of $D$ so that each edge $\varepsilon$ will be
represented as a $4$-tuple $[v_1, v_2, n_1, n_2]$ where $v_1, v_2 \in V(D)$ and
$n_1, n_2 \in [4]$. This $4$-tuple represents an edge connecting $v_1$ and $v_2$
whereas the number $n_i$ serves to mark the strand of $v_i$ used by
$\varepsilon$. We may number the strands of $v_i$ so that they appear
consecutively in clockwise order and we may assume that the strands $1$ and $3$
are overpasses while $2$ and $4$ are underpasses; see~\cref{f:lbl}, left. This
is a generic representation for graphs embedded in surfaces, called the rotation
system, see~\cite{rotation}.

\begin{figure}
\begin{center}
\includegraphics{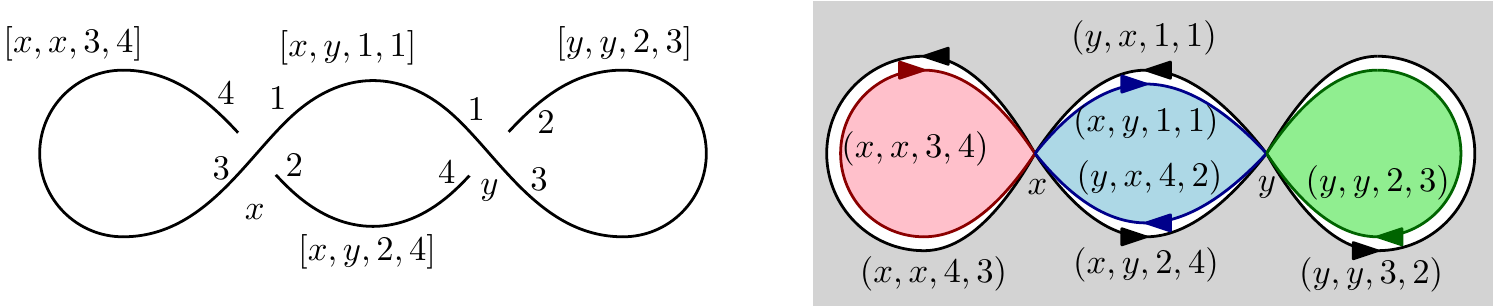}
  \caption{Left: Representation of the diagram in the figure is given by setting
  $V(D) = \{x,y\}$, and adding edges $[x,x,3,4]$,
  $[x,y,1,1]$, $[y,y,2,3]$ and
  $[x,y,2,4]$. Right: Each edge corresponds to a pair of edges, in which case, the directed
  edges determine faces. In our case, we have bounded faces
  $\{(x,x,3,4)\}$, $\{(x,y,1,1),(y,x,4,2)\}$, $\{(y,y,2,3)\}$ and unbounded face
  $\{(x,x,4,3),(y,x,1,1), (y,y,3,2),(x,y,2,4)\}$.}
\label{f:lbl}
\end{center}
\end{figure}

We remark that we consider the edges unoriented, that is, $[v_1, v_2, n_1, n_2]
= [v_2, v_1, n_2, n_1]$. However, for purpose of representing faces, it is
convenient to follow the idea of doubly-connected edge list (see~\cite{dcel}) so
that each edge $[v_1, v_2, n_1, n_2]$ is represented as a pair of oriented edges
$(v_1, v_2, n_1, n_2)$ and $(v_2, v_1, n_2, n_1)$ where $(v_1, v_2, n_1, n_2)$
stands for the directed edge leaving $v_1$ via strand marked $n_1$ and entering
$v_2$ via strand $n_2$. Each directed edge uniquely determines a face on the `right
hand side' and if two oriented edges $(v_1, v_2, n_1, n_2)$ and $(v_3, v_4, n_3,
n_4)$ satisfy $v_2 = v_3$ and $n_2 = n_3 + 1 \pmod 4$, then they determine the
same face; see~\cref{f:ee}. If we declare two such oriented edges as equivalent,
then faces of $D$ can be represented as the classes of the (inclusion-wise)
smallest equivalence containing such pairs of equivalent oriented edges;
see~\cref{f:lbl}, right. In order to complete our combinatorial description of a
diagram, it remains to mark one face as the \emph{outer face}---the unique
unbounded face.  (Of course, not every choice of $V(D)$ and $4$-tuples $[v_1,
  v_2, n_1, n_2]$ as above yields a valid representation of a diagram. In our
considerations, we will assume that we really have a valid representation.)

\begin{figure}
\begin{center}
\includegraphics{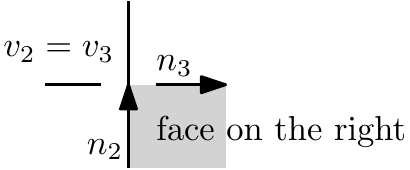}
 \caption{Directed edges pointing to the same face on the right.}
  \label{f:ee}
\end{center}
\end{figure}

We remark that if we have a diagram $D$ with $n$ crossings, then this way
we can reach an encoding of $D$ of size $O(n \log n)$ if we interpret crossings
as integers in binary. For purposes of the implementation of the algorithm, we
always assume that we have an encoding of $D$ of such size. However, for
purposes of the analysis of the algorithm, $V(D)$ will be just some abstract set.

\section{Reidemeister moves}
\label{s:r_moves}

\subsection{Combinatorial description of Reidemeister moves}
\label{ss:combinatorial_moves}

Given a diagram represented combinatorially as explained earlier, we want to provide a combinatorial
description of all feasible Reidemeister moves in $D$. When doing so, we will
have the following aims.

\begin{itemize}
 \item We want to describe each move with as less combinatorial data as possible that is
   required to determine a given move uniquely. For example, it is not hard to
    see that a $\IImin$ move in a diagram $D$ is uniquely determined by the two
    crossings removed by the move. Thus it makes sense to identify $\IImin$
    moves with pairs of crossings of $D$. The cost of this approach is that
    not every pair of crossings in a given diagram yields a feasible $\IImin$
    move. Thus we have to distinguish which pairs do.
 \item The previous approach allows us to say, in some cases, that two moves in
   different diagrams are (combinatorially) the same. For example, if we have
    diagrams $D$ and $D'$ both of them containing crossings $x$ and $y$ and in
    both diagrams there is a feasible $\IImin$ move removing $x$ and $y$. Then
    we can regard this move as the same move in both cases. This will be useful
    in stating various rearrangement lemmas on the ordering of the moves in a sequence
    of moves untangling a given diagram.
 \item We want to be able to enumerate all possible Reidemeister moves in a
   given diagram in polynomial time. This can be regarded as obvious but we
    want to include the details for completeness. Thus, for this aim, we can use
    the combinatorial description in the first item above and check feasibility in
    each case.
 \item For running and analysing our algorithm, it will be also useful if a
   feasible move $m$ in a diagram $D$ uniquely determines the combinatorial
    description of the resulting diagram $D(m)$ after performing the move.
\end{itemize}

Our description of the moves will in particular yield a proof of the following
lemma.

\begin{lemma}
  \label{l:complexity_moves}
  Let $D$ be a diagram of a knot with $n$ crossings (and with encoding of
  size $O(n \log n)$).
  \begin{enumerate}
  \item There are $O(n^2)$ feasible Reidemeister moves and all of them can be
    enumerated in time $O(n^2 \log n)$.
  \item Given a feasible move $m$ in $D$, the diagram $D(m)$ can be
    constructed in time $O(n \log n)$.
  \end{enumerate}
\end{lemma}

We prove the lemma simultaneously with our description of the moves.

We start with $\Imin$, $\IImin$ and $\III$ moves:

\paragraph{$\Imin$, $\IImin$ and $\III$ moves: Definitions.}
By a simple case analysis, each feasible Reidemeister $\Imin$, $\IImin$ or
$\III$ move $m$ in $D$ is uniquely determined by the set of crossings $X$ in
$D$ affected by the move. In the case of $\Imin$ move there is seemingly
an exceptional case of a diagram with a single crossing and two loops attached
to the crossing; see~\cref{f:dl}. However, the resulting diagram after
performing the $\Imin$ move is $U$ independently of the choice
of the loop which is removed. Thus we do not really distinguish these two cases.
Therefore, we can define a \emph{$\Imin$ move} in $D$ as a $1$-element subset
of $V(D)$, a \emph{$\IImin$ move} as a $2$-element subset of $V(D)$, and a \emph{$\III$
move} as a $3$-element subset of $V(D)$. In this terminology, not every
$\Imin$, $\IImin$ or $\III$ move yields feasible Reidemeister move, thus we
always add the adjective \emph{feasible} if it indeed yields.

 We find it often convenient to emphasize the type of the move in the notation (if we
 know it), thus
 we also write $\{x\}_{\Imin}$ instead of $\{x\}$ for a $\Imin$ move removing
 $\{x\}$, and similarly we use $\{x,y\}_{\IImin}$ and $\{x,y,z\}_{\III}$.

\begin{figure}
\begin{center}
  \includegraphics{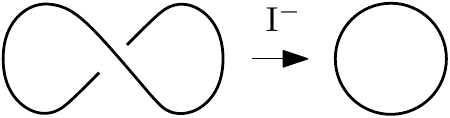}
\end{center}
\caption{The resulting diagram after applying $\Imin$ move does not depend on
  the choice of the loop which is removed.}
  \label{f:dl}
\end{figure}

\paragraph{$\Imin$, $\IImin$ and $\III$ moves: Checking feasibility.}
For checking whether a given $\Imin$, $\IImin$, or a $\III$ move is feasible in
a given diagram, we can use the following approach.

In case of a $\Imin$ move $\{x\}_{\Imin}$, it is sufficient to check whether
$D$ contains an $[x,x,n_1,n_2]$ edge where $n_1, n_2 \in \{1, 2, 3,4\}$ such
that one of the outer edges $(x,x,n_1, n_2)$ or $(x,x,n_2, n_1)$ forms a face
different from the outer face of $D$. Therefore, if we want to find all
feasible $\Imin$ moves, then it is sufficient to find all faces formed by a
single directed edge and check whether they are distinct from the outer face. This can
be done in time $O(n \log n)$ where $n$ is the number of
crossings due to encoding the crossings as integers in binary.

In case of a $\IImin$ move $\{x,y\}_{\IImin}$, we need a face, distinct from the
outer face, formed by two
directed edges $(x,y,n_1, n_2)$ and $(y,x,n_3, n_4)$. In addition we need that
one of the edges $[x,y,n_1,n_2]$ and $[y,x, n_3,n_4]$ enters both $x$ and $y$
as overpass and the other one as underpass. This can be determined from $n_1$,
$n_2$, $n_3$ and $n_4$. Therefore, if we want to enumerate all feasible
$\IImin$ moves, we can do it in time $O(n \log n)$ where $n$ is the number of
crossings due to encoding the crossings as integers in binary.

In case of a $\III$ move $\{x,y,z\}_{\III}$, we need a face, distinct from the
outer face, formed by three directed edges $(x,y,n_1, n_2)$, $(y,z,n_3, n_4)$,
$(z,x, n_5, n_6)$ up to a permutation of $\{x, y, z\}$. We also need that one
of the edges $[x,y,n_1, n_2]$, $[y,z,n_3, n_4]$, or $[z,x, n_5, n_6]$ enters
both its endpoints as overpass, the second one enters one endpoint as overpass
and the other one as underpass, and the third one enters both its endpoints as
underpass. This can be determined by providing a finite list of allowed options
for $(n_1, n_2, n_3, n_4, n_5, n_6)$. Altogether, we can again enumerate all
feasible $\III$ moves in time $O(n \log n)$ where $n$ is the number of 
crossings due to encoding the crossings as integers in binary.

\paragraph{$\Imin$, $\IImin$ and $\III$ moves: Construction of $D(m)$.}
Now, we briefly describe how to get $D(m)$ from the knowledge of $D$ and $m$.

  Let us start with a $\Imin$ move as in~\cref{f:rm_Imin_move}. Assume
  that it removes a crossing $x$ and that the strands around $x$ are numbered as
  in the figure. Assume also that the loose ends of the arc affected by the
  move head to crossings $x'$ and $x''$ via strands numbered $n'$ and $n''$.
  Then we get $D(m)$ from $D$ by removing $x$ and the edges $[x, x, n_2,n_3]$,
  $[x,x',n_4, n']$, and $[x, x'', n_1, n'']$ while adding an edge $[x', x'', n',
    n'']$.  (There is a single exception, though: If $x = x' = x''$, then $D(m)
  = U$ is the diagram without crossings.) We also need to update the list of
  faces, but this is straightforward. In the case in the figure, for the face
  of $D$ containing directed edges $(x'',x,n'', n_1)$ and $(x, x', n_4, n')$
  these two edges get replaced with the directed edge $(x'',x', n'',
  n')$. Similarly, for the face containing $(x', x, n', n_4), (x,x,n_3, n_2),
  (x,x'',n_1,n'')$ these edges get replaced with $(x',x'', n',
  n'')$. Altogether, we perform a constant number of operations. However, as we
  are removing some crossings, we may want to relabel the crossings, if we want
  to reach a representation of $D(m)$ of size $O(\widehat n \log \widehat n)$ where
  $\widehat n$ is the number of crossings of $D(m)$. Doing so naively (which is
  fully sufficient for our purposes), we may need time $O(n \log n)$ for
  this. But we do this relabelling only for the purpose of implementing the
  algorithm. For the analysis of the algorithm, we keep the names the same.

\begin{figure}
\begin{center}
  \includegraphics[page=3]{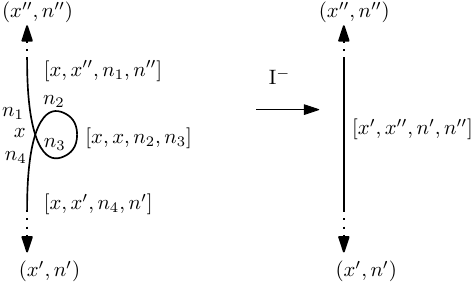}
\end{center}
  \caption{$D(m)$ coming from $D$ after a $\Imin$ move. Here we cover both
  cases whether the strands at $n_1$ and $n_3$ are underpasses or overpasses,
  thus we didn't make an explicit choice in the figure.}
  \label{f:rm_Imin_move}
\end{figure}

Moves of types $\IImin$ or $\III$ can be treated similarly, and we omit most of
the details, because we believe that the general approach can be understood from
the $\Imin$ example, while it is a bit tedious to list all possible cases. Thus
we only point out important differences. For a $\IImin$ move $m =
\{x,y\}_{\IImin}$ there are a few more exceptional cases how to rebuild the
diagram depending on whether there are loops at $x$ or at $y$. (Note that there
cannot be three edges connecting $x$ and $y$ if we have a diagram of a knot.)
For a $\III$~move $m = \{x,y,z\}_{\III}$ there are again more exceptional cases,
depending whether there loops at $x$, $y$, $z$ or additional edges connecting
two of the points among $x$, $y$, $z$. In addition, we want to keep the names of
the crossings during the move. Consider the move on~\cref{f:rm_xyz}; here
for simplicity of the figure, we label the strands only around $z$ on the left
figure and around $x$ on the right figure. We shift the labels of the
crossings as indicated in the figure. For example, $x$ is an intersection of
arcs $\alpha$ and $\beta$ in the left figure. These two arcs are shifted to
$\alpha'$ and $\beta'$ in the right figure and we label the intersection of
$\alpha'$ and $\beta'$ again by $x$. Regarding the edges when coming from $D$ to
$D(m)$, for example the edge $[z,z'',n_1,n_z'']$ is removed while it is replaced
with the edge $[x,z'',1,n_z'']$.

  \begin{figure}
\begin{center}
  \includegraphics[page=4]{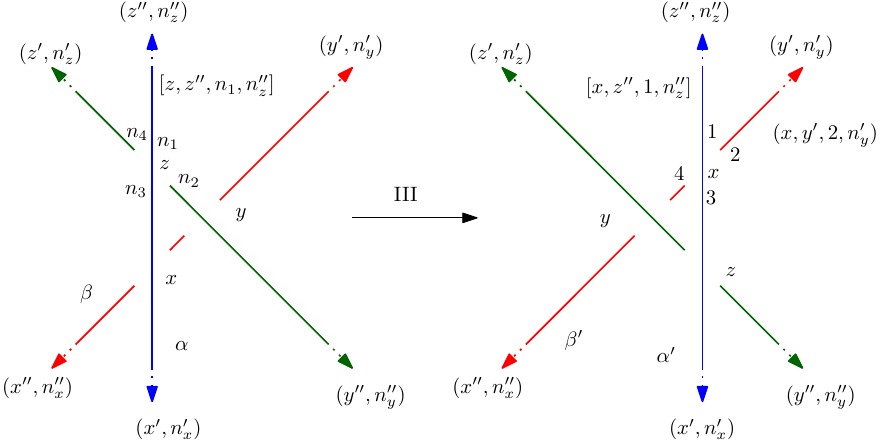}
\end{center}
  \caption{Crossings before and after a $3$-move.}
  \label{f:rm_xyz}
\end{figure}

\paragraph{$\Imin$, $\IImin$ and $\III$ moves: Circle.}
We will also need the following notion: for a feasible Reidemeister $\Imin$,
$\IImin$, or $\III$ move $m$ affecting the set of crossings $X$ on a diagram
with at least two crossings\footnote{This assumption is used to avoid the
example in~\cref{f:dl}.}, we will also denote by $c(m) \subseteq \R^2$ the
unique circle formed by images of edges in $D$ which contains the crossings of
$X$ but no other crossings; see~\cref{f:rm_circle}. (Here we regard a diagram and
edge as topological notions according to our earlier definitions). We will refer
to $c(m)$ as to the \emph{circle of $m$}.

\begin{figure}
\begin{center}
  \includegraphics[page=2]{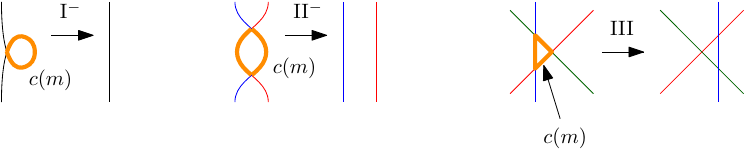}
\end{center}
  \caption{The circle $c(m)$ of a feasible $\Imin$, $\IImin$, or $\III$ move $m$. For
  emphasizing $c(m)$ the overpasses are not depicted. }
  \label{f:rm_circle}
\end{figure}
\paragraph{\Ipl moves: Definitions.}
We define a \emph{$\Ipl$ move} as a pair $(\epsvec, \sigma)$, where $\epsvec$ is
a directed edge of $D$ and $\sigma \in \{+,-\}$. Here $\epsvec$ determines the
undirected edge $\varepsilon$ on which the move is performed as well as the side
on which the loop is introduced. For being consistent with our definition of faces, we
require that the loop is introduced on the right hand-side of $\epsvec$. The
sign determines whether the crossing is positive or negative with respect to the
orientation of $\epsvec$; see~\cref{f:side_sign}. In an exceptional case
$D = U$, $\epsvec$ represents $U$ with positive or negative orientation.

\begin{figure}
\begin{center}
  \includegraphics{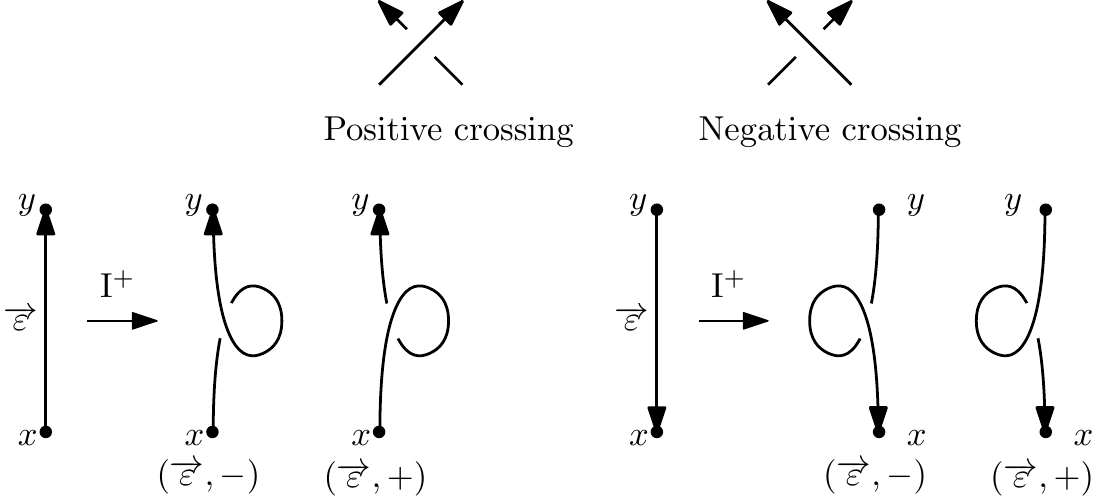}
\end{center}
\caption{Positive and negative crossings and four options for $\Ipl$ affecting
  a given edge.}
  \label{f:side_sign}
\end{figure}

\paragraph{\IIpl moves: Definitions.} 
Any feasible Reidemeister $\IIpl$ move can be performed in the following way.
We pick two points $a$, $b$ on $D$ different from crossings. We connect them
with a standard topological arc $\zeta$ inside one of the faces of $D$ (except
of the endpoints $a$, $b$ which do not belong to the face). We pull a finger
from $a$ towards $b$ along $\zeta$, thereby performing the $\IIpl$ move;
see~\cref{f:finger}. We may assume that the finger overpasses $b$ because an
underpass could be obtained, up to isotopy, by pulling a finger from $b$ to $a$
along $\zeta$.

\begin{figure}
\begin{center}
  \includegraphics{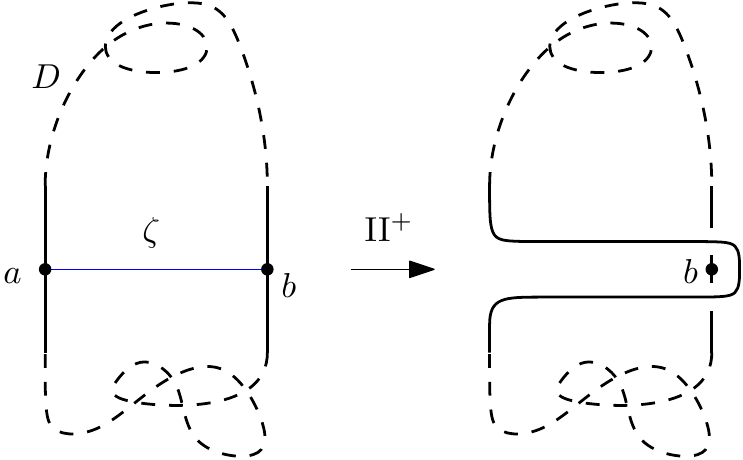}
\end{center}
\caption{A $\IIpl$ move along $\zeta$.}
  \label{f:finger}
\end{figure}

This means that any feasible $\IIpl$ move is uniquely
determined by the following data.

\begin{itemize}
  \item A pair of directed edges $(\epsvec, \phivec)$ on which the move is
    performed, subject to the condition that $\epsvec$ and $\phivec$ belong to
    the same face $f$. 
    Comparing with the description above, $a$ would belong to
    $\epsvec$, $b$ would belong to $\phivec$, and $\zeta$ would belong to the
    face containing $\epsvec$ and $\phivec$. (Here we mean the containment in
    the combinatorial sense where a face is a collection of directed edges.)
    We remark that we
    explicitly allow $\epsvec=\phivec$ (this corresponds to the case when
    $a$ and $b$ belong to the same edge). 
  \item A \emph{winding} $w \in \{+, -\}$. The winding is relevant only if $f$
    is the unbounded face. For fixed $a$ and $b$, the winding determines $\zeta$
    up to isotopy. Indeed, first consider the case that $f$ is bounded, then $f$
    is homeomorphic to an open disk (because $D$ is a diagram of a knot) and
    then $\zeta$ is unique up to isotopy. If $f$ is unbounded, then $f$ is
    homeomorphic to an open annulus.  This is easy to see in one point
    compatification of $\R^2$ by adding a point $\infty$ thereby obtaining the
    sphere $S^2$. Then $f \cup \{\infty\}$ is a (necessarily bounded) face in
    $S^2$ therefore an open disk.  Consequently, $f$ is an open disk minus a
    point that is an open annulus. Then we have two options, either $\zeta$
    passes from $a$ to $b$ in positive direction or negative direction;
    see~\cref{f:sign_zeta}. Then we set $w = +$ in the former case and $w = -$
    in the latter case.
\begin{figure}
\begin{center}
  \includegraphics{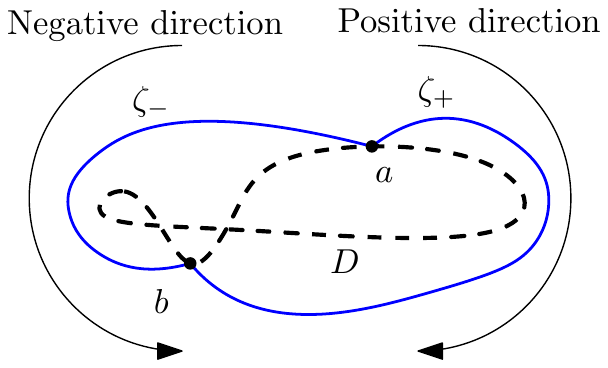}
\end{center}
\caption{If $\zeta$ belongs to the unbounded face, there are two possible
  isotopy classes for $\zeta$.}
  \label{f:sign_zeta}
\end{figure}
  \item An \emph{order} (of $a$ and $b$) denoted $o \in \{1, 2\}$. The order is
    relevant only if $\epsvec = \phivec$. Indeed, if $\epsvec \neq
    \varphi$, then any shift of $a$ along $\epsvec$ or any shift of $b$
    along $\phivec$ yields an isotopic $\II^+$, therefore the data above
    uniquely determine the move. (Note that if $\epsvec \neq  \phivec$, then
    also the undirected edges determined by $\epsvec$ and $\phivec$ are
    distinct because $\epsvec$ and $\phivec$ belong to the same face $f$.)
    If $\epsvec = \phivec$, it remains to
    distinguish whether $a$ precedes $b$ along $\epsvec$ or whether $b$
    precedes $a$. In the former case we set $o = 1$, in the latter case, we set
    $o = 2$.
\end{itemize}

Thus we define a \emph{$\IIpl$ move} as quadruple $(\epsvec, \phivec, w, o)$ as
above. If, for example, $w$ is irrelevant for given $\epsvec$, $\phivec$ and
$o$, then $(\epsvec, \phivec, +, o)$ and $(\epsvec, \phivec, -, o)$ yield the
same move and we regard them as the same combinatorial object.
The most complex case when $\zeta$ is in the unbounded face and $\epsvec =
\phivec$ is depicted at~\cref{f:zeta_complicated}.

\begin{figure}
\begin{center}
  \includegraphics{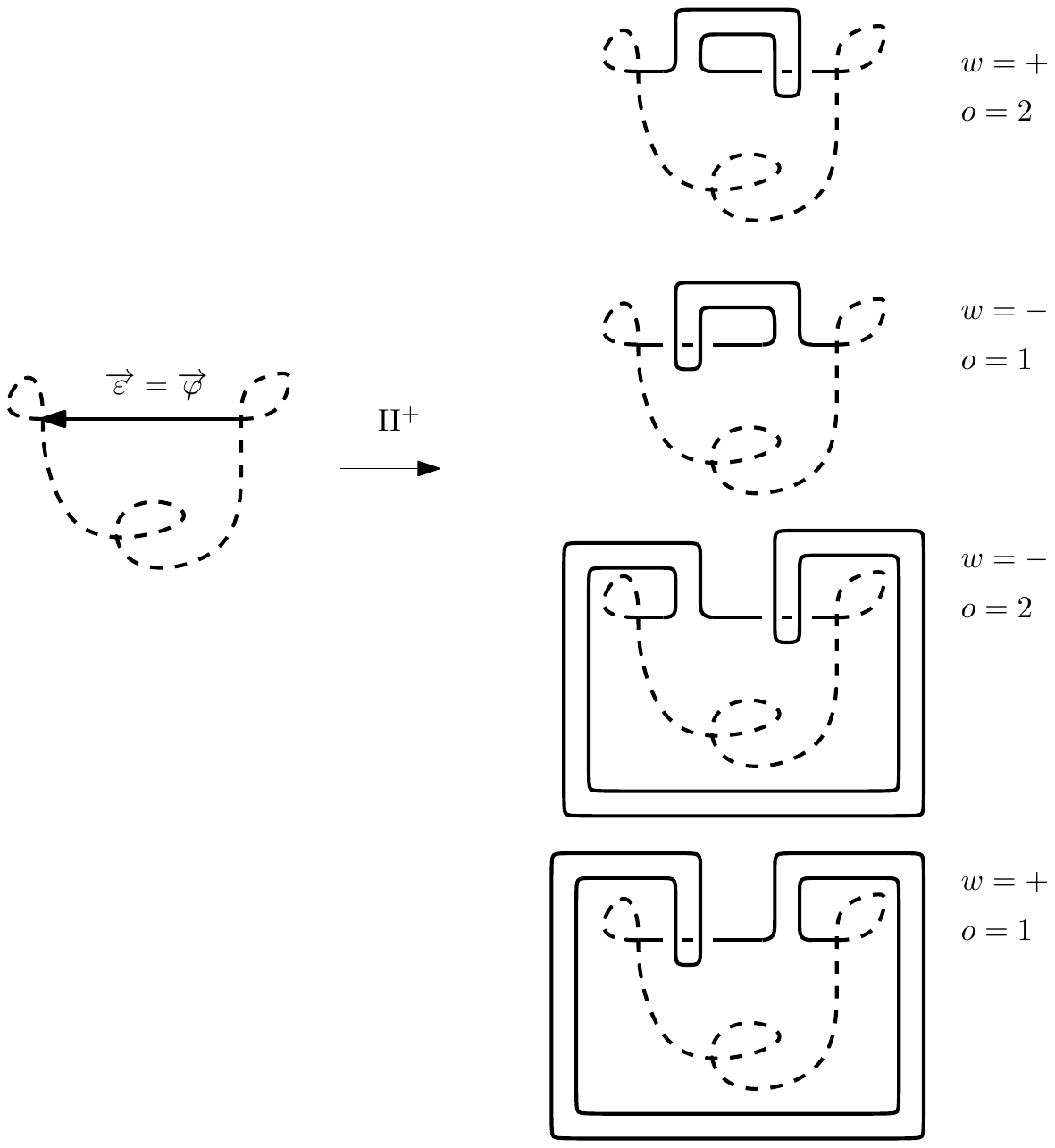}
\end{center}
\caption{Four possible outputs of an $\IIpl$ move depending on $w$ and $o$.}
  \label{f:zeta_complicated}
\end{figure}

In an exceptional case when $D = U$, $\epsvec$ and $\phivec$ represent
$U$ with a positive or negative orientation. This orientation has to be the
same for $\epsvec$ and $\phivec$.

\paragraph{\Ipl moves and \IIpl moves: Checking feasibility.} There is nothing
to do regarding checking feasibility for \Ipl and \IIpl moves. Every $\Ipl$ move
$(\epsvec, \sigma)$ is feasible. Similarly, every $\IIpl$ move $(\epsvec,
\phivec, w, o)$ is feasible subject to the condition $\epsvec$ and $\phivec$
belong to a same face.  
Thus the number of feasible \Ipl moves or \IIpl moves is at most
quadratic in $n$, and they can be enumerated in time $O(n^2 \log n)$ due to
encoding the labels of the crossings in binary.

\paragraph{\Ipl moves and \IIpl moves: Construction of $D(m)$.}
For a construction of $D(m)$ for a \Ipl move $m$, we can invert the process
explained for a \Imin move; see again~\cref{f:rm_Imin_move} (of course, in
this case, we need to set $n_1, n_2, n_3, n_4$ as concrete numbers depending on
the sign of $m$). We also need to treat the case $D = U$ separately. It is
similarly straightforward to treat \IIpl moves while $D = U$ is again an
exceptional case.

The only missing information, if we want to get $D(m)$ uniquely is to add the
names of the newly introduced crossings. Here, our convention depends on
the purpose. If we want to construct $D(m)$ for the algorithm, then we pick the
label of the new crossing as the first available integer in binary so that we
get $D(m)$ of correct size. On the other hand for the purpose of analysis of
the algorithm when we consider $V(D)$ as an abstract set, we assume that the
newly introduced crossings have fresh new labels that haven't been used
previously (in case that $D(m)$ is an intermediate step of some sequence of
Reidemeister moves on some diagram).

This also finishes the proof of~\cref{l:complexity_moves} by checking how
much which moves contribute to each item of the lemma.

\paragraph{A piece of notation for diagrams.} Here we extend our notation
$D(m)$ to several consecutive moves. 
We inductively define the notation $D(m_1, \dots, m_{k-1},
m_k) := D(m_1, \dots, m_{k-1})(m_k)$ provided that $m_k$ is a feasible
Reidemeister move in $D(m_1, \dots, m_{k-1})$. We also say that a sequence
$(m_1, \dots, m_k)$ is a \emph{feasible} sequence of Reidemeister moves for $D$
if $m_i$ is a feasible Redemeister move in $D(m_1, \dots, m_{i-1})$ for every
$i \in [k]$.

\subsection{Special and greedy moves revisited}
Let $X$ be a set of crossings in diagram $D$. We say that a move $m$
  \emph{avoids} $X$ if $m$ is 
\begin{itemize}
\item a $\Imin$ move removing a crossing $x \not \in X$, or
\item a $\IImin$ move removing crossings $x, y \not \in X$, or 
\item a $\III$ move affecting crossings $x, y, z \not \in X$, or
\item a $\Ipl$ move performed on an edge $\varepsilon$ such that no endpoint of
  $\varepsilon$ belongs to $X$, or
\item a $\IIpl$ move performed on edges $\epsvec$, $\phivec$ (possibly
  $\epsvec = \phivec$) such that no endpoint of any of these two edges
    belongs to $X$.
\end{itemize}

Let us recall from the introduction that given a set $S$ of crossings in a
diagram $D$, a feasible Reidemeister move $m$ is \emph{greedy} (with respect to
$S$) if it is a $\IImin$ move avoiding $S$. The definition of a special move
from the introduction can be equivalently reformulated so that $m$ is
\emph{special} (with respect to $S$) if it avoids crossings not in $S$, that is,
the set $V(D) \setminus S$.

Obviously, it may happen that some feasible move in $D$ is neither greedy nor
special with respect to $S$. Because of our algorithm
\sctt{Special\linebreak[0]{}Greedy}, we will be interested in sequences of moves
on $D$ such that each move is either greedy or special. However, the set $S$ may
vary while performing the special moves, thus we need to be a bit careful with
the definitions.

Given a move $m$ in $D$, we denote by $X(m)$ the following set of crossings in $D(m)$. 
\begin{itemize}
\item If $m$ is a $\III$ move, then $X(m) =
  \emptyset$.
\item If $m$ is a $\IImin$ move or a $\Imin$ move, then $X(m)$ is 
  the set of crossings removed by $m$.
\item If $m$ is a $\Ipl$ or a $\IIpl$ move, then $X(m)$ is the set of crossings
  introduced by $m$.
\end{itemize}
Let $\triangle$ denote the symmetric difference of two sets. We get the
following observation.
\begin{observation}
\label{o:Sm}
Let $m$ be a special or greedy move in a diagram $D$ with respect to a set $S$.  Let $S(m)$ be the set defined in the introduction when describing the
  \sctt{Special\linebreak[0]{}Greedy} algorithm. Then 
\begin{itemize}
\item
$S(m) = S$ if $m$ is greedy; while
\item
  $S(m) = S \triangle X(m)$ if $m$ is special.
\end{itemize}
\end{observation}
\begin{proof}
  The first item is exactly the definition of $S(m)$ for a greedy move. For the
  second item, when compared with the definition in the introduction, we have
  defined $X(m)$ as $S \triangle S(m)$ from which the observation follows.
\end{proof}

In the remainder of the paper, we plan to use~\cref{o:Sm} as an equivalent definition
of $S(m)$. The advantage of this reformulation is that $X(m)$ does not depend
on $S$.

\paragraph{Feasible moves for pairs.} Now we aim to extend our definition of a
feasible sequence $(m_1, \dots, m_\ell)$ of Reidemeister moves in a
diagram $D$ to a pair $(D, S)$ where $S$ is a set of crossing of $D$ so
that we allow only special or greedy moves. For this, we say that $m$
is a \emph{feasible} move for a pair $(D,S)$ if it is a move in $D$
special or greedy with respect to $S$. If $m$ is feasible for $(D,S)$,
then we also use the notation $(D,S)(m)$ for the resulting pair after
performing the move, that is, $(D, S)(m) := (D(m), S(m))$.  Now, we
inductively define the notation $(D,S)(m_1, \dots, m_{\ell-1}, m_\ell)
:= (D,S)(m_1, \dots, m_{\ell-1})(m_\ell)$ provided that $m_\ell$ is a
feasible Reidemeister move in $(D,S)(m_1, \dots, m_{\ell-1})$. We also
say that a sequence $(m_1, \dots, m_\ell)$ is a \emph{feasible}
sequence of Reidemeister moves for $(D,S)$ if $m_i$ is a feasible move
in $(D,S)(m_1, \dots, m_{i-1})$ for every $i \in [\ell]$. It also may
be convenient to write $(D,S)(m_1, \dots, m_{\ell-1}, m_\ell)$
directly as a pair $(D', S')$. In this case, we get $D' = D(m_1,
\dots, m_\ell)$ which follows immediately from the earlier definition
of $D(m_1, \dots, m_\ell)$ and we define $S(m_1, \dots, m_{\ell})$ as
$S'$ above. From the definition above, we also get $S(m_1, \dots,
m_{\ell}) = S(m_1, \dots, m_{\ell - 1})(m_\ell)$. However, we have to
be a bit careful about how to interpret this formula: $m_\ell$ is regarded
here as a move in $D(m_1, \dots, m_{\ell-1})$ special or greedy with
respect to $S(m_1, \dots, m_{\ell- 1})$ although $D(m_1, \dots,
m_{\ell-1})$ does not really appear in the notation. We get the
following extension of~\cref{o:Sm}.

\begin{lemma}
  \label{l:Smulti}
  Let $(m_1, \dots, m_\ell)$ be a feasible sequence for a pair $(D, S)$ where
  $D$ is a diagram and $S$ is a set of crossings. Then
  \[S(m_1, \dots, m_{\ell}) = 
  S \triangle X(m_{i_1}) \triangle \cdots \triangle X(m_{i_s})\]
  where $(i_1, \dots, i_s)$ is the subsequence of $(1, \dots, \ell)$ consisting
  of exactly the indices $i_j$ such that $m_{i_j}$ is special in $D(m_1,
  \dots, m_{i_j -1})$ with respect to $S(m_1,
    \dots, m_{i_j -1})$.
\end{lemma}

\begin{proof}
The proof follows by induction in $\ell$. If $\ell = 1$, then it directly
  follows from~\cref{o:Sm}.

If $\ell > 1$, let $(i_1, \dots, i_s)$ be the subsequence of $(1, \dots,
\ell-1)$ consisting of exactly the indices $i_j$ such that $m_{i_j}$ is special
in $D(m_1, \dots, m_{i_j -1})$ with respect to $S(m_1, \dots, m_{i_j -1})$. Let
us distinguish whether $m_{\ell}$ is greedy or special (in $D(m_1, \dots,
m_{\ell-1})$ with respect to $S(m_1, \dots, m_{\ell-1})$).

If $m_\ell$ is greedy, then by induction and by~\cref{o:Sm} applied to $m_\ell$
we get
  \[S(m_1, \dots, m_{\ell}) = S(m_1, \dots, m_{\ell-1})(m_\ell) =   (S \triangle X(m_{i_1}) \triangle \cdots
  \triangle X(m_{i_s}))(m_\ell) = S \triangle X(m_{i_1}) \triangle \cdots
  \triangle X(m_{i_s})\] as required.

  Similarly, if $m_\ell$ is special, we get
  \[S(m_1, \dots, m_{\ell}) =   (S \triangle X(m_{i_1}) \triangle \cdots
  \triangle X(m_{i_s}))(m_\ell) = S \triangle X(m_{i_1}) \triangle \cdots
  \triangle X(m_{i_s}) \triangle X(m_\ell)\] as required.
\end{proof}

As a consequence, we also get the following description of $S(m_1, \dots,  m_\ell)$.

\begin{lemma}\label{lem:special_set}
  Let $D$ be a diagram of a knot, $S$ be a set of crossings in $D$ and $(m_1,
  \dots, m_\ell)$ be a feasible sequence for the pair $(D,S)$. Then
  $S(m_1, \dots,  m_\ell)$ contains exactly the crossings in $D(m_1, \dots,
  m_\ell)$ which either already belong to $S$ or do not belong to $D$.
\end{lemma}
\begin{proof}
  Let $x$ be a crossing in $D(m_1, \dots, m_\ell)$. If $x$ does not
  belong to any of $X(m_{i_1}), \dots, X(m_{i_s})$ where $(i_1, \dots,
  i_s)$ is a sequence as in~\cref{l:Smulti}, then $x$ necessarily
  belongs to $D$ and in addition it belongs to $S$ if and only if it
  belongs to $S(m_1, \dots, m_\ell)$ by~\cref{l:Smulti}, which is
  exactly what we need.
  
  If $x$ belongs to some $X(m_{i_j})$, then $m_{i_j}$ is necessarily a $\Ipl$ or
  a $\IIpl$ move, otherwise $x$ would be removed by $m_{i_j}$ and according to
  our conventions it cannot be reintroduced later on with the same label (but we
  know that $x$ belongs to $D(m_1, \dots, m_\ell)$). We also get that in this
  case $m_{i_j}$ is unique as $x$ cannot be added to the diagram twice.  Thus
  $x$ does not belong to $D$ while it belongs to $S(m_1, \dots, m_\ell)$ due
  to~\cref{l:Smulti} which is what we need.
\end{proof}

  The following lemma shows that if the diagram can be untangled with defect
  $d$, then there exists a choice of $S$ at step 0 of the algorithm
  \sctt{Special\linebreak[0]{}Greedy$(D, k)$} that will allow the algorithm to
  find an untangling sequence.

\begin{lemma}
\label{l:step0}
Let $D$ be a diagram of a knot and let $(m_1, \dots, m_{\ell})$ be a feasible
  sequence of Reidemeister moves in $D$. Then there is a set of crossings $S$
  in $D$ such that $(m_1,
  \dots, m_{\ell})$ is feasible for $(D,S)$.

In addition, if this sequence induces an untangling with defect at most $k$,
  then $|S| \leq 3k$.
\end{lemma}

\begin{proof}
  First, we put a crossing $x$ of $D$ into an auxiliary set $S'$ in the following cases:
\begin{itemize}
\item If $x$ is removed by a $\Imin$ move in the sequence.
\item If $x$ is affected by some $\III$ move in the sequence.
\item If there is a $\Ipl$ move or a $\IIpl$ move performed on an edge with
  endpoint $x$.
\end{itemize}
  Then we define $S''$ as a set of crossings $x$ in $V(D) \setminus S'$
  such that there is $y \in S'$ such that $x$ and $y$ are simultaneously
  removed by a
  $\IImin$ move in the sequence $(m_1, \dots, m_{\ell})$ and $S'''$ as a set of
  crossings $x$ in $V(D) \setminus S'$ such that $x$ is removed by a
    $\IImin$ move in the sequence $(m_1, \dots, m_{\ell})$ simultaneously with
    some crossing $y$ not belonging to $V(D)$ (that is, with $y$ which is
    introduced by some $\Ipl$ or $\IIpl$ move before applying the move removing
    $x$ and $y$). See~\cref{f:untangling_S} for an example how may $S'$,
    $S''$ and $S'''$ look like.
We set $S := S' \cup S'' \cup S'''$ and we want to check that $(m_1,
    \dots, m_{\ell})$ is feasible with respect to $S$. 

  \begin{figure}
  \begin{center}
    \includegraphics{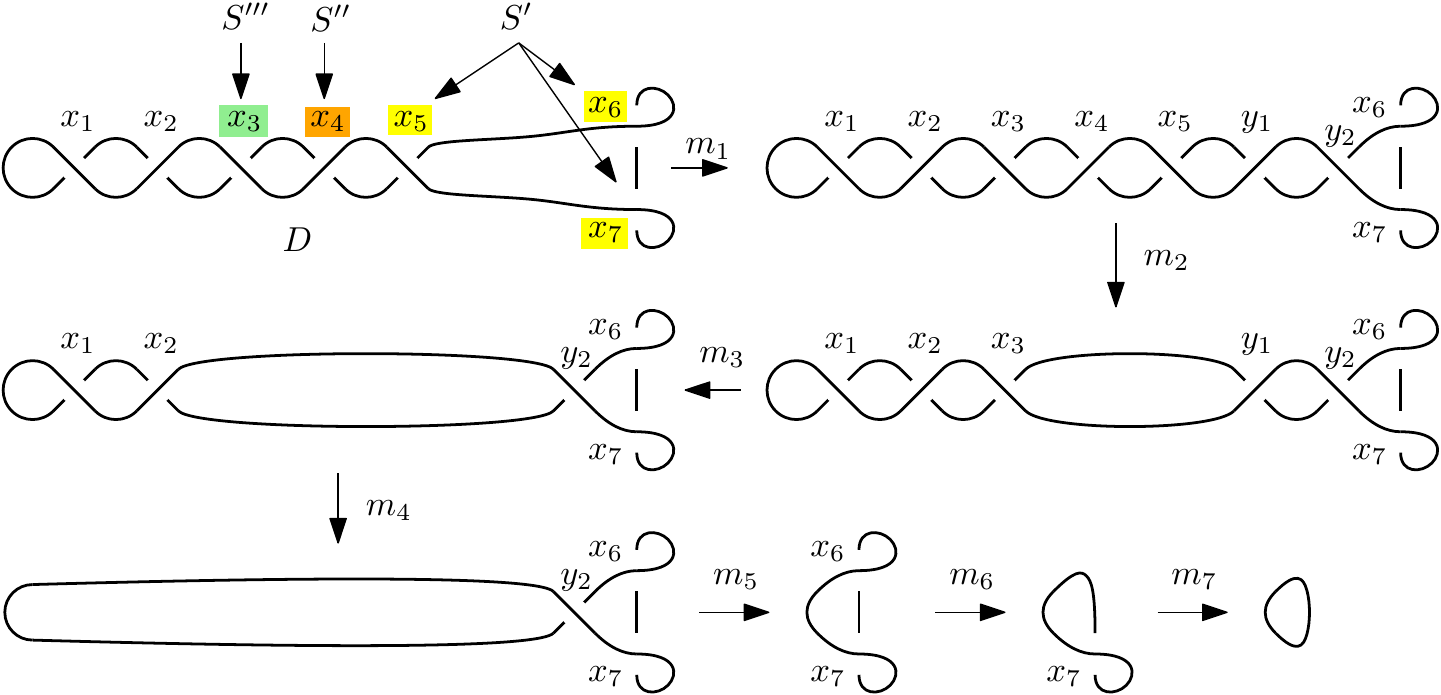}
  \end{center}
\caption{Consider a slightly artificial untangling of the diagram $D$ in the
  top left corner, given by the sequence $(m_1, \dots, m_7)$. Then $x_5$, $x_6$
  and $x_7$ belong to $S'$ because they are endpoints of edges on which the
  $\IIpl$ move $m_1$ is performed (another reason why $x_6$ and $x_7$ are in
  $S'$ is that they are removed by $\Imin$ moves $m_6$ and $m_7$). Consequently $x_4
  \in S''$ because it is removed simultaneously with $x_5 \in S'$ by a $\IImin$
  move $m_2$. We also get $x_3 \in S'''$ because it is removed simultaneously
  with $y_1 \not \in V(D)$ by a $\IImin$ move $m_3$. The move $m_4$ is the only
  greedy move in this sequence, all other moves are special.}
  \label{f:untangling_S}
\end{figure}

  In other words, we want to check that for every $i \in [\ell]$, the move
  $m_i$ is either special or greedy in $D(m_1, \dots, m_{i-1})$ with respect to
  $S(m_1, \dots, m_{i-1})$. For simplicity of notation, let $D_{i-1} := D(m_1,
  \dots, m_{i-1})$ and $S_{i-1} := S(m_1, \dots, m_{i-1})$. We will distinguish
  several cases according to the type of $m_i$.

  If $m_i$ is a $\Imin$ move removing a crossing $x$ in $D_{i-1}$ then either
  $x$ belongs to $D$ or it has been introduced by some previous $\Ipl$ or
  $\IIpl$ move (see~\cref{lem:special_set}). If $x$ belongs to $D$ then
  it belongs to $S'$ by the first condition of the definition of $S'$,
  consequently $x$ belongs to $S$ and $S_{i-1}$ as well by the definition of
  $S_{i-1}$. If $x$ has been introduced by some previous $\Ipl$ or $\IIpl$ move,
  then it belongs to $S_{i-1}$ by the definition of $S_{i-1}$ (newly introduced
  vertex is always added). In particular $m_i$ is special with respect to $S$.

  If $m_i$ is a $\III$ move then by analogous reasoning, using the second
  condition in the definition of $S'$, we deduce that every
  crossing affected by $m_i$ belongs to $S_{i-1}$. Therefore $m_i$ is special
  as well.

  If $m_i$ is a $\Ipl$ or a $\IIpl$ move, let $E$ be the set of endpoints in
  $D_{i-1}$ of the edges on which $m_i$ is performed. Again by similar
  reasoning, using the third condition in the definition of $S'$, we get that
  all crossings in $E$ belong to $S_{i-1}$. Therefore $m_i$ is special as well.

 The last case is to consider when $m_i$ is a $\IImin$ move. Let us say that $m_i$
  removes crossings $x$ and $y$. If any of these crossings does not belong to
  $D$, then it belongs to $S_{i-1}$ by~\cref{lem:special_set} and the
  other crossing either does not belong to $D$ as well which implies that it
  belongs to $S_{i-1}$, or it belongs to $D$ which implies that it belongs to
  $S'$ or $S'''$ by the definition of $S'''$. Altogether the other crossing
  belongs to $S_{i-1}$ in each case and $m_i$ is special. Thus we may assume
  that both $x$ and $y$ belong to $D$. If one of them belongs to $S'$ then the
  other one belongs to $S'$ or $S''$ and consequently $m_i$ is special. If none
  of them belongs to $S'$ then none of them belongs to $S'' \cup S'''$ by the
  definition of $S''$ and $S'''$, thus none of them belongs to $S$. Therefore
  $m_i$ is greedy in this case.

 It remains to bound the size of $S$. Let $c^+(m_i)$ be the number of newly
 created crossings by a move $m_i$ and let $c'(m_i)$ be the number of crossings
  inserted into $S'$ because of $m_i$. Explicitly, if $m_i$ is a $\Imin$ move,
  then we count the crossing removed by $m_i$ (if present in $D$) that is
  $c'(m_i) \leq 1$; if $m_i$ is a $\IImin$ move there is no such crossing, thus
  $c'(m_i) = 0$; if $m_i$ is a $\III$ move then there are at most three such
  crossings affected by $m_i$, thus $c'(m_i) \leq 3$; if $m_i$ is a $\Ipl$ move
  then there are at most two such crossings as the endpoints of the edge on
  which is $m_i$ performed, thus $c'(m_i) \leq 2$; and if $m_i$ is a $\IIpl$
  move then there are at most four such endpoints thus $c'(m_i) \leq 4$.
  Altogether, by considering the values $c^+(m_i)$, we get that $c^+(m_i) +
  c'(m_i) \leq \frac32 w(m_i)$ where $w(m_i)$ is the weight defined before
 ~\cref{l:defect}.

On the other hand, we get the following inequalities:
\[
 |S'| \leq \sum_{i=1}^\ell c'(m_i)
\]
by the definition of $c'(m_i)$;
\[
 |S''| \leq |S'| \leq \sum_{i=1}^\ell c'(m_i)
\]
by the definition of $S''$ and the previous inequality; and
\[
 |S'''| \leq \sum_{i=1}^\ell c^+(m_i)
\]
by the definition of $c^+(m_i)$ and $S'''$. Altogether, we get
\[
  |S| \leq |S'| + |S''| + |S'''| \leq \sum_{i=1}^\ell 2(c'(m_i) + c^+(m_i))
  \leq \sum_{i=1}^\ell 3 w(m_i) \leq 3k
\]
where the last inequality uses~\cref{l:defect}.
\end{proof}

We finish this section by stating the following theorem that is the key ingredient
for the proof of correctness of our \sctt{Special\linebreak[0]{}Greedy$(D, k)$}
algorithm. Namely, it shows that we cannot make wrong choice in Step 2 of the
algorithm.

\begin{theorem}
\label{t:swappable}
  Let $D$ be a diagram of a knot and $S$ be a set of crossings in $D$. Let $(m_1,
  \dots, m_\ell)$ be a feasible sequence of Reidemeister moves for $(D,S)$, inducing
  an untangling of $D$ with defect $k$.
  Assume that there is a greedy move $\widetilde m$ in $D$
  with respect to $S$. Then there is a feasible sequence of
  Reidemeister moves for $(D,S)$ starting with $\widetilde m$ and inducing an
    untangling of $D$ with defect $k$.
\end{theorem}

The proof of this theorem is given at the end of~\cref{s:rearrange}.

\section{W[P]-membership modulo~\cref{t:swappable}.}
\label{s:algo_correct}

Assuming~\cref{t:swappable} we have collected enough tools to show that
\sctt{Special\linebreak[0]{}Greedy$(D, k)$} is a W[P]-algorithm which provides a correct
answer to \sctt{Unknotting via defect}. Because this part is technically easier
and because some auxiliary lemmas are more fresh now, we show this before
proving~\cref{t:swappable}. For the reader's convenience, we first
recall the algorithm \sctt{Special\linebreak[0]{}Greedy$(D, k)$}.

\sctt{Special\linebreak[0]{}Greedy$(D, k)$:}
\begin{enumerate}[0.]
  \item (Non-deterministic step.) Guess a set $S$ of at most $3k$ crossings in $D$. Then run
     \sctt{Special\linebreak[0]{}Greedy$(D, S, k)$} below.
\end{enumerate}

\sctt{Special\linebreak[0]{}Greedy$(D, S, k)$:}

\begin{enumerate}[1.]
\item If $k < 0$, then output \sctt{No.}  If $D = U$ is a diagram without
  crossings and $k \geq 0$, then output \sctt{Yes}. In all other cases continue
  to the next step.
\item If there is a feasible greedy Reidemeister $\IImin$ move $m$ with respect
  to $S$, run \sctt{Special\linebreak[0]{}Greedy$(D(m), S(m), k)$} otherwise
  continue to the next step.
\item (Non-deterministic step.) If there is no feasible greedy Reidemeister
  $\IImin$ move with respect to $S$, enumerate all possible special Reidemeister
  moves $m_1, \dots, m_t$ in $D$ with respect to $S$ up to isotopy. If there is
  no such move, that is, if $t = 0$, then output \sctt{No}. Otherwise, make a
  guess which $m_i$ is performed first and run
  \sctt{Special\linebreak[0]{}Greedy$(D(m_i), S(m_i), k - w(m_i))$}.
\end{enumerate}

\begin{proof}[Proof of W{[P]}-membership in~\cref{t:main} modulo~\cref{t:swappable}.]

  First, we show that \sctt{Special\linebreak[0]{}Greedy$(D, k)$} provides a correct answer to
  \sctt{Unknotting via defect}. We have already explained that if the algorithm outputs \sctt{Yes}, it
  is the correct answer to \sctt{Special\linebreak[0]{}Greedy$(D, k)$}. Thus it remains to show
  that if we start with an input $(D,k)$ which admits an untangling with defect
  $k$, then the algorithm will output \sctt{Yes}. In particular, $k \geq 0$ in this
  case.

  By~\cref{l:step0}, given such an input $(D,k)$, there is a set of crossings
  $S$ in $D$ of size at most $3k$ such that there is a sequence $\mathcal M$ of
  Reidemeister moves feasible for $(D,S)$ inducing an untangling of $D$ with
  defect at most $k$.  Thus it is sufficient to show that
  \sctt{Special\linebreak[0]{}Greedy$(D, S, k)$} outputs \sctt{Yes} if such a
  sequence for a given $S$ exists. (In fact this is if and only if but we only
  need the if case.)  We will show this by a double induction on $k$ and
  $\crr(D)$, where $\crr(D)$ is the number of crossings in $D$. The outer
  induction is on $k$, the inner one is on $\crr(D)$.  It would be sufficient to
  start our induction with the pair $(k,\crr(D)) = (0,0)$; however, whenever
  $\crr(D) = 0$, then the algorithm outputs \sctt{Yes} in step 1. Thus we may
  assume that $\crr(D) > 0$.

If $(D,S)$ admits any greedy move, then we are in step 2. For every greedy move
$m$ there is a sequence of Reidemeister moves starting with $m$ feasible for
$(D,S)$ inducing an untangling of $D$ with defect at most $k$
by~\cref{t:swappable}. This move has weight $0$. Thus by~\cref{l:defect} there
is a sequence of Reidemeister moves feasible for $(D(m),S(m))$ inducing an
untangling of $D(m)$ with defect at most $k$. In addition $\crr(D(m)) <
\crr(D)$, thus \sctt{Special\linebreak[0]{}Greedy$(D(m), S(m), k)$} outputs \sctt{Yes} by
induction.

If $(D,S)$ does not admit any greedy move, then we are in step 3. We in
particular know that the first move in the sequence $\mathcal M$ must be
special. We guess this move as we are in a non-deterministic step and denote it
$m_i$ (in consistence with the notation in step 3). Now, if we remove $m_i$ from
$\mathcal M$, this is a feasible sequence for $(D(m_i),S(m_i))$ inducing an
untangling of $D(m_i)$ with defect at most $k - w(m_i)$ by~\cref{l:defect}. Thus
\sctt{Special\linebreak[0]{}Greedy$(D(m_i), S(m_i), k - w(m_i))$} outputs \sctt{Yes} by
induction, as we need.

\bigskip
 
  We also need to show that the algorithm \sctt{Special\linebreak[0]{}Greedy$(D, k)$} is a 
W[P]-algorithm. 

We have provided an encoding of diagrams such that if $D$ contains $n$
crossings, then the size of the encoding is $\Theta(n \log n)$.
Therefore, when estimating the complexity of the algorithm, we may
estimate it with respect to the number of crossings of $D$ (up to a polynomial
factor).

Assume that $(D,k)$ is a positive instance of \sctt{Special\linebreak[0]{}Greedy$(D, k)$}. In
step 0, there are $\binom n{3k} + \cdots + \binom n0 \leq (3k+1)n^{3k}$ choices for
how to pick $S$. This can be encoded as at most $O(\log_2((3k+1)n^{3k}))$ binary
choices; thus we perform $O(k \log n)$ non-deterministic choices in step 0.

Now, we analyse the complexity of the remaining steps, that is, the complexity
of \sctt{Special\linebreak[0]{}Greedy$(D, S, k)$}. We can easily bound the number of greedy
$\IImin$ moves performed in step 2 by $n/2$ because each such move removes two
non-special crossings and there are at most $n$ non-special crossings in the
beginning.

We also want to bound the number of the special moves guessed in step 3. Here,
we only consider a sequence of guesses yielding \sctt{Yes}. (While we terminate
the algorithm if we need more guesses than the claimed bound.)

Let $s_{\Imin}$, $s_{\Ipl}$, $s_{\IImin}$, $s_{\IIpl}$ and
$s_{\III}$ be the number of special moves guessed in step 3 in this sequence, of
type $\Imin$, $\Ipl$, $\IImin$, $\IIpl$ and $\III$ respectively. For the moves
of positive weight, we get the following bounds $s_{\Imin} \leq k$, $s_{\III}
\leq k/2$, $s_{\Ipl} \leq k/3$ and $s_{\IIpl} \leq k/4$ as a move with positive
weight reduces $k$ by its weight. We also want to bound $s_{\IImin}$ by a
function of $k$. We observe that 
\[
  |S| = 2 s_{\IImin} + s_{\Imin} - s_{\Ipl} - 2s_{\IIpl}
\]
because at the end of the sequence we have no crossings while each special
$\IImin$ move reduces the number of special crossings by $2$, each special
$\Imin$ move reduces the number of special crossings by $1$, etc. This gives
\[
s_{\IImin} = |S|/2 - s_{\Imin}/2 + s_{\Ipl}/2 + s_{\IIpl} \leq 3k/2 + 0 + k/6 +
k/4 \leq 2k.
\]
Altogether, these bounds show that the number of special moves guessed in
step 3 is $O(k)$.

Now, we put everything together. By bounding the number of occurences of step 2
and step 3, we see that any intermediate diagram appearing in the algorithm has
size at most $O(n+k)$. Thus, by~\cref{l:complexity_moves}(i), there at most
$O((n+k)^2)$ feasible special moves in step 3 which can be constructed in time
$O((n+k)^2\log(n+k))$. In particular, guesses in step 3 can be encoded as
$O(\log(n+k))$ binary choices. This means that the number of non-deterministic
steps (including step 0) is $O(k \log n)$ as required. It also follows,
using~\cref{l:complexity_moves}(ii) and using again that steps 2 and 3 are
performed at most $O(n+k)$ times that the algorithm runs in nondeterministic
time $O((n+k)^2\log(n+k))$. This belongs to $O(k^3 n^3)$ which is form
$O(f(k)n^c)$ where $c$ is a fixed constant and $f(k)$ is a computable function
as required.
\end{proof}

\section{Rearranging \IImin moves}
\label{s:rearrange}

The goal of this section is to prove~\cref{t:swappable}. We start with auxiliary
results on rearranging two consecutive moves.

To each feasible Reidemeister move $m$ in a diagram $D$, we can assign its
\emph{private} $2$-ball $B(m) \subseteq \R^2$ so that any modifications of $D$
via $m$ are performed only inside $B(m)$; see~\cref{f:rm_private_balls}.

\begin{figure}
\begin{center}
  \includegraphics[page=5]{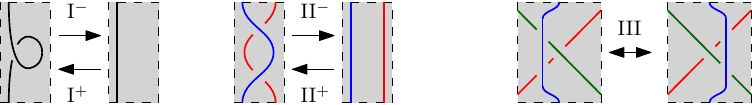}
\end{center}
  \caption{Private balls $B(m)$ (in gray) around Reidemeister moves.}
  \label{f:rm_private_balls}
\end{figure}

Now let us assume that we have two feasible Reidemeister moves $m$ and $m'$ in a
diagram $D$ such that we can choose $B(m)$ and $B(m')$ so that they are
disjoint. Considering the moves purely topologically (ignoring the combinatorial
description) we can perform the two moves in any order yielding the same
diagram; see~\cref{f:im}.

\begin{figure}
\begin{center}
  \includegraphics[page=1]{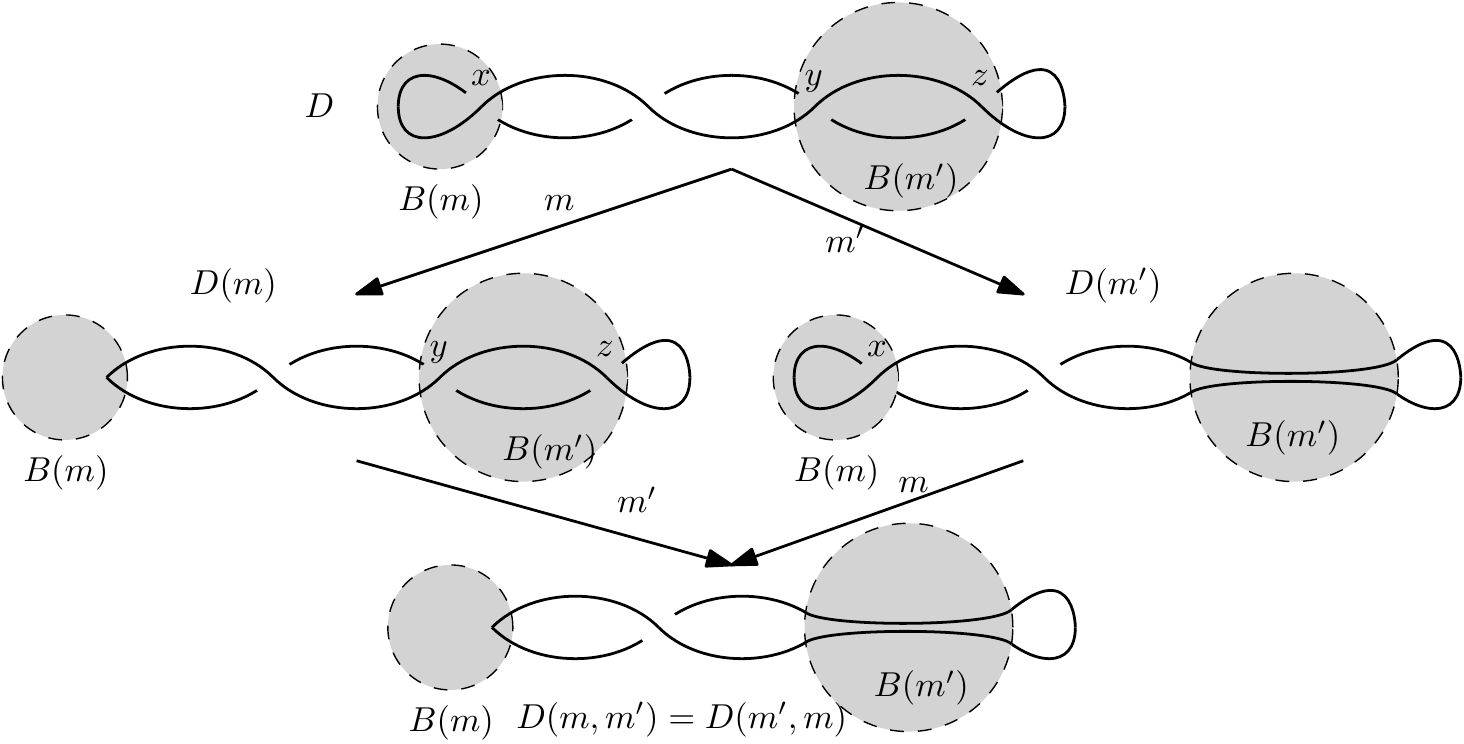}
\end{center}
  \caption{Two feasible moves $m = \{x\}_{\Imin}$ and $m' = \{y,z\}_{\IImin}$
  in a diagram $D$ with disjoint private balls. After applying the two moves in
  any order, we get the same diagram.}
  \label{f:im}
\end{figure}

In example, on~\cref{f:im}, the two used moves have the same combinatorial
description in both orders. Namely $m$ is feasible in $D(m')$, $m'$ is feasible
in $D(m)$ and $D(m, m') = D(m', m)$. However, this need not be true in general
even if $B(m)$ and $B(m')$ are disjoint. An example of such behavior is given on
Figure~\ref{f:im_close}. Thus, we will sometimes need to introduce a new name
for a move after swapping (for example $\widehat m'$ on
Figure~\ref{f:im_close}).  But this is essentially only a notational
problem.\footnote{With some extra effort we would be able to explain that we
only need swapping moves which remain combinatorially the same after the
swap. But it is simply not worth some additional case analyses. Changing the
notation is easier.}

  \begin{figure}
  \begin{center}
    \includegraphics[page=2]{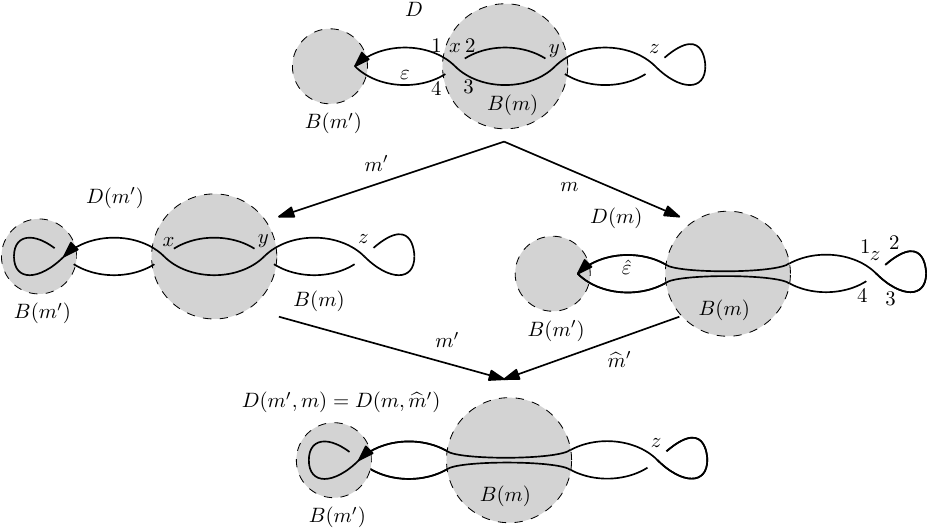}
  \end{center}
  \caption{Two feasible moves $m = \{x,y\}_{\IImin}$ and $m' =
    (\epsvec,+)_{\Ipl}$in a diagram $D$ with disjoint private balls. After
    applying $m$, the move $m'$ is not feasible anymore but instead of it we
    have to use a combinatorially different move $\widehat m' =
    (\protect\overrightarrow{\widehat{\varepsilon}},+)_{\Ipl}$ to get the same
    result. (Note that the edge $\varepsilon$ is different from $\widehat
    \varepsilon$.)}
  \label{f:im_close}
  \end{figure}
  
Now we provide the main technical lemma which allows swapping the moves.

\begin{lemma}
  \label{l:swap_i}
Let $D$ be a diagram of a knot and $S$ be a set of crossings in $D$. Let $m =
\{x,y\}_{\IImin}$ be a $\IImin$ move in $D$, greedy with respect to $S$. Let $m'$
be a feasible move for the pair $(D,S)$ which avoids $\{x,y\}$. Then there is a
move $\widehat m'$ feasible in $D(m)$ of same type as $m'$ and the
following conditions hold:
  \begin{enumerate}[(i)]
\item $m$ is greedy in $D(m')$ with respect to $S(m')$, in particular $m$ is
    feasible in $(D,S)(m')$
\item 
    $\widehat m'$ is feasible in $(D,S)(m)$; and
\item
    $(D,S)(m, \widehat m') = (D,S)(m', m)$.
\end{enumerate}
\end{lemma}

\begin{proof}
  We will use case analysis according to the type of $m'$
  in order to show that the private balls $B(m)$ and $B(m')$ can be chosen to
  be disjoint and we will perform the moves only inside these private
  balls. As soon as we get this, 
let $\widehat m'$ be the move in $D(m)$ which performs the
  same change inside $B(m')$ as $m$ in~$D$. We also observe that the move which
  performs the same change inside $B(m)$ in $D(m')$ as $m$ in $D$ is again $m$,
  because a $\IImin$ move is uniquely determined by the crossings removed by this
  move.
  Then we get that $m$ is feasible in $D(m')$ and  $D(m, \widehat m') = D(m',
  m)$, which will be our first step in the proof.  
  
  Now, let us distinguish cases according to the type of $m'$.

\smallskip
  We start with the case that $m'$ is a $\Imin$ move, $\IImin$ move or a
  $\III$ move. The condition that $m'$ avoids $\{x, y\}$ implies
  that the circle of $c(m)$ is disjoint from the circle $c(m')$ of $m'$.
  (Recall that $c(m)$ is depicted at~\cref{f:rm_circle}.) In
  addition, these two circles cannot be nested as $D$ does not pass through the
  open disks bounded by them. Thus $B(m)$ can be chosen so that it contains
  $c(m)$, the disk bounded by $c(m)$ and a small neighborhood of $c(m)$.
  Similarly, $B(m')$ can be chosen so that it contains
    $c(m')$, the disk bounded by $c(m')$ and a small neighborhood of $c(m')$.
    Because $c(m)$ and $c(m')$ have strictly positive distance, the `small
    neighborhoods' above can be chosen sufficiently small so that $B(m)$ and
    $B(m')$ are disjoint. 
    
\smallskip
  Next, let us consider the case that $m'$ is a $\Ipl$ move performed on an
  edge $\varepsilon'$. Because $m'$ avoids $\{x,y\}$ we
  get that the image of the edge $\varepsilon'$ is disjoint from $c(m)$. Thus
  we can chose $B(m)$ as in the previous case so that it avoids $\varepsilon'$
  while it is sufficient to pick $B(m')$ as a small neighborhood of arbitrary
  point on $\varepsilon'$ disjoint from $B(m)$.  
  
\smallskip
  Finally, let us consider the case that $m'$ is a $\IIpl$ move performed on
  edges $\varepsilon'$ and $\varphi'$, not necessarily distinct.
  Because $m'$  avoids $\{x,y\}$ we
    get that the image of the edges $\varepsilon'$ and $\varphi'$ is disjoint
    from $c(m)$. Following the description of feasible $\IIpl$ moves, let
    $\zeta'$ be a topological arc (inside a face of the diagram, except the
    endpoints) connecting $\varepsilon'$ and $\varphi'$ along which the
    move $m'$ is performed. Then $\zeta'$ is also disjoint from $c(m)$.
    Consequently, we can choose $B(m)$ as above so that it is disjoint from
    $\zeta'$ while $B(m')$ is chosen in a small neighborhood of $\zeta'$ so
    that it is disjoint from $B(m)$. 
    
\smallskip
This finishes the initial case analysis, thus we know that $m$ is feasible in
$D(m')$ and $D(m, \widehat m') = D(m', m)$, which will be our first step in the
proof. In order to finish the proof, we observe that $S(m) = S$ because $m$ is
greedy in $D$ with respect to $S$ and we distinguish whether $m'$ is greedy or
special in $D$ with respect to $S$.

If $m'$ is greedy, then it is a $\IImin$ move. We get that $S(m') = S$. Say that
$m'$ removes crossings $w,z$ distinct from $x, y$. Then $S$ does not contain
$w,z$ because $m'$ is greedy. Consequently both $m$ and $\widehat m'$ are greedy
with respect to $S = S(m) = S(m')$ (this proves (i) and (ii)), therefore both
$(m, \widehat m')$ and $(m', m)$ are feasible in $(D,S)$ and $S(m, \widehat m')
= S(m', m) = S$ by~\cref{l:Smulti}. Together with $D(m, \widehat m') = D(m', m)$
this proves (iii).

Now, let us assume that $m'$ is special. Because neither $x$ nor $y$ belongs to
$X(m')$ (as $m'$ avoids $\{x,y\}$), we get that $m$ is greedy in $D(m')$ with
respect to $S(m')$ which in particular gives (i). Consequently $(m', m)$ is
feasible for $(D,S)$ and $S(m', m) = S \triangle X(m')$ by~\cref{l:Smulti}. On
the other hand, $\widehat m'$ is special in $D(m)$ with respect to $S(m)$ as
$S(m) = S$ (and from the definition of $\widehat m'$) which in particular gives
(ii). We also get that $X(m') = X(\widehat m')$ as they are of
the same type and they affect the same set of crossings.  Altogether, $(m,
\widehat m')$ is feasible for $(D,S)$ and $S(m, \widehat m') = S \triangle
X(\widehat m') = S \triangle X(m')$ by~\cref{l:Smulti}. Thus we checked $S(m',
m) = S(m, \widehat m')$. Together with $D(m', m) = D(m, \widehat m')$, this
gives (iii).
\end{proof}

\begin{corollary}
  \label{c:swap_i}
  Let $D$ be a diagram, $S$ be a set of crossings in $D$, $\ell \geq 2$, and
  $(m_1, \dots, m_\ell)$ be a feasible sequence of
  Reidemeister moves for the pair $(D,S)$.
  Assume that $m_\ell$ is a $\IImin$ move which is also feasible and greedy in
  $D$ with respect to $S$.
  Then there is a sequence of Reidemester moves
  $(m_\ell, \widehat m_1, \dots, \widehat m_{\ell-1})$ feasible for $(D,S)$ such that
  \begin{enumerate}[(i)]
    \item $(D,S)(m_1, \dots, m_{\ell}) = (D,S)(m_\ell, \widehat m_1,
  \dots, \widehat m_{\ell-1})$; and
  \item $w(m_i) = w(\widehat m_i)$ for $i \in [\ell-1]$.
\end{enumerate}
\end{corollary}

\begin{proof}
We will prove the corollary by induction in $\ell$. If $\ell=2$, then this
  immediately follows from~\cref{l:swap_i} used with $m = m_2$ and $m' =
  m_1$ obtaining $\widehat m_1 := \widehat m'$. Note that $w(m_1) = w(\widehat
  m_1)$ because $m_1$ and $\widehat m_1$ are of the same type. Thus we may assume
  that $\ell \geq 3$.

  Let $D' := D(m_1)$ and $S' := S(m_1)$. We get that $(m_2, \dots, m_\ell)$ is a
  feasible sequence for $(D',S')$. We first apply~\cref{l:swap_i} in $D$
  with $m = m_\ell$ and $m' =
    m_1$, at this moment, we only deduce that $m_\ell$ is feasible in $D'$ and
    greedy with respect to $S'$ from the item (i) of~\cref{l:swap_i}. We
    also get that $(m_2, \dots, m_\ell)$ is feasible for $(D',S')$ because $(m_1,
    \dots, m_\ell)$ is feasible for $(D,S)$.
  Therefore, by induction, there is a sequence
  $(m_\ell, \widehat m_2, \dots, \widehat  m_{\ell-1})$ feasible for $(D',S')$ such
  that  
  \begin{enumerate}[(i')]
    \item $(D',S')(m_2, \dots, m_\ell) = (D',S')(m_\ell, \widehat m_2,
  \dots, \widehat m_{\ell-1})$; and
  \item $w(m_i) = w(\widehat m_i)$ for $i \in \{2, \dots, \ell-1\}$.
\end{enumerate}

  Now, for checking (i) and (ii), we use~\cref{l:swap_i} in its full
  strength, again applied with $m = m_2$ and $m' = m_1$. Let $\widehat
  m_1 := \widehat m'$, using the notation of the lemma. 
 We immediately deduce (ii) from (ii') as $m_1$ and $\widehat m_1$ are of
  the same type. We also get $(D,S)(m_1,  m_\ell) = (D,S)(m_\ell, \widehat
  m_1)$ from
  item (iii) of~\cref{l:swap_i}.
  Using this and (i') we get 
  \begin{align*}
    (D,S)(m_1, \dots, m_\ell) &= (D',S')(m_2, \dots, m_\ell)\\
		       &= (D',S')(m_\ell, \widehat m_2, \dots, \widehat m_{\ell-1})\\
		       &= (D,S)(m_1, m_\ell)(\widehat m_2, \dots, \widehat
		       m_{\ell-1})\\
		       &= (D,S)(m_\ell, \widehat m_1)(\widehat m_2, \dots,
		       \widehat m_{\ell-1})\\
		       &= (D,S)(m_\ell, \widehat m_1, \widehat m_2, \dots,
		       \widehat m_{\ell-1})
  \end{align*}
  which verifies (i); note that the computations also show that the sequence
  $(m_\ell, \widehat m_1, \widehat m_2, \dots, \widehat m_{\ell-1})$ is feasible for
  $(D,S)$.
\end{proof}

We also need one more swapping lemma for two consecutive moves (and one
more auxiliary move).

\begin{lemma}
  \label{l:swap_ii}
Let $D$ be a diagram of a knot and $S$ be a set of crossings in $D$. 
  Assume that $m = \{x,y\}_{\IImin}$ and $\widetilde m = \{x,z\}_{\IImin}$, with
  $y \neq z$, are feasible $\IImin$ moves in $D$, greedy with respect to $S$.
  Assume also that $m'$ is a feasible move for $(D(m), S(m))$ which
  avoids $\{z\}$. Then there is a move $\widehat m'$ feasible for
  $(D,S)$ of same type as $m'$ such that the following conditions hold:
\begin{enumerate}[(i)]
\item 
  $m$ is feasible for $(D,S)(\widehat m')$;
\item 
  $(D,S)(m, m') = (D,S)(\widehat m', m)$.
\end{enumerate}
\end{lemma}

\begin{proof}
\medskip
  Many steps in the proof will be similar to the proof of~\cref{l:swap_i}; however, we have to be
  careful with the analogy as the assumptions are different. Our first aim is again
  to get the private balls $B(m)$ and $B(m')$ disjoint by case analysis
  according to the type of $m'$. Note that this makes
  sense in particular when focusing to diagram $D(m)$ where $B(m)$ appears
  because $D(m)$ was obtained by performing $m$ while $B(m')$ appears because
  $m'$ is feasible in $D(m)$. Once we get that $B(m)$ and $B(m')$ are disjoint,
  then we know that there is a move $\widehat m'$ in $D$ which performs the same
  change inside $B(m')$ as the move $m'$ in $D(m)$. Similarly, there is a move
  $\widehat m$ in $D(\widehat m')$ which performs the same change as $m$ in
  $D$. Then we get $D(m,m') = D(\widehat m', \widehat m)$. In addition,
  similarly as in the proof of~\cref{l:swap_i}, we get $m = \widehat m$
  for free because $m$ is a $\IImin$ move. 
  
  Before we start the
  case analysis according to the type of $m'$, we start with some
  common preliminaries. 

  Let us pick the private ball $B(m)$ arbitrarily so that it meets the diagram
  $D$ as in~\cref{f:rm_Bm}, left. Let $\alpha$ and $\beta$ be two arcs in $D$
  such that $B(m)$ intersects (the image of) $D$ in (the images of) $\alpha$
  and $\beta$. Let $\alpha(m)$ and $\beta(m)$ be the arcs in $D(m)$ coming from
  $\alpha$ and $\beta$ after applying the move $m$. We also define
  $\varepsilon(m)$
  and $\varphi(m)$ as the edges of $D(m)$ which contain $\alpha(m)$
  and $\beta(m)$. An
  important observation is that one of the endpoints of each of
  $\varepsilon(m)$ and $\varphi(m)$ is $z$. Indeed, because $\widetilde m = \{x,z\}_{\IImin}$ is
  feasible in $D$, there are two edges $\widetilde \varepsilon$, $\widetilde
  \varphi$ in $D$ connecting $x$ and $z$; see~\cref{f:rm_Bm}, right.
  Then, one of them, say $\widetilde \varepsilon$, partially coincides with $\alpha$ while
  the other one, $\widetilde \varphi$, partially coincides with $\beta$. It
  follows that $\varepsilon(m)$ contains the image of the arc obtained
  from $\widetilde \varepsilon$ after
  performing $m$. In particular, $\varepsilon(m)$ contains $z$ which implies that
  $z$ is an endpoint of $\varepsilon(m)$. Similarly, $z$ is an endpoint of
  $\varphi(m)$. Note also that if we perform the move $m$ on $D$ as on
 ~\cref{f:rm_Bm}, left, then $\widetilde \varepsilon$ and $\widetilde
  \varphi$ as well as $\varepsilon(m)$ and $\varphi(m)$ have to
  leave $B(m)$ from the `upper side' as on~\cref{f:rm_Bm}, right, because $z \neq y$.
  Now we start distinguishing the cases.

\begin{figure}
\begin{center}
  \includegraphics[page=6]{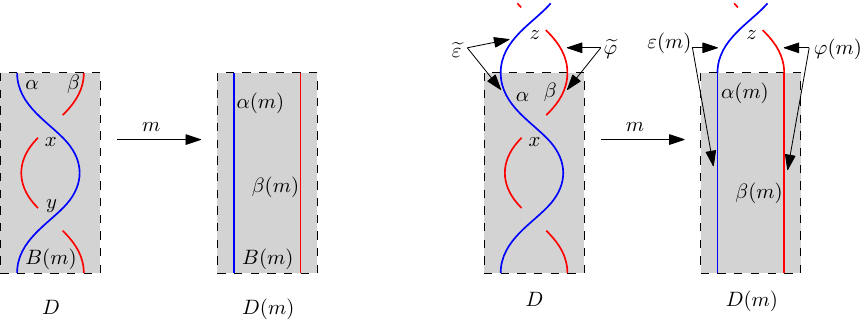}
\end{center}
  \caption{Performing the move $m$.}
  \label{f:rm_Bm}
\end{figure}
  We again start with the case that $m'$ is a $\Imin$ move, $\IImin$ move or a
  $\III$ move. In this case, we want to check that $c(m')$ is disjoint from
  (the images of) $\alpha(m)$ and $\beta(m)$. However, if $c(m')$ meets, for
  example $\alpha(m)$, then it has to contain (the image of) $\varepsilon(m)$ and
  in particular $z$. But this is impossible as $m'$ avoids $z$.
  Therefore, $c(m')$ has positive distance from $\alpha(m)$ and $\beta(m)$ which
  also implies that it has positive distance from $B(m)$. In addition, $B(m)$
  cannot be contained in the open disk bounded by $c(m')$ because $D(m)$ does
  not pass through this disk. Therefore, $B(m')$ can be chosen disjoint from
  $B(m)$ so that it  contains this disk and a small neighborhood of $c(m')$.
  
  Next, let us consider the case that $m'$ is a $\Ipl$ move; say $m' =
  (\epsvec',\sigma')_{\Ipl}$. Because $m'$ avoids $\{z\}$, we get that the image
  of the edge $\varepsilon'$ is disjoint from (the images of) $\alpha(m)$ and
  $\beta(m)$ (because they are in edges containing $z$). Thus it is possible to
  pick $B(m')$ as a small neighborhood of an arbitrary point on $\epsvec'$
  disjoint from $B(m)$.

  Finally, let us consider the case that $m'$ is a $\IIpl$ move; say $m' =
  (\epsvec', \phivec', w', o')_{\IIpl}$. Because $m'$ avoids $\{z\}$ we get that
  the image of the edges $\epsvec'$ and $\phivec'$ is disjoint from
  $\varepsilon(m)$ and $\varphi(m)$. Following the description of feasible
  $\IIpl$ moves, let $\zeta'$ be a topological arc (inside a face of the diagram
  $D(m)$, except the endpoints) connecting $\epsvec'$ and $\phivec'$ along which
  the move $m'$ is performed. We want to show that, up to isotopy, $\zeta'$ can
  be chosen so that it is disjoint with $B(m)$. Indeed we can perform an
  (ambient) isotopy which keeps the diagram $D(m)$ fixed and which maps $B(m)$
  to a ball $\Bar B(m)$ inside a small neighborhood of the union (of the images)
  of $\epsvec(m)$ and $\phivec(m)$ so that $\Bar B(m)$ and $\zeta'$ are
  disjoint; see~\cref{f:rm_isotopy}. This isotopy is indeed possible
  because $\epsvec(m)$ and $\phivec(m)$ leave $B(m)$ from the `upper side' as we
  argued above in the case analysis.

  Now, if we want to get $\zeta'$ disjoint with $B(m)$, up to isotopy, it is
  sufficient to apply the inverse of this isotopy to $\zeta'$. Once we have
  $\zeta'$ disjoint with $B(m)$, we can pick $B(m')$ as a small neighborhood of
  $\zeta'$, disjoint with $B(m)$. This finishes the proof.

\begin{figure}
\begin{center}
  \includegraphics[page=7]{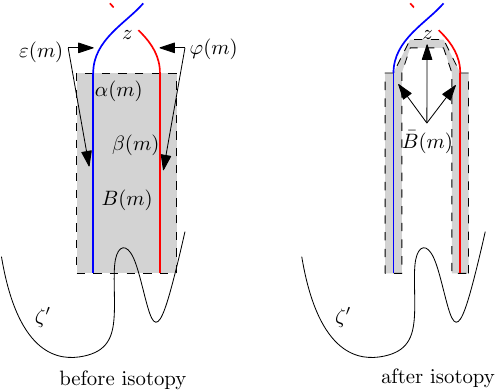}
\end{center}
  \caption{Isotoping $B(m)$ so that it avoids $\zeta'$.}
  \label{f:rm_isotopy}
\end{figure}
\smallskip

  This finishes our initial case analysis, thus we get that $\widehat m'$ is
  feasible in $D$, $m$ is feasible in $D(\widehat m')$ and $D(m,m') = D(\widehat
  m', m)$. It remains to check that these properties are kept also with respect
  to $S$.

First, we realize that $S = S(m)$ due to~\cref{o:Sm} because $m$ is
greedy in $D$ with respect to $S$. Next, we distinguish whether $m'$ is greedy
or special with respect to $S(m)$.

First assume that $m'$ is greedy, in particular a $\IImin$ move. Let $m' =
\{x',y'\}_{\IImin}$.  Then we get that $x',y' \not \in S(m) = S$. Thus,
$\widehat m'$ is also greedy (with respect to $S$) as it removes the same
crossings. In particular $\widehat m'$ is feasible for $(D,S)$. By another
application of~\cref{o:Sm}, we get $S(\widehat m') = S$. Therefore $m$ is greedy
in $D(\widehat m')$ with respect to $S(\widehat m')$, which gives (i). In
addition $S(m,m') = S = S(\widehat m', m)$ by~\cref{l:Smulti} which gives (ii),
using also $D(m,m') = D(\widehat m', m)$.

Now, let us assume that $m'$ is special. If $\widehat m'$ is a $\Imin$,
$\IImin$, or a $\III$ move, we deduce that $\widehat m'$ is special in $D$ with
respect to $S$ because it removes/affects the same crossings as $m'$ and $S =
S(m)$. If $\widehat m'$ is a $\Ipl$ or a $\IIpl$ move, we deduce that $\widehat
m'$ is special in $D$ with respect to $S$ as soon as we check that the endpoints
of the edge(s) on which is $\widehat m'$ performed belong to $S$.  However, this
edge/these edges cannot belong to the private ball $B(m)$ by our earlier case
analysis (we have obtained disjointess with $\alpha(m)$ and $\beta(m)$). In
particular, this edge/these edges are not affected by $m$ when performed in $D$,
which shows that their endpoints have to belong to $S$ (as $m'$ is special in
$D(m)$ with respect to $S(m) = S$).  In particular, in all cases according to
the type of $m'$ we get that $\widehat m'$ is feasible with respect to $(D, S)$
and we also get $X(m') = X(\widehat m')$ and $x,y \not \in
X(m')$. By~\cref{o:Sm}, we get $S(\widehat m') = S \triangle X(\widehat m') = S
\triangle X(m')$. This shows that $m$ is greedy in $D(\widehat m')$ with respect
to $S(\widehat m')$ as $x,y \not \in S$ and $x,y \not \in X(m')$. In particular,
this imples (i). Finally, by~\cref{l:Smulti}, we get $S(m,m') = S \triangle
X(m') = S \triangle X(\widehat m') = S(\widehat m', m)$. Together with $D(m,m')
= D(\widehat m', m)$, this gives (ii).
\end{proof}
\begin{corollary}
  \label{c:swap_ii}
Let $D$ be a diagram, $S$ be a set of crossings in $D$, $\ell \geq 2$ and $(m_1,
\dots, m_\ell)$ be a feasible sequence of moves for $(D,S)$. Assume that $m_1 =
\{x_1, y_1\}_{\IImin}$ and $m_\ell = \{x_\ell, y_\ell\}_{\IImin}$ are $\IImin$
moves where $x_1, y_1, x_\ell, y_\ell \not \in S$ (in particular $m_1$ is greedy
in $D$ with respect to $S$). Finally, assume also that $\{x_1,
x_\ell\}_{\IImin}$ is a feasible Reidemeister move for $(D,S)$ (again, it must
be greedy). Then, there is a feasible sequence $(\widehat m_2, \widehat m_3,
\dots, \widehat m_{\ell-1}, m_1)$ of moves for $(D,S)$ such that
  \begin{enumerate}[(i)]
    \item $(D,S)(m_1, \dots, m_{\ell-1}) =  (D,S)(\widehat m_2, \widehat m_3,
      \dots, \widehat m_{\ell-1}, m_1)$; and
\item $w(m_i) = w(\widehat m_i)$ for $i \in \{2, \dots, \ell-1\}$.
\end{enumerate}
\end{corollary}

\begin{proof}
  We will prove the corollary by induction in $\ell$. For $\ell=2$, there is
  nothing to prove (considering $m_2, \dots, m_{\ell-1}$ as an empty sequence).

  Let us assume $\ell \geq 3$. Let us apply~\cref{l:swap_ii} with $m = m_1$
  and $m' = m_2$, obtaining $\widehat m_2 := \widehat m'$. In particular, we
  get $w(m_2) = w(\widehat m_2)$, verifying (ii) for $i=2$.

  We also
  know that $\widehat m_2$ is feasible for $(D,S)$. Let $D' := D(\widehat m_2)$ and $S'
  = S(\widehat m_2)$. Then $m_1$ is feasible for $(D',S')$ by item (i) of
 ~\cref{l:swap_ii} and we get 
  \begin{equation}
    \label{e:DS_swap_ii}
    (D,S)(m_1, m_2) = (D,S)(\widehat m_2, m_1) = (D',S')(m_1)
  \end{equation}
by item (ii) of~\cref{l:swap_ii}.
This together gives that  the sequence
  $(\widehat m_2, m_1, m_3, m_4, \dots, m_{\ell})$ is feasible for $(D,S)$,
  therefore $(m_1, m_3, m_4, \dots, m_\ell)$ is feasible for $(D', S')$.
  We want to check the
  induction  assumptions for $D'$ and this sequence. We need to check:
\begin{itemize}
  \item $m_1$ is feasible for $(D',S')$. This follows from item (i) of
   ~\cref{l:swap_ii}.
  \item $x_1, y_1, x_\ell, y_\ell$ do not belong to $S'$. This follows from the
    assumption $x_1, y_1, x_\ell, y_\ell \not \in S$ and~\cref{o:Sm}. Indeed if
    $\widehat m_2$ is greedy, then $S'=S$ and there is nothing to do. Thus we
    may assume that $\widehat m_2$ is special and $S' = S \triangle X(\widehat
    m_2)$.  Then it is sufficient to check that none of $x_1, y_1, x_\ell,
    y_\ell$ belongs to $X(\widehat m_2)$. If $\widehat m_2$ is a $\Imin$ move,
    $\IImin$ move or a $\III$ move, then $X(\widehat m_2) \cap \{x_1, y_1,
    x_\ell, y_\ell\}$ is empty because $\widehat m_2$ is special with respect to
    $S$, thus it cannot remove/affect a crossing not in $S$. If $\widehat m_2$
    is a $\Ipl$ move or a $\IIpl$ move, then $X(\widehat m_2)$ is the set of
    crossing created by $\widehat m_2$. The crossings $x_1, y_1, x_\ell$ cannot
    be created by $\widehat m_2$ because they are already present in $D$ (due to
    feasible moves $m_1$ and $m_\ell$). However, $y_\ell$ is a crossing of $D$
    as well as it is removed by $m_\ell$ while it cannot be created during the
    sequence $(m_1, \dots, m_{\ell-1})$ because this would make $m_\ell$ special
    (with respect to $S(m_1, \dots, m_{\ell-1})$) contradicting that $x_\ell$
    does not belong to $S$.
  \item $\{x_1, x_\ell\}_{\IImin}$ is a feasible Reidemeister move in $D'$. This
    follows from~\cref{l:swap_i}(i) applied to $(D,S)$ with $m =
    \{x_1,x_{\ell}\}_{\IImin}$ and $m' = \widehat m_2$. For this, we need to check that
    $\widehat m_2$ avoids $\{x_1,x_\ell\}$. 
    
    This is almost for free if $\widehat m_2$
    is a special move in $D'$ with respect to $S'$, because in this case
    $\widehat m_2$ avoids crossings not in $S'$ by the definition of special move.
    In particular, it avoids $x_1$ and $x_\ell$ which are not in $S$ because
    they can be removed by the greedy $\{x_1, x_\ell\}$ move in $D$ with
    respect to $S$ --- therefore, they do not belong to $S'$ either.

    However, it is also easy to check that $\widehat m_2$ avoids
    $\{x_1,x_\ell\}$ if $\widehat m_2$ is greedy, in particular a $\IImin$
    move. We know that $m_2$ and $\widehat m_2$ are of the same type, thus $m_2$
    is greedy as well. In addition, from $D(m_1, m_2) = D(\widehat m_2, m_1)$
    (which is a special case of~\eqref{e:DS_swap_ii}), we get that $\widehat
    m_2$ has to remove the same crossings as $m_2$, thus $\widehat m_2 =
    m_2$. Then $m_2$ avoids $x_1$ because $(m_1,m_2)$ is a feasible sequence in
    $D$ and $m_1$ removes $x_1$ and $m_2$ avoids $x_\ell$ because $(m_1, \dots,
    m_{\ell})$ is a feasible sequence in $D$ and thus $x_\ell$, which is removed
    by $m_\ell$ has to persist until $D(m_1, \dots, m_{\ell-1})$.
\end{itemize}
Therefore, by induction we get a feasible sequence $(\widehat m_3, \widehat
m_4, \dots, \widehat m_{\ell-1}, m_1)$ for $(D',S')$ and 
\begin{equation}
  \label{e:induction_swap_ii}
  (D',S')(m_1, m_3,
  m_4, \dots, m_{\ell-1}) =  (D',S')(\widehat m_3, \widehat m_4, \dots, \widehat m_{\ell-1},
m_1).
\end{equation} In particular, $w(m_i) = w(\widehat m_i)$ for $i \in \{3,\dots,
\ell-1\}$ by induction which together with $w(m_2) = w(\widehat m_2)$ verifies
item (ii) of this lemma.

For~(i), by \eqref{e:DS_swap_ii} and~\eqref{e:induction_swap_ii} we obtain
  \begin{align*}
    (D,S)(m_1, \dots, m_{\ell-1}) &= (D,S)(m_1,m_2)(m_3, \dots, m_{\ell-1}) \\
		       &=(D',S')(m_1, m_3, m_4, \dots, m_{\ell-1})\\
		       &=(D',S')(\widehat m_3, \widehat m_4, \dots, \widehat m_{\ell-1}, m_1)\\
		       &=(D,S)(\widehat m_2, \widehat m_3, \dots, \widehat m_{\ell-1}, m_1)\\
  \end{align*}
as required, which also shows feasibility of the sequence $(\widehat m_2,
\widehat m_3, \dots, \widehat m_{\ell-1}, m_1)$ for $(D,S)$.
\end{proof}

\begin{lemma}
\label{l:swap_iii}
Let $D$ be a diagram of a knot and let $m = \{x, y\}_{\IImin}$, $\widetilde m
  = \{x,z\}_{\IImin}$ be two feasible Reidemeister moves in $D$ where $y
  \neq z$. Assume also
  that $m' = \{w, z\}_{\IImin}$ is feasible in $D(m)$ with $x, y, z, w$
  mutually distinct. Then 
  \begin{enumerate}[(i)]
\item
  $\widetilde m' = \{w, y\}_{\IImin}$ is feasible in $D(\widetilde m)$; and 
\item
 $D(m, m') = D(\widetilde m, \widetilde m')$.
\end{enumerate}
\end{lemma}

\begin{proof}

We will provide a direct argument. However, as an alternative proof, the reader
may want to check that there are combinatorially only four options (up to a
mirror image) for how $x, y, z,$ and $w$ can be arranged in the diagram, each of them
yielding the required conclusion; see~\cref{f:four_options}.

\begin{figure}
\begin{center}
  \includegraphics{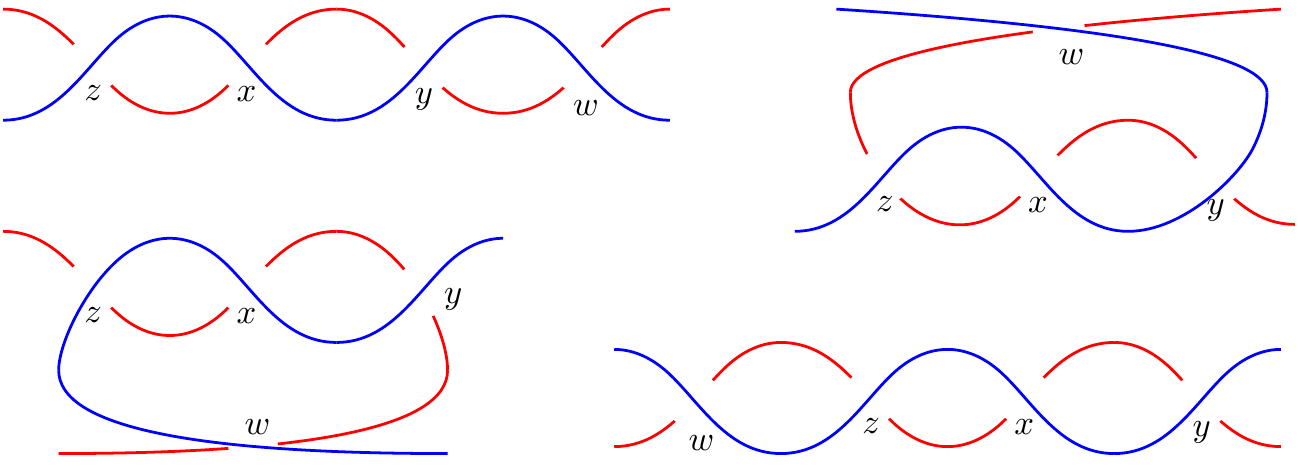}
  \caption{Four options how may $x, y, z$ and $w$ be arranged.}
  \label{f:four_options}
\end{center}
 \end{figure}

Let $\alpha$ and $\beta$ be the edges connecting $x$ and $y$ in $D$ along which
the move $m$ is performed. Similarly, let $\widetilde \alpha$ and $\widetilde
\beta$ be the edges connecting $x$ an $z$ along which the move $\widetilde m$ is
performed. Let us also assume that $\alpha$ and $\widetilde \alpha$ enter $x$ as
overpasses and $\beta$ and $\widetilde \beta$ as underpasses. We also know that
$\alpha, \beta, \widetilde \alpha$ and $\widetilde \beta$ are four distinct
edges because $y \neq z$. Finally, because the `bigon' $B$ between $\alpha$ and
$\beta$ is empty as well as the `bigon' $\widetilde B$ between $\widetilde
\alpha$ and $\widetilde \beta$ is empty, and those two bigons touch at $x$,
there is a unique way, up to an isotopy and a mirror image, how this portion of
diagram may look like; see~\cref{f:double_bigon}.

\begin{figure}
  \begin{center}
    \includegraphics{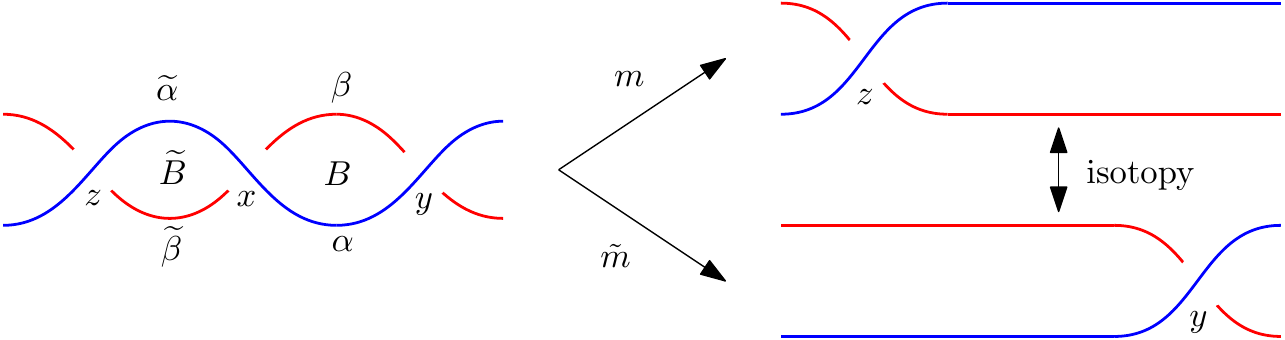}
  \end{center}
  \caption{Left: Portion of the
    diagram $D$ around $x$, $y$ and $z$. Right: results after applying $m$ and
    $\widetilde m$.}
  \label{f:double_bigon}
\end{figure}

Now, when we compare $D(m)$ and $D(\widetilde m)$ we get the same diagram up to
isotopy and relabelling $z \leftrightarrow y$. This in particular gives (i):
if $m' = \{w, z\}_{\IImin}$ is feasible in $D(m)$, then after relabelling, we get
that $\widetilde m' = \{w, y\}_{\IImin}$ is feasible in $D(\widetilde m)$. 

In addition, because of relabelling, performing $m'$ in $D(m)$ and $\widetilde
m'$ in $D(\widetilde m)$ yields the same diagrams $D(m,m')$ and $D(\widetilde m,
\widetilde m')$ with respect to our combinatorial description of diagrams.
Indeed, if $D$ contains only four crossings $x$, $y$, $z$ and $w$, then $D(m,m')
= D(\widetilde m, \widetilde m')$ is the diagram without crossings. If $D$
contains more than four crossings, then we know that the sets of crossings is
the same in $D(m,m')$ and $D(\widetilde m, \widetilde m')$. Thus it remains to
check that we also get exactly same edges (that is we get the same combinatorial
description of edges). Let $\varepsilon' = (v_1, v_2, n_1, n_2)$ be an edge in
$D(m,m')$. Its pre-image before performing $m'$ is an arc $\varepsilon$ in
$D(m)$ connecting $v_1$ and $v_2$ entering $v_1$ in strand $n_1$ and entering
$v_2$ in strand $n_2$, possibly passing through $z$ and $w$ but not through
other crossings of $D(m)$. Because of isotopy between $D(m)$ and $D(\widetilde
m)$, this arc is isotoped to an arc $\widetilde \varepsilon$ in $D(\widetilde
m)$ again connecting $v_1$ and $v_2$ entering $v_1$ in strand $n_1$ and entering
$v_2$ in strand $n_2$, possibly passing through $y$ and $w$ but not through
other crossings of $D(\widetilde m)$. After performing the move $\widetilde m'$,
we get the edge $(v_1, v_2, n_1, n_2)$ in $D(\widetilde m, \widetilde m')$ out
of this arc. Analogously, by reverting this argument, an edge in $D(\widetilde
m, \widetilde m')$ induces the same edge in $D(m,m')$.
\end{proof}
Now, we have all tools to prove~\cref{t:swappable}. For reader's
convenience, we recall the statement.

\begin{tswappable}
  Let $D$ be a diagram of a knot and $S$ be a set of crossings in $D$. Let $(m_1,
  \dots, m_\ell)$ be a feasible sequence of Reidemeister moves for $(D,S)$, inducing
  an untangling of $D$ with defect $k$.
  Let $\widetilde m$ be a greedy move in $D$
  with respect to $S$. Then there is a feasible sequence of
  Reidemeister moves for $(D,S)$ starting with $\widetilde m$ and inducing an
    untangling of $D$ with defect $k$.
\end{tswappable}

\begin{proof}[Proof of~\cref{t:swappable}]
  First, let us assume that $\widetilde m = m_j$ for some $j \in [\ell]$. If $j
  =1$, there is nothing to prove, thus we may assume $j \geq 2$. Then, by
  using~\cref{c:swap_i} on the sequence $(m_1, \dots, m_j)$, we get a sequence
  of moves $(m_j, \widehat m_1, \dots, \widehat m_{j-1})$ feasible for
  $(D,S)$. By item~(i) of~\cref{c:swap_i}, the sequence $(m_j, \widehat m_1,
  \dots, \widehat m_{j-1}, m_{j+1}, \dots, m_{\ell})$ is also feasible for
  $(D,S)$ and induces an untangling of $D$. By item~(ii) of~\cref{c:swap_i}
  and~\cref{l:defect}, the defect of this sequence equals $k$, and that is what
  we need.

  Thus it remains to consider the case where $\widetilde m \neq m_j$ for every
  $j \in [\ell]$. Because, $\widetilde m$ is greedy, it is in particular a
  $\IImin$ move, say $\widetilde m = \{x,z\}_{\IImin}$. Because the final
  diagram $D(m_1, \dots, m_\ell)$ contains no crossings, the crossings $x$ and
  $z$ have to be removed by some moves in the sequence $(m_1, \dots,
  m_\ell)$. Say that a move $m_i$ removes $x$ and $m_j$ removes $z$. Since we
  assumed that $\widetilde m$ is greedy with respect to $S$, we get $x,z \notin
  S$. In particular $x \notin S(m_1, \dots, m_{i-1})$ by~\cref{l:Smulti} ($x$
  never belongs to any $X(m_{i'})$ for $i' < i$ as it is not removed earlier),
  thus $m_i$ has to be greedy move in $D(m_1, \dots, m_{i-1})$ with respect to
  $S(m_1, \dots, m_{i-1})$. Similarly, $m_j$ is greedy in $D(m_1, \dots,
  m_{j-1})$ with respect to $S(m_1, \dots, m_{j-1})$. Because we assume that
  $\widetilde m$ is not in the sequence $(m_1, \dots, m_{\ell})$, we get $i \neq
  j$. Without loss of generality, let us assume $i < j$. Let $y$ and $w$ be such
  that $m_i = \{x,y\}_{\IImin}$ and $m_j = \{w,z\}_{\IImin}$. As these moves are
  greedy, we get $y, w \not \in S$ (again using~\cref{l:Smulti}).

  Let $D' := D(m_1, \dots, m_{i-1})$ and $S' := S(m_1, \dots,
  m_{i-1})$. By~\cref{c:swap_ii}, applied to $D'$, $S'$ and the sequence $(m_i,
  \dots, m_j)$ feasible for $(D', S')$, we get another sequence $(\widehat
  m_{i+1}, \dots, \widehat m_{j-1}, m_i)$ feasible for $(D',S')$. By item~(i)
  of~\cref{c:swap_ii}, we get that $(\widehat m_{i+1}, \dots, \widehat m_{j-1},
  m_i, m_j, \dots, m_\ell)$ is also feasible for $(D', S')$.

  Next, let $D'' := D'(\widehat m_{i+1}, \dots, \widehat m_{j-1})$ and $S'' :=
  S'(\widehat m_{i+1}, \dots, \widehat m_{j-1})$. Because $x, z \notin S$, the
  edges between $x$ and $z$ are not used to perform \Ipl or \IIpl moves and the
  moves $m_1, \dots, m_{i-1}, \widehat m_{i+1}, \dots, \widehat m_{j-1}$ are
  performed outside of the private ball $B(\widetilde{m})$. By~\cref{l:Smulti},
  $x,z \notin S''$ and $\widetilde{m}$ is feasible in
  $(D'',S'')$. By~\cref{l:swap_iii} used in $D''$, we get that $\widetilde m' :=
  \{w,y\}_{\IImin}$ is feasible in $(D'',S'')(\widetilde m)$ and $(D'',S'')(m_i,
  m_j) = (D'',S'')(\widetilde m, \widetilde m')$. Altogether, by expanding $D''$
  and $D'$
  \begin{equation}
\label{e:modified_seq}
  (m_1, \dots, m_{i-1}, \widehat m_{i+1}, \dots, \widehat m_{j-1}, \widetilde
  m, \widetilde m', m_{j+1}, \dots, m_{\ell})
  \end{equation}
 is an untangling of $D$, feasible for $(D,S)$. We also get that the defect of
 this untangling is equal to $k$ by~\cref{l:defect} and item~(ii)
 of~\cref{c:swap_ii}, when we used it. Note that all the moves $m_i$, $m_j$,
 $\widetilde m$ and $\widetilde m'$ are $\IImin$ moves, thus they do not
 contribute to the weight.

  The sequence~\eqref{e:modified_seq} is not the desired sequence yet, because it does not start with
  $\widetilde m$. However, it contains $\widetilde m$, thus we can further
  modify this sequence to the desired sequence starting with $\widetilde m$ as
  in the first paragraph of this proof.
\end{proof}


\section{W[P]-hardness}
\label{s:hardness}

\subsection{Minimum axiom set}
We prove the W[P]-hardness in Theorem~\ref{t:main} by an FPT-reduction
from \sctt{Minimum axiom set} defined below. 
It is well known
that~\sctt{Minimum axiom set} is W[P]-hard;
see~\cite[Exercise~3.20]{flum-grohe06} or
\cite[Lemma~25.1.3]{fund_param_complexity} (however, let us recall that our
definition of FPT-reduction is consistent with~\cite{flum-grohe06}).

\begin{problem*}[\sctt{Minimum axiom set}] \
	
	\begin{tabular}{ll}
		\sctt{Input} & A finite set $S$, and \\ &a finite set
		$\RR$ which consists of pairs of form $(T, t)$ where $T \subseteq
		S$ and $t \in S$.\\
		\sctt{Parameter} & $k$. \\
		\sctt{Question} & Does there exists a subset $S_0 \subseteq S$ of
		size $k$ such that if we define inductively $S_i$\\ &to be 
		the union of $S_{i-1}$ and all $t \in S$ such that there is $T \subseteq
		S_{i-1}$ with $(T, t) \in \RR$, \\
		&then $\bigcup_{i=1}^{\infty} S_i = S$?
	\end{tabular}
\end{problem*}

The problem above deserves a brief explanation. The elements of $S$ are called
\emph{sentences} and the elements of $\RR$ are relations. A relation $(T,t)$
with $T = \{t_1, \dots, t_m\}$ should be understood as an implication
\[
t_1 \wedge t_2 \wedge \cdots \wedge t_m \Rightarrow t.
\]
Given a set $S_0 \subseteq S$, let us define the \emph{consequences} of $S_0$ as
$c(S_0) := \bigcup_{i=1}^{\infty} S_i$ where $S_i$ is defined as in the
statement of the problem. Intuitively, $c(S_0)$ consists of all sentences that
can be deduced from $S_0$ via the \emph{relations} (implications) in $\RR$. As we work
with finite sets, $c(S_0) = S_i$ for some high enough $i$ (it is sufficient to
take $i = |S|$). A set $A$ is a set of \emph{axioms} if $c(A) = S$. Therefore,
the goal of the minimum axiom set problem is to determine whether there is a set
of axioms of size $k$. Note that the axiom sets are upward-closed: If $A$ is an
axiom set and $A \subseteq A' \subseteq S$, then $A'$ is an axiom set as well.

The following boosting lemma is very useful in our reduction. 

\begin{lemma}[Boosting lemma]
	\label{l:boosting}
	Let $(S, \RR)$ be an input of the minimum axiom set problem (ignoring
        the parameter for the moment). Let $\mu \colon S \to \Z$ be a
        non-negative function. Given $U \subseteq S$, let $\mu(U) = \sum_{s \in
          U} \mu(s)$.\footnote{This means that $\mu$ is in fact a measure on
        $S$.}Assume that $\mu(U) \ge 1$ for all $U$ such that $S \setminus U$ is
        not an axiom set. (Equivalently, $U$ meets every axiom set.) Then
        $\mu(S) \geq k^*$ where $k^*$ is the size of a minimum axiom set.
\end{lemma}

In our reduction, we will need to check that $\mu(S) \geq k^*$ for a suitable
function $\mu$. However, it will be much easier to verify only $\mu(U) \ge 1$
for suitable sets $U$.

\begin{proof}
	Let $Z := \{s \in S\colon \mu(s) = 0\}$ be the zero set of $\mu$. Then
        $\mu(Z) = 0$, thus $S \setminus Z$ is a set of axioms by the
        assumptions.  This gives $|S \setminus Z| \geq k^*$ and, in addition,
        $\mu(S) = \mu (S \setminus Z) \geq |S \setminus Z| \geq k^*$ because
        $\mu(s) \geq 1$ for every $s \in S \setminus Z$.
\end{proof}

\subsection{Construction of the reduction}
Our aim is to show that there is an FPT-reduction from the \sctt{Minimum axiom
	set} to 
\sctt{Unknotting via defect}.

\paragraph{Preprocessing.} Let $(S, \mathcal{R}, k)$ be an instance of the
\sctt{Minimum axiom set}. We first perform a
preprocessing on $(S, \mathcal{R}, k)$. We repeatedly apply the following reduction rules:
\begin{enumerate}[(i)]
	\item If $(\emptyset, t) \in \RR$, then we remove $t$ from $S$ and we replace 
	all relations $(T', t')$ with $(T' \setminus \{t\}, t')$  if $t \neq t'$ while we
	remove all relations of the form $(T, t)$ from $\RR$.
	\item If there is $s \in S$ such that $\RR$ contains no relation of a form
	$(T, s)$, then we replace all relations $(T', t')$ with $(T' \setminus
	\{s\}, t')$, we remove $s$ from $S$ and we decrease $k$ by $1$.
	\item If there is $(T, t) \in \RR$ with $t \in T$, then we remove this relation
	$(T,t)$ from $\RR$.
\end{enumerate}

It is easy to check that the preprocessing can be done in polynomial time. In
addition, it yields an equivalent instance: For (i) we easily see that such $t$
belongs to $S_1$ starting with arbitrary $S_0$.  Thus $t \in c(S_0)$ for
an arbitrary $S_0$. Consequently, if we want to deduce that $t' \in c(S_0)$ via
$(T', t')$, it is sufficient to use $(T' \setminus\{t\}, t')$ as we already
know that $t \in c(S_0)$. We also may remove all relations $(T,t)$ from $\RR$
as we already know that $t \in c(S_0)$. For (ii), we see that such $s$ can
never be deduced from the relations, thus such $s$ must be in every axiom set.
Equivalently, if we remove $s$ from the instance, then we need to find an axiom
set of size $k-1$. For (iii), we see that such a relation $(T, t)$ never yields
a new sentence in $c(S_0)$ for arbitrary $S_0$.

Altogether, this means that this preprocessing is an FPT-reduction. Note that in
(ii) we decrease the parameter. This is consistent with the second item of our
definition of FPT-reduction by setting $g'(k) = k$.

\paragraph{Doubling the instance.} Now let $(S, \RR, k)$ be a preprocessed instance of 
the \sctt{Minimum axiom set}. We will need the following doubled instance: Let
$\widehat S := \{\widehat s\colon s \in S\}$ be a copy of $S$ (formally
speaking, $\widehat s$ could be identified, for example with the pair $(s, 0)$
assuming that no such pair is in $S$). Given $T = \{t_1,\dots, t_m\}\subseteq
S$, let $\widehat T:= \{\widehat t_1, \dots, \widehat t_m\}$. Then we define
$\widehat \RR := \{(\widehat T, \widehat{t}) \ \colon (T, t) \in \RR\}$. Then
$(S \cup \widehat S, \RR \cup \widehat \RR, 2k)$ is a \emph{double} of the
instance $(S, \RR, k)$.

\begin{observation}
	\label{o:double}
	The pair $(S, \RR)$ admits an axiom set of size $k$ if and only if its
        double $(S \cup \widehat S, \RR \cup \widehat \RR)$ admits an axiom set
        of size $2k$.
\end{observation}
\begin{proof}
	If $A$ is an axiom set of size $k$ for $(S, \RR)$, then $A \cup \widehat
        A$ is an axiom set of size $2k$ for $(S \cup \widehat S, \RR \cup
        \widehat \RR)$. On the other hand, if $\bar A$ is an axiom set of size
        $2k$ for $(S \cup \widehat S, \RR \cup \widehat \RR)$, let $B := \bar A
        \cap S$ and $\widehat C:= \bar A \cap \widehat S$. Then both $B$ and $C$
        (obtained by removing the hats) are axiom sets for $(S, \RR)$.  One of
        them has size at most $k$ (and it can be extended to size $k$ if
        needed).
\end{proof}
\paragraph{Brunnians.}
A \emph{Brunnian link} is a nontrivial link that
becomes trivial whenever one of the link components is removed. We will use the
following well known construction of a Brunnian link with $\ell \geq 2$
components. We take an untangled unknot and we interlace it with two
`neighboring' unknots as in Figure~\ref{f:brunnian}, left. We repeat this
$\ell$-times and we get a Brunnian link with $\ell$ components as in
Figure~\ref{f:brunnian}, right.

\begin{figure}
	\begin{center}
	  \includegraphics[page=1]{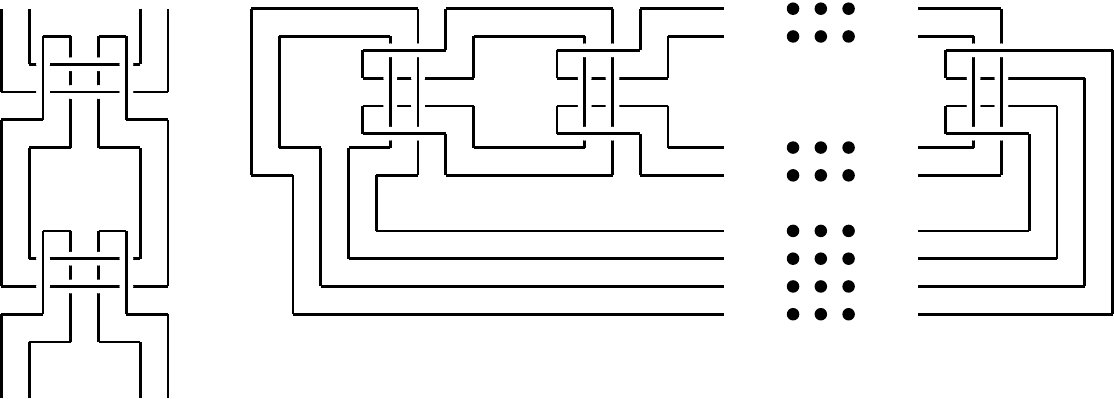}
	\end{center}
	\caption{A Brunnian link}
	\label{f:brunnian}
\end{figure}

\paragraph{Gadgets.} From now on let us assume that $(S, \RR, k)$ is
a preprocessed instance of the \sctt{Minimum axiom set} and $(S \cup \widehat S,
\RR \cup \widehat \RR, 2k)$ is its double. Our aim is to build a diagram $D(S,
\RR)$ such that $D(S, \RR)$ untangles with defect $2k$ if and only if $(S, \RR)$
has an axiom set of size $k$. We will build $D(S, \RR)$ using several
gadgets. Formally speaking, gadgets will be maps of a form $G\colon I_1 \sqcup
\cdots \sqcup I_h \to \R^2$ where $I_1 \sqcup \cdots \sqcup I_h$ stands for a
disjoint union of compact intervals.  We will work with them in a same way as
with diagrams. In particular, we assume the same transversality assumptions on
crossings as for diagrams and we have the same convention with respect to
underpasses and overpasses. We also extend the notion of arc to this setting: It
is a set $G(A)$, where $A$ is a closed subinterval of one of the intervals in
$I_1 \sqcup \cdots \sqcup I_h$.

\paragraph{Sentence gadget.} 
For each sentence $s \in S$ we define a sentence gadget $G(s)$ as follows. (We
will also create an analogous gadget $G(\widehat s)$ for $\widehat s \in
\widehat S$ which we specify after describing $G(s)$.) We consider all relations
$R \in \RR$ of the form $(T,s)$. Let $\ell = \ell(s)$ be the number of such
relations and we order these relations as $R_1(s), \dots, R_{\ell}(s)$, or
simply as $R_1, \dots, R_{\ell}$ if $s$ is clear from the context (which is the
case now).  Note that $\ell \geq 1$ due to preprocessing.

Now we take our Brunnian link with $\ell + 1$ components, which we denote by
$C_0(s), \dots, C_{\ell}(s)$, in the order along the Brunnians. We disconnect
each component of this Brunnian link in an arc touching the outer face and we
double each such disconnected component. See Figure~\ref{f:sentence_gadget},
left. For the further description of the construction, we assume that our
construction is rotated exactly as in the figure. (If $\ell \neq 2$, then we
just modify the length of $C_0(s)$ and insert more or less components $C_i(s)$
in the same way as $C_1(s)$ and $C_2(s)$ are inserted for $\ell = 2$.) Now we
essentially have a collection of interlacing arcs where each former component of
Brunnians yields a pair of parallel arcs with four loose ends.

We merge each pair of arcs, coming from $C_i(s)$ into a single arc $\gamma_i(s)$
as follows; see Figure~\ref{f:sentence_gadget}, right. For the pair coming from
$C_i(s)$ for $i \neq 0$, we connect the bottom loose ends by a straight segment
up to isotopy. In Figure~\ref{f:sentence_gadget}, right, we have isotoped the
figure a bit which will be useful in further steps of the construction, and we
call this subarc the \emph{head} of $\gamma_i(s)$.  For the pair coming from
$C_0(s)$, we first cross the arcs next to the top loose ends as well as next to
the bottom loose ends as in the figure.  Then we connect the bottom loose
ends. Note that if we remove $\gamma_1(s), \dots, \gamma_l(s)$ from the figure,
then the crossings on $\gamma_0(s)$ can be removed by a \IImin move. This
finishes the construction of $G(s)$. The gadget $G(\widehat s)$ is a mirror
image of $G(s)$ along the vertical line ($y$-axis).

\begin{figure}
	\begin{center}
	  \includegraphics[page=2,scale=.95]{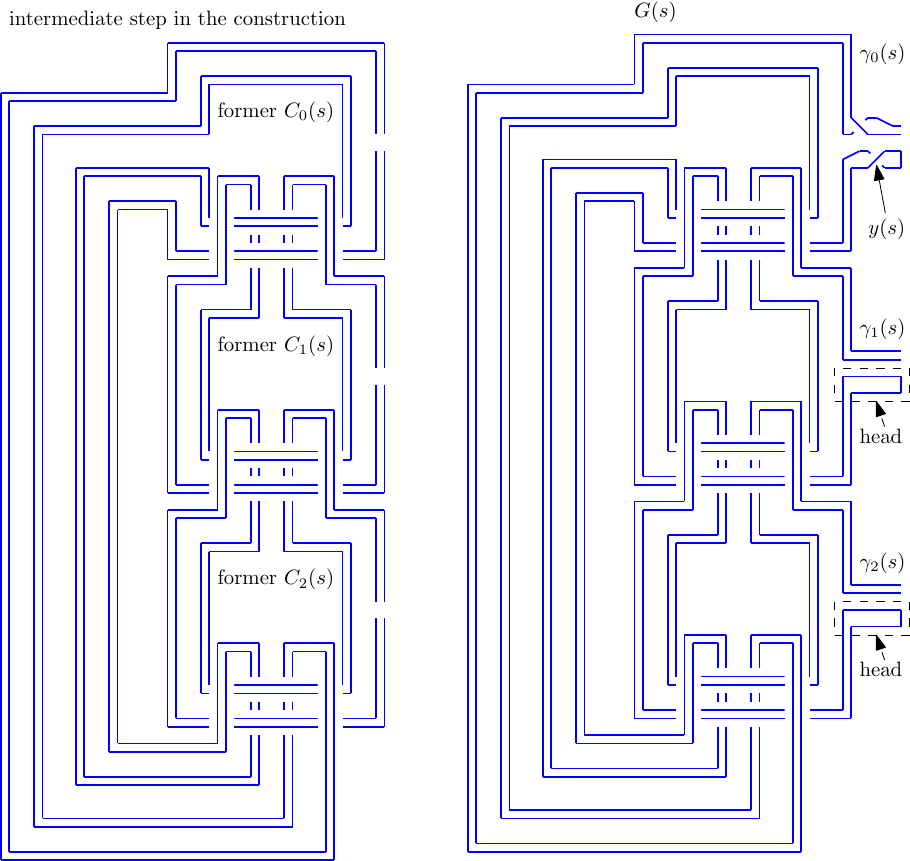}
	\end{center}
	\caption{The sentence gadget with $\ell = 2$.}
	\label{f:sentence_gadget}
\end{figure}

For later use, we also need some notation for the crossings in the sentence
gadget. When working with the lower index of the arcs $\gamma_i(s)$, we will use
arithmetic modulo $\ell + 1$.  In particular, $\gamma_{\ell+1}(s)$ coincides
with $\gamma_0(s)$. Now, observe that $\gamma_i(s)$ and $\gamma_{i+1}(s)$ have
$32$ crossings for $i \in \{0, \dots, \ell\}$. We will denote these crossings
$x^{1}_{i}(s), \dots, x^{32}_{i}(s)$ one by one from left to right and then from
top to down; see Figure~\ref{f:crossings_sentence_gadget}.\footnote{Perhaps more
natural notation would be to use $x^{1}_{i+1/2}(s)$ but we want to avoid this as
the notation is already complicated enough.} In addition, by $y(s)$ we denote
the bottom self crossings of $\gamma_0(s)$ as on Figure~\ref{f:sentence_gadget}.
We use an analogous notation also for $\widehat s$.

\begin{figure}
	\begin{center}
		\includegraphics[page=3]{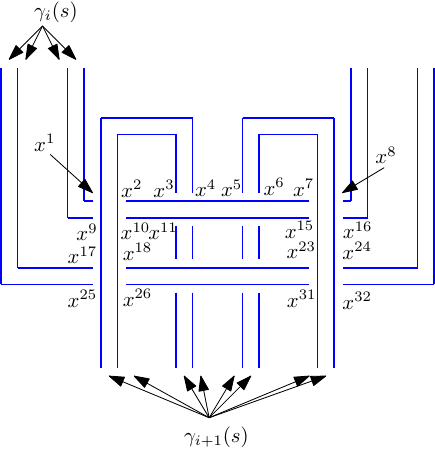}
	\end{center}
	\caption{Crossings between $\gamma_i(s)$ and $\gamma_{i+1}(s)$. For
          simplicity of the figure, we write $x^j$ instead of $x^j_i(s)$. For
          the same reason, we also leave out a few obvious labels.}
	\label{f:crossings_sentence_gadget}
\end{figure}

\paragraph{Merging gadget.} 
We will also need a merging gadget depicted on Figure~\ref{f:merging_gadget}.
We order all sentences as $s_1, \dots, s_m$. The merging gadget consists of
subgadgets $M(s_1), M(\widehat s_1), M(s_2), \dots, M(s_m), M(\widehat s_m)$
separated by the dotted lines in the figure. Each subgadget has several loose
ends. Two or four of them serve for connecting it to other subgadgets. The
remaining ones come in pairs and the number of pairs equals to $\ell(s) +
1$. (Recall that this is the number of components of the sentence gadget
$G(s)$.) These pairs of loose ends will be eventually connected to the loose
ends of $\gamma_0(s), \dots, \gamma_{\ell(s)}(s)$ in the top-down order.

\begin{figure}
	\begin{center}
	  \includegraphics{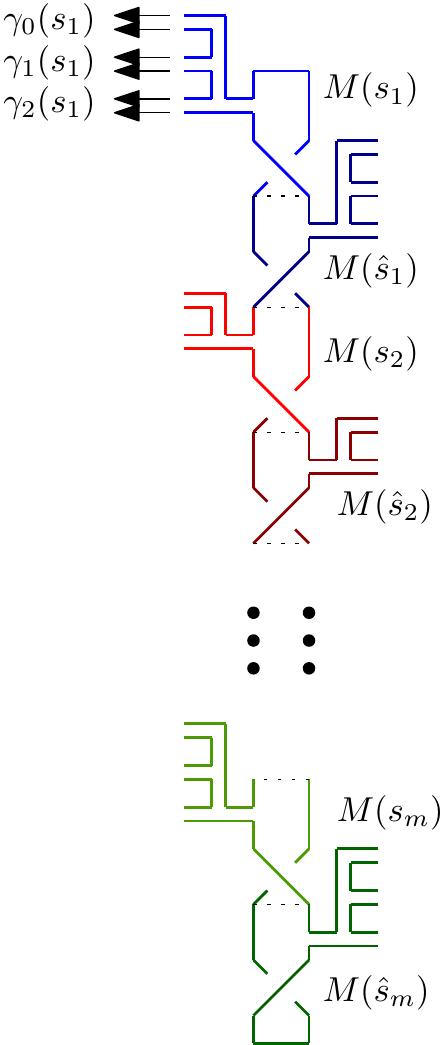}
	\end{center}
	\caption{The merging gadget.}
	\label{f:merging_gadget}
\end{figure}

\paragraph{Interconnecting the gadgets.}
Now, we describe how to interconnect the gadgets.

We place the sentence gadgets $G(s_1), \dots, G(s_m)$ to the left of the merging
gadget $M$ in a top-down ordering. Similarly, we place $G(\widehat s_1), \dots,
G(\widehat s_m)$ to the right, again in a top-down ordering. First, for any
$\bar s \in S \cup \widehat S$ and $i \in \{0, 2, 3, \dots, \ell(\bar s)\}$
(that is, $i \neq 1$), we connect $\gamma_i(\bar s)$ with the $(i+1)$th pair of
loose ends of $M(\bar s)$ by a pair of parallel arcs as directly as possible
without introducing new crossings; see
Figure~\ref{f:connecting_gamma_i}.

\begin{figure}
	\begin{center}
	  \includegraphics[page=3]{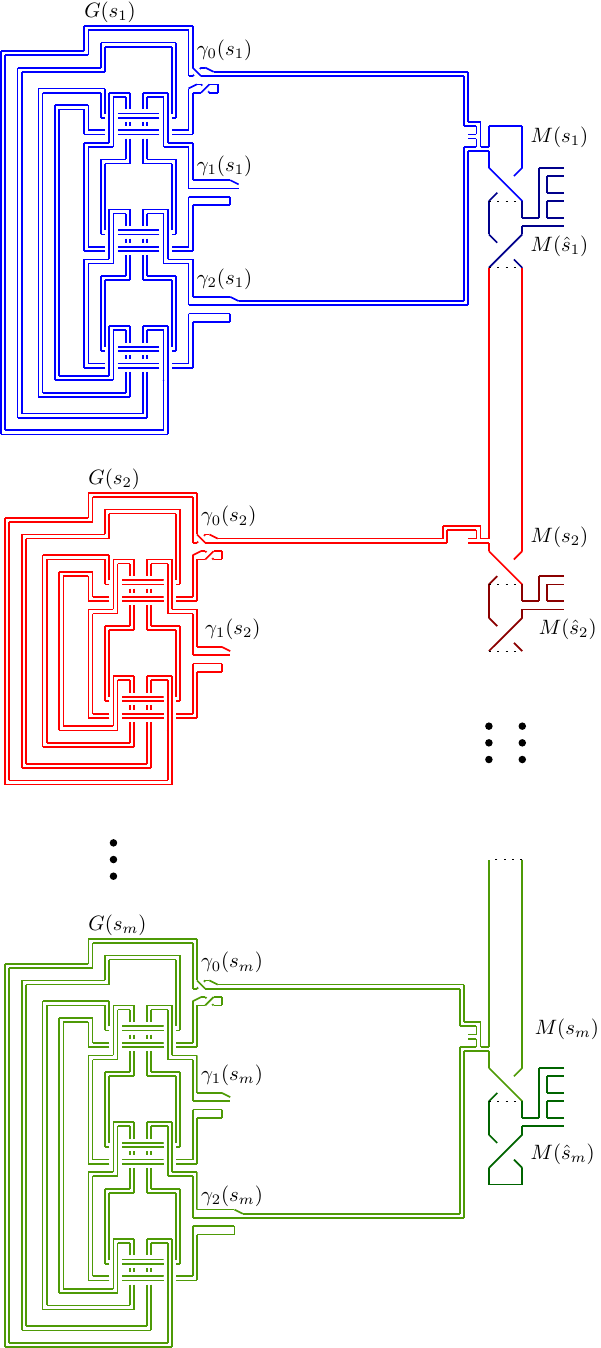}
		\caption{The first step of interconnecting the gadgets:
                  Connecting $M(s)$ to the arcs $\gamma_i(s)$ for $i \neq
                  1$. The mirror symmetric part for $\widehat s$ is not
                  displayed.}
		\label{f:connecting_gamma_i}
	\end{center}
\end{figure}

Now we want to connect $\gamma_1(\bar s)$ to the second pair of loose ends of
$M(\bar s)$. For simplicity, we describe this in the case $\bar s = s \in
S$. The case $\bar s \in \widehat S$ is mirror symmetric. We pull a pair of
parallel arcs from $\gamma_1(s)$ with the aim to reach $M(s)$ while obeying the
following rules:
\begin{enumerate}[(R1)]
	\item We are not allowed to cross the merging gadget or the sentence gadgets
	(except the case described in the third rule below). We keep the newly
	introduced arcs on the left side from the merging gadget.
	\item We are allowed to cross other pairs of parallel arcs introduced
	previously (when connecting $\gamma_i(s')$ to $M(s')$ for some $i$ and
	$s'$). However, if we cross such a pair, we require that all four newly
	introduced crossings are resolved simultaneously (for example the new pair
	of parallel arcs  is always above the older one). We even allow a self
	crossing of the newly introduced parallel arcs (but we again require that
	four newly introduced crossings are resolved simultaneously).
	\item For every relation $R = (T, s')$ where $s \in T$, let $R = R_i(s')$. 
	We interlace the newly introduced parallel arcs with $\gamma_i(s')$ as in
	Figure~\ref{f:interlacing_parallel}. For further reference, two important
	crossings are denoted $z^1_i(s',s)$ and $z^2_i(s',s)$ as in the figure.
	\item The total number of crossing is of polynomial size in the size of our
	instance $(S, \RR)$.
\end{enumerate}

\begin{figure}
	\begin{center}
	  \includegraphics[page=4]{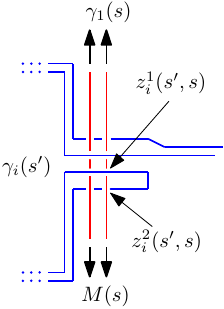}
		\caption{Interlacing the parallel arcs pulled out of
                  $\gamma_1(s)$ with $\gamma_j(s')$.}
		\label{f:interlacing_parallel}
	\end{center}
\end{figure}

As we have quite some freedom how to perform the construction obeying the rules
above, the resulting construction is not unique. An example how to get this
construction systematically is sketched on
Figure~\ref{f:connecting_gamma_1}: 
The newly introduced pair of arcs is pulled little bit to the right, then we
continue down towards the level of $G(s_m)$, then up towards the level of $G(s_1)$,
and then back down to the original position. On this way we make a detour
towards $\gamma_j(s')$ whenever we need to apply the rule (R3). Finally, we
connect the parallel arcs to $M(s)$ without any further detour. 

\begin{figure}
	\begin{center}
	  \includegraphics[page=2]{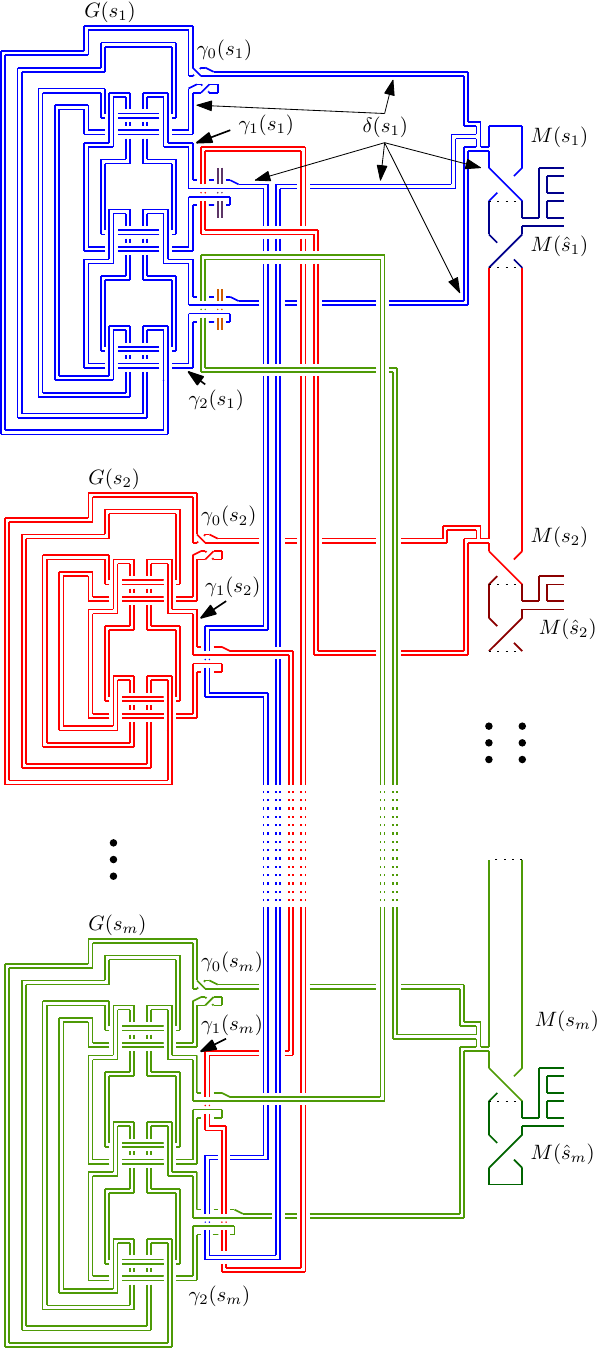}
		\caption{The second step of interconnecting the gadgets:
                  Connecting $M(s)$ to the arc $\gamma_1(s)$. The mirror
                  symmetric part of $\widehat s$ is again not displayed.  The
                  arcs $\delta(s)$ are distinguished by colors.}
		\label{f:connecting_gamma_1}
	\end{center}
\end{figure}
For a later analysis, note that for any $s \in S$, the merging gadget $M(s)$,
the sentence gadget $G(s)$ and the newly introduced pairs of parallel arcs
connecting $\gamma_i(s)$ and $M(s)$ altogether form a single arc. We will
denote this arc $\delta(s)$. In particular each arc $\gamma_i(s)$ is a subarc
of $\delta(s)$. As usual, we also use the analogous notation for $\widehat s$.

Note that after all these steps, we obtain a diagram of a knot, which is our
desired diagram---it is the union of arcs $\delta(\bar s)$ for $\bar s \in S
\cup \widehat S$. We denote this diagram $D(S, \RR)$.
A little bit more of investigation reveals that $D(S, \RR)$ is
actually a diagram of an unknot but we do not really need this.

\subsection{A small set of axioms implies small defect}
\label{ss:hardness_easier}

Let us assume that $D(S, \RR)$ is a diagram coming from a preprocessed instance
$(S, \RR, k)$. The aim of this subsection is to show that if there exists a set
$A$ of axioms of size $k$ for $(S, \RR)$, then there exists an untangling of
$D(S, \RR)$ with defect $2k$.  In other words, we will show that if $(S, \RR,
k)$ is a (preprocessed) \sctt{Yes} instance of \sctt{Minimum Axiom Set}, then
$(D(S, \RR), 2k)$ is a \sctt{Yes} instance of \sctt{Unknotting via defect}.

We will provide an untangling $\D$ that uses only $\Imin$ and $\IImin$ moves. By
Lemma~\ref{l:defect}, the defect of $\D$ equals to the number of $\Imin$ moves.

We first describe the \Imin moves. For every $a \in A$, we unscrew the loop next
to $y(a)$ by a \Imin move (see Figure~\ref{f:sentence_gadget} for $y(a)$) as
well as we unscrew the loop next to $y(\widehat a)$. For this, we need $2k$
\Imin moves. It remains to show that the resulting diagram can be untangled
using $\IImin$ moves only.

Next, we use \IImin moves to untangle $\gamma_0(a)$ from the sentence gadget;
see Figure~\ref{f:untangle_gamma0}. After this step, the only remainder of
$\gamma_0(a)$ is a loop next to the other self-crossing of $\gamma_0(a)$.

\begin{figure}
	\begin{center}
		\includegraphics[page=4,scale=.95]{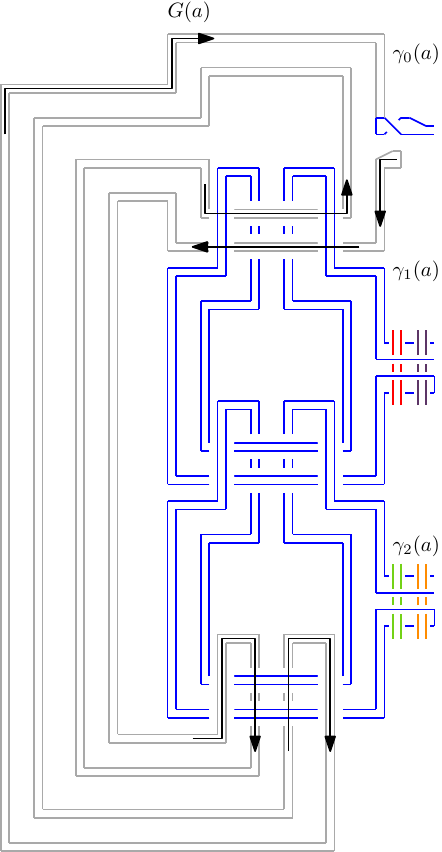}
		\caption{\IImin moves after unscrewing the loop at
                  $y(a)$. Compare with Figures~\ref{f:sentence_gadget}
                  and~\ref{f:interlacing_parallel}.}
		\label{f:untangle_gamma0}
	\end{center}
\end{figure}

After untangling $\gamma_0(a)$, we may untangle $\gamma_1(a)$ from $G(a)$ using
the \IImin moves only; see Figure~\ref{f:untangle_gamma1}. First, we untangle
$\gamma_1(a)$ form $\gamma_2(a)$ (if $\gamma_2(a)$ exists) as suggested in the
left figure. (This can be done so that we first untangle the inner `finger' and
then the outer one.) Then we untangle the head of $\gamma_1(a)$ form the other
$\delta(s)$ as suggested in the right figure.  (Again we first untangle the
inner finger, then the outer one.) Finally, we untangle the remainder all the
way to $M(a)$. Due to our conventions (R1) and (R2) when interconnecting
gadgets, it is possible to do this step via \IImin moves.

\begin{figure}
	\begin{center}
		\includegraphics[page=5]{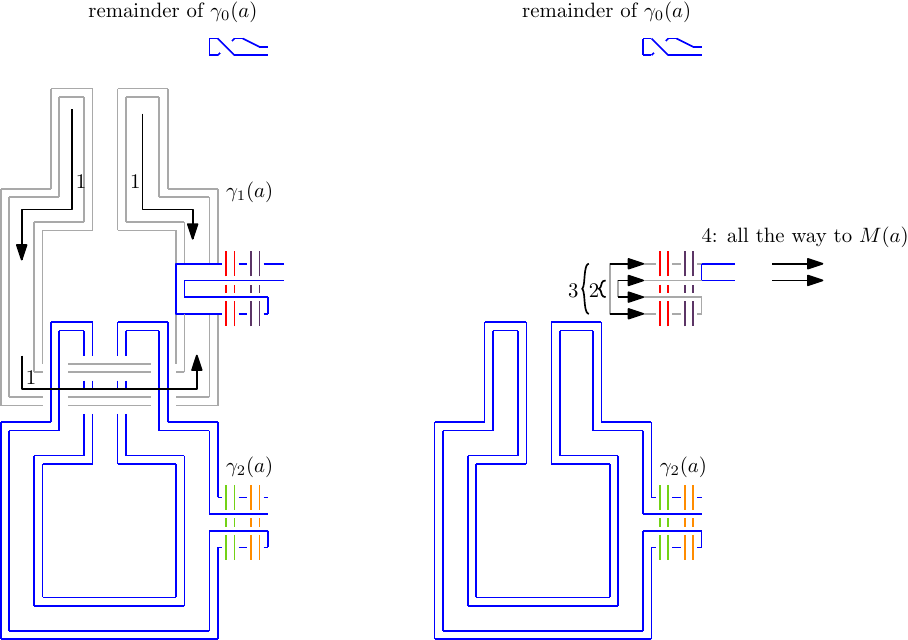}
		\caption{\IImin moves untangling $\gamma_1(a)$ from the gadget $G(a)$.}
		\label{f:untangle_gamma1}
	\end{center}
\end{figure}

After untangling $\gamma_1(a)$ from $G(a)$, we may continue and untangle
$\gamma_2(a)$ the same way, etc. until we untangle $\gamma_{\ell}(a)$. After
these steps the only remainder from $\delta(a)$ is the merging gadget $M(a)$, a
pair of parallel arcs towards the remainder of $\gamma_0(a)$, and this remainder
of $\gamma_0(a)$ itself (a loop next to an original self-crossing of
$\gamma_0(a)$).

Now, we repeatedly perform the following steps. Whenever $R = (T,s)$ is a
relation in $\RR$ such that each $t \in T$ is already almost untangled, then we
want to untangle $s$ from the diagram. We recall that $R = R_i(s)$ for some $i
\in \ell(s)$. Because each $t \in T$ is already almost untangled, the head of
$\gamma_i(s)$ is now empty (it is obvious for axioms in $T$ but we will also
keep the property that any newly almost untangled sentences leave all
heads). This allows us to untangle $\gamma_i(s)$ from $G(s)$ in completely
analogous way as we have untangled $\gamma_0(a)$ on
Figure~\ref{f:untangle_gamma0}. Next, we can untangle $\gamma_{i+1}(s), \dots,
\gamma_{\ell(s)}(s)$ using the same procedure as for $\gamma_1(a)$ on
Figure~\ref{f:untangle_gamma1}. Using the same procedure again, we can also
untangle $\gamma_0(s)$ with the exception that we stop at the two self crossings
of $\gamma_0(s)$; then we remove these self-crossings by a \IImin move, and then
we continue all the way towards $M(s)$. After untangling $\gamma_0(s)$, we
continue to $\gamma_1(s)$ and so on until we reach $\gamma_{i-1}(s)$. At this
moment we say that $s$ is \emph{almost untangled} from the diagram. (One could
argue that $s$ is fully untangled at this moment because even the self-crossings
of $\gamma_0(s)$ have been eliminated.  But we want to use just one notion for
well working recursion.)

We perform completely analogous steps for $\widehat R$ and $\widehat s$, getting
$\widehat s$ untangled as well.

Because $A$ was a set of axioms, every sentence becomes almost untangled after
finitely many steps. Now the diagram looks like as on
Figure~\ref{f:untangle_last_step}. The loop in position of $\gamma_0(s)$ and
$\gamma_0(\widehat s)$ appears if and only if $s \in A$. It is now easy to
finish the untangling using \IImin moves only. We untangle each $M(s)$ and
$M(\widehat s)$ one by one in top-down order. If $s$ is an axiom then the
remaining self-crossing of $\gamma_0(s)$ cancels out with the self-crossing of
$M(s)$, and the same is true for $\widehat s$. If $s$ is not an axiom, then the
self crossings of $M(s)$ and $M(\widehat s)$ cancel out together.

\begin{figure}
	\begin{center}
		\includegraphics[page=5]{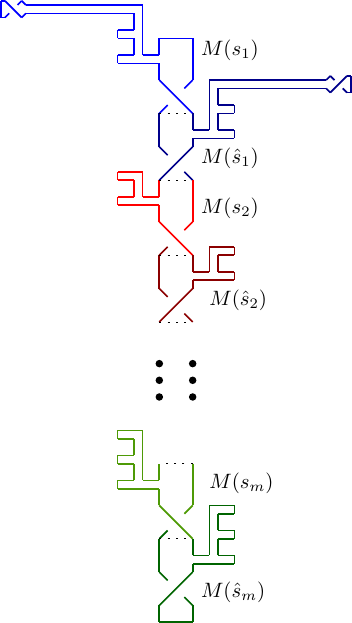}
		\caption{The diagram after untangling every sentence. In this
                  figure $s_1 \in A$ while $s_2, s_m \not\in A$.}
		\label{f:untangle_last_step}
	\end{center}
\end{figure}

\subsection{Big minimum axiom set implies big defect}
\label{ss:hardness_harder}

Let us assume that $D(S, \RR)$ is a diagram coming from a preprocessed instance
$(S, \RR, k)$. Our aim is to show that if the minimum axiom set for $(S, \RR)$
has size at least $k$, then any untangling of $D(S, \RR)$ has defect at least
$2k$.  This implies that if $(S, \RR, k-1)$ is a (preprocessed) \sctt{No}
instance of \sctt{Minimum Axiom Set}, then $(D(S, \RR), 2(k-1))$ is a \sctt{No}
instance of \sctt{Unknotting via defect}. (The shift of the parameter by $1$
appears only because of the convenience of the first formulation with which we
will work.)

By Observation~\ref{o:double} our aim can be equivalently rephrased as follows:
If the minimum axiom set for the double $(S \cup \widehat S, \RR \cup \widehat
\RR)$ has size at least $2k$, then any untangling of $D(S, \RR)$ has defect at
least $2k$.

One of the key tools will be the boosting lemma (Lemma~\ref{l:boosting}) used
for $(S \cup \widehat S, \RR \cup \widehat \RR)$. Given an untangling $\D$ of
$D(S, \RR)$ we will design a function $\mu \colon S \cup \widehat S \to \Z$
satisfying the assumptions of the lemma such that $\defe(\D) \geq \mu(S \cup
\widehat S)$; here we again extend the definition of $\mu$ to all subsets of $S
\cup \widehat S$ as in the statement of the lemma. Then, by the lemma we deduce
that $\defe(\D) \geq 2k$ as we need.

\paragraph{Weights on crossings.} In order to define $\mu$ we will need to
bound the defect via weights on crossings as
in~\cite{demesmay-rieck-sedgwick-tancer21}. However, we need a little bit less,
thus we will use a simpler weight function.

Let $\D$ be an untangling of a diagram $D$. Every crossing in $D$ has to be
eventually removed by a $\Imin$ move or a $\IImin$ move during the untangling
(of course, new crossings may appear but we are not interested in those for the
moment). A $\IImin$ move in the untangling $\D$ is \emph{economical} if it
removes two crossings of $D$ whereas it is \emph{wasteful} if at least on of the
crossings removed by this $\IImin$ move was not originally in $D$ but it has
been introduced during the untangling. We define the weight $w'(x)$ of a
crossing $x$ in $D$ by the following formula\footnote{We use $w'$ in order to
distinguish this weight function from our earlier weight function on the
Reidemeister moves.}
\[
w'(x) = \begin{cases}
0 & \hbox{ if $x$ is removed by an economical $\IImin$ move,}\\
1 & \hbox{ if $x$ is removed by a $\Imin$ move,}\\
2 & \hbox{ if $x$ is removed by a wasteful $\IImin$ move.}\\
\end{cases}
\]

The following lemma is a weaker form of Lemma~7
from~\cite{demesmay-rieck-sedgwick-tancer21}.

\begin{lemma}
	\label{l:def_weights}
	Let $\D$ be an untangling of a diagram $D$. Then
	\[
	\defe(\D) \geq \sum_x w'(x)
	\]
	where the sum is over all crossings $x$ in $D$.
\end{lemma}

For completeness, we provide (almost complete) sketch of the lemma. For a
slightly different full proof we refer
to~\cite{demesmay-rieck-sedgwick-tancer21}.

\begin{proof}[Sketch.]
	Using Lemma~\ref{l:defect} it is sufficient to show that $\sum_m w(m)
        \geq \sum_x w'(x)$ where the first sum is over all Reidemeister moves
        $m$ in $\D$.  If $w'(x) = 1$, then $x$ is removed by a $\Imin$ move $m$,
        and thus we have $w(m) \geq w'(x)$ for this $m$ and $x$. If $w'(x) = 2$,
        then $x$ is removed by a $\IImin$ move $m$ simultaneously with some $y$
        which has been created during the untangling. This $y$ has been created
        either by a $\I^+$ move $m'$ or a $\II^+$ move $m''$. Then we use $w(m')
        \geq w'(x)$ or $\frac 12 w(m'') \geq w'(x)$. Summing all such
        inequalities, possibly with trivial inequalities $c w(m) \geq 0$ where
        $c \in \{1/2, 1\}$ yields the desired result.
\end{proof}

\paragraph{Important crossings} Now we recall that our diagram $D(S, \RR)$ contains several important
crossings, namely the crossings $x^j_i(\bar s)$ for $\bar s \in S \cup \widehat
S$, $i \in \{0, \dots, \ell(\bar s)\}$ and $j \in [32]$; the crossings $y(\bar
s)$ for $\bar s \in S \cup \widehat S$ (see the definition of the sentence
gadget and Figures~\ref{f:crossings_sentence_gadget} and~\ref{f:sentence_gadget}
for the first two types); and the crossings $z^1_i(\bar s, \bar s'), z^2_i(\bar
s, \bar s')$ for $\bar s, \bar s' \in S \cup \widehat S$ and $i \in \{0,\dots,
\ell(\bar s)\}$ such that $R_i(\bar s)$ is of form $(T, \bar s)$, where $\bar s'
\in T$ (see the rule (R3) when interconnecting the gadgets\footnote{Note that
$\bar s$ here corresponds to $s'$ in (R3).} and
Figure~\ref{f:interlacing_parallel}).

Now, for $\bar s \in S \cup \widehat S$ we set 
\begin{align*} 
&X_i(\bar s) := \{x^j_i(\bar s)\colon j \in [32]\} &\hbox{ for } i \in \{0,
\dots, \ell(s)\}; \\
&X(\bar s) := \bigcup_{i=0}^{\ell(\bar s)} X_i(\bar s); & \\
&Z_i(\bar s) :=
\{z_i^j(\bar s, \bar s')\colon j \in [2], R_i(\bar s) = (T, \bar s) \hbox{ with } \bar s' \in T 
\} &\hbox{ for }  i \in [\ell(s)]; \\
&Z(\bar s) := \bigcup_{i=1}^{\ell(\bar s)} Z_i(\bar s); \hbox{ and} & \\
& Q(\bar s) := X(\bar s) \cup \{y(\bar s)\} \cup Z(\bar s).& \\
\end{align*}
Here $Q(\bar s)$ is a set of important crossings assigned to $\bar s$. 
Note that if $\bar s \neq \bar s'$, then $Q(\bar s)$ and $Q(\bar s')$ are disjoint.

\paragraph{Definition of $\mu$.}
From now on let $\D$ be any untangling of $D(S, \RR)$. We define $\mu\colon S
\cup \bar S \to \Z$ as
\[
\mu(\bar s) := \sum_{x \in Q(\bar s)} w'(\bar s)
\]
where the weight $w'$ corresponds to the untangling $\D$, of course. The
function $\mu$ is nonnegative. In order to verify the assumptions of
Lemma~\ref{l:boosting}, we need to show the following.

\begin{proposition}
	\label{p:verify_boosting}
	Let $\bar U \subseteq S \cup \widehat S$ be such that $(S \cup \widehat
        S) \setminus \bar U$ is not a set of axioms. Then $\mu(\bar u) \geq 1$
        for some $\bar u \in \bar U$.
\end{proposition}

Before we prove Proposition~\ref{p:verify_boosting}, let us explain how we
finish our hardness proof (which we already sketched earlier). Assuming
Proposition~\ref{p:verify_boosting}, we obtain $\mu(S \cup \widehat S) \geq 2k$
by Lemma~\ref{l:boosting}. On the other hand, as the sets $Q(\bar s)$ are
pairwise disjoint, we get that $\mu(S \cup \widehat S)$ is at most the sum of
the weights of all crossings in $D(S, \RR)$. Therefore $\defe(\D) \geq \mu(S
\cup \widehat S) \geq 2k$ by Lemma~\ref{l:def_weights}. This is exactly what we
need. Thus it remains to prove Propostion~\ref{p:verify_boosting} which still
requires a bit of effort.

\paragraph{Close neighbors.}
For a proof of Proposition~\ref{p:verify_boosting}, we will need a notion of
close neighbors from~\cite{demesmay-rieck-sedgwick-tancer21}. Let $Q$ be a set
of crossings in a diagram $D$. Let $x$ and $y$ be two crossings of $D$, not
necessarily in $Q$. We say that $x$ and $y$ are \emph{close neighbors} with
respect to $Q$ if $x$ and $y$ can be connected by two arcs $\alpha$ and $\beta$
satisfying the following conditions.
\begin{itemize}
	\item $\alpha$ enters both $x$ and $y$ as an overpass;
	\item $\beta$ enters both $x$ and $y$ as an underpass;
	\item $\alpha$ and $\beta$ may have self-crossings; however neither $x$ nor
	$y$ is in the interior of $\alpha$ or $\beta$;
	\item no point of $Q$ is in the interior of $\alpha$ or $\beta$.
\end{itemize}
Let us remark that what we call close neighbors is a special case of $c$-close
neighbors in~\cite{demesmay-rieck-sedgwick-tancer21} for $c = 0$.

The following lemma is a special case of Lemma~9
in~\cite{demesmay-rieck-sedgwick-tancer21}.

\begin{lemma}
	\label{l:first_removed}
	Let $Q$ be a set of crossings in a diagram $D$. Let us fix an untangling
        of $D$ and let $q$ be the first crossing in $Q$ removed by an economical
        \IImin move (in this untangling). (In case of a tie, the conclusion is
        valid for both first two crossings.)  Then there is a crossing $q'$ in
        $D$ such that $q$ and $q'$ are close neighbors with respect to $D$.
\end{lemma}

We again provide only a sketch while we refer
to~\cite{demesmay-rieck-sedgwick-tancer21} for a full proof.

\begin{proof}[Sketch]
	Because $q$ is removed by an economical \IImin move, the second crossing
        removed by this \IImin move also belongs to $D$. This is the desired
        $q'$. At the moment of removal of $q$ and $q'$, they from a `bigon'
        along which the \IImin move is performed. In particular they are
        connected by two arcs, one entering both as overpass and another one as
        underpass. Now we trace back the `history' of the untangling up to this
        point and each step we `pull back' this pair of arcs. By an analysis of
        the Reidemeister moves, it turns out that we will get $\alpha$ and
        $\beta$ required by the definition of close neighbors after returning to
        $D$ (while tracing back the history).
\end{proof}

\begin{proof}[Proof of~Proposition~\ref{p:verify_boosting}.]
	We know that $(S \cup \widehat S) \setminus \bar U$ is not a set of
        axioms for $(S \cup \widehat S, \RR \cup \widehat \RR)$. By the
        definition of the double of $(S, \RR)$, either $S \setminus \bar U$ is
        not a set of axioms for $(S, \RR)$, or $\widehat S \setminus \bar U$ is
        not a set of axioms for $(\widehat S, \widehat \RR)$.  Without loss of
        generality, we assume the former case, the argument for the latter case
        is symmetric. So we set $U := \bar U \cap S$, and we assume that $S
        \setminus U = S \setminus \bar U$ is not a set of axioms. We will find
        $u \in U$ with $\mu(u) \geq 1$ as required.
	
	First we perform the following simplification step. Assume that there is a
	relation $(T, s) \in \RR$ such that $s \in U$ and $T \subseteq S \setminus
	U$. Let $U' := U \setminus \{s\}$. We observe that $S \setminus U'$ is
	not a set of axioms. Indeed, $T \subseteq S \setminus U'$ because we assume
	that $s \not\in T$ by the third step of preprocessing. This means that $s$
	can be deduced from $S \setminus U'$ which shows that the closures $c(S
	\setminus U)$ and $c(S\setminus U')$ are equal. Thus $S\setminus U'$ is not a
	set of axioms by an assumption that $S \setminus U$ is not a set of axioms.
	Thus, we can use $U'$ instead of $U$ in our further considerations. Once we
	show that $\mu(u) \geq 1$ for some $u \in U'$, then $u \in U$ as well. We
	repeat this simplification step as long as possible and we will perform our
	considerations for the resulting set. With a slight abuse of notation, we
	will still denote this set by $U$. In other words, we assume, without loss of
	generality, that if $(T, s) \in \RR$ and $s \in U$, then $T \cap U \neq
	\emptyset$.

	For contradiction, let us assume that $\mu(u) = 0$ for every $u \in
        U$. This means that $w'(x) = 0$ for every $x \in Q$ where $Q :=
        \bigcup_{u \in U} Q(u)$. By the definition of $w'$, this means that
        every $x \in Q$ is removed by an economical \IImin move. Let $q$ be the
        first crossing in $I$ removed by an economical \IImin move (in case of a
        tie, we take $q$ to be any of the first two removed crossings). By
        Lemma~\ref{l:first_removed}, there is a crossing $q'$ of $D$ such that
        $q$ and $q'$ are close neighbors with respect to $Q$.

	Recall that $Q(u) = X(u) \cup \{y(u)\} \cup Z(u)$. We will distinguish
        between the cases whether $q \in X(u)$ for some $u$, or $q = y(u)$ for
        some $u$, or $q \in Z(u)$ for some $u$. In each case we will derive a
        contradiction with the fact that $q$ and $q'$ are close neighbors.
	
	\bigskip
	
	Case 1: $q \in X(u)$ for some $u \in U$. As $q \in X(u)$ we get that $q
        \in X_i(u)$ for some $i \in \{0, \dots, \ell(u)\}$, that is, $q =
        x_i^j(u)$ for some $j \in [32]$. For each of the $32$ options, we seek
        for a possible $q'$.  A necessary condition, via the definition of close
        neighbors, is that $q$ and $q'$ are connected by two arcs which do not
        contain any point of $Q$ in their interiors. In order to treat all $32$
        options efficiently, we form a graph $G_1$ whose set of vertices is $Q$
        and whose edges $xy$ where $x \in X_i(u)$ and $y \in Q$ correspond to
        arcs in $D$ which do not have any point of $Q$ in the interior; see
        Figure~\ref{f:graph_x}. We allow loops and multiple edges, therefore
        each $x \in X_i(u)$ has degree exactly $4$. An edge in the graph is
        drawn by a solid line if we guarantee that it corresponds to an arc
        which has no crossing in the interior. A dashed line is used if there
        may be some crossing (necessarily outside $Q$) in the interior of the
        corresponding arc. (In other words, we do not guarantee that this line
        is solid.)  It is slightly non-trivial to check what are the other
        vertices of the edges emanating right from the points $x^{24}_i(u)$ and
        $x^{32}_i(u)$. This we will discuss in detail a bit later.  All other
        edges depicted on Figure~\ref{f:graph_x} follow immediately from the
        construction. (For the edges emanating down from $x^{31}_i(u)$ and
        $x^{32}_i(u)$ we do not make any claim on the figure.)
	
	\begin{figure}
		\begin{center}
			\includegraphics{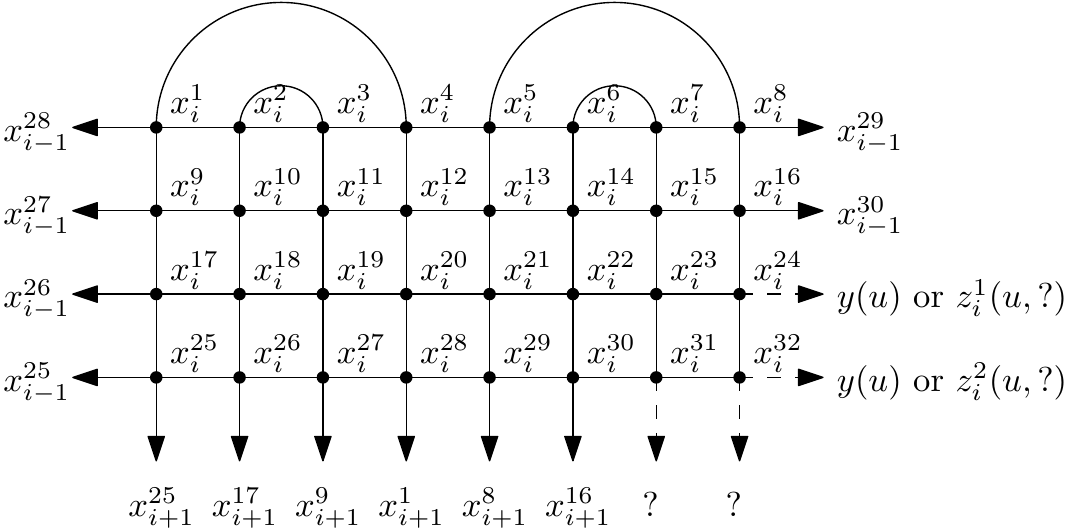}
			\caption{The graph $G_1$. In the figure, we write
                          $x_i^j$ instead of $x_i^j(u)$. Compare with
                          Figures~\ref{f:crossings_sentence_gadget}
                          and~\ref{f:sentence_gadget} for $x$-labels and
                          $y$-labels, and with
                          Figure~\ref{f:interlacing_parallel} for $z$-labels.}
			\label{f:graph_x}
		\end{center}
	\end{figure}

	A necessary (but not sufficient) condition for $q$ and $q'$ to be close
        neighbors is that $q \in X_i(u)$ and (i) either $q' \in Q$ is connected
        to $q$ by at least two edges in $G_1$, or (ii) $q'$ is a crossing of $D$
        not in $Q$ and it belongs to at least two arcs corresponding to the
        edges in the graph. We will rule out both subcases but first let us
        analyze what are the neighbors of $x^{24}_i(u)$ and $x^{32}_i(u)$ on the
        right, as promised.
	
	If $i = 0$, then both $x^{24}_i(u)$ and $x^{32}_i(u)$ head to $y(u)$
        without passing through any crossing. If $i \geq 1$, then we recall that
        $\gamma_i(u)$ corresponds to a relation $R_i(u) = (T, u) \in \RR$. Let
        us also recall that the initial simplification of $U$ implies that $T
        \cap U \neq \emptyset$.\footnote{This is the place where we really use
        our initial assumption that $X \setminus U$ is not an axiom set. After
        the simplification of $U$ it translates into this condition.} Using the
        rule (R3) when interconnecting gadgets, this gives that the right
        neighbor of $x^{24}_i(u)$ is the rightmost $z_i^1(u, s)$ among all $s
        \in T \cap U$ and then the right neighbor of $x^{32}_i(u)$ is $z_i^2(u,
        s)$. In particular there is no double edge between $x^{24}_i(u)$ and
        $x^{32}_i(u)$.

	Now, we rule out subcase (i), that is $q' \in Q$ and $q$ and $q'$ are
        connected by a double edge. We have double edges $x_i^2(u) x_i^3(u)$,
        $x_i^6(u) x_i^7(u)$, and possibly $x_i^{25}(u) x_{i+1}^{25}(u)$, if
        $\ell(u) = 1$, that is, $i -1 = i + 1 \pmod{\ell(u) + 1}$. However, none
        of these cases yields close neighbors as the overpasses/underpasses
        disagree with the definition of close neighbors. (Compare with
        Figure~\ref{f:crossings_sentence_gadget} and with a modification of
        Figure~\ref{f:sentence_gadget} for $\ell = 1$.) Seemingly, there could
        be also a double edge using $x^{32}_i(u)$ if the lower neighbor of
        $x^{32}_i(u)$ were $y(u)$ or $z^2_i(u,s)$. It is easy to rule out $y(u)$
        as a lower neighbor of $x^{32}_i(u)$ by checking the arcs emanating from
        $y(u)$. For ruling out $z^2_i(u,s)$, we observe that an arc connecting
        $x^{32}_i(u)$ and $z^2_i(u,s)$ emanating down from $x^{32}_i(u)$ has to
        contain, for example, $x^{31}_{i}(u)$ before it may reach $z^2_i(u,s)$
        (we simply pass through the diagram in this direction starting in
        $x^{32}_i(u)$ and check where we get).  Thus $z^2_i(u,s)$ cannot be a
        lower neighbor of $x^{32}_i(u)$.
	
	Next, we rule out subcase (ii), that is $q'$ is not in $Q$ and it
        belongs to at least two arcs corresponding to the edges in the graph. By
        checking the dashed edges, the only option to get two such arcs is if $q
        = x^{32}_i(u)$.  We rule out this case by a similar argument as
        before. Candidate $q'$ to the right of $x^{32}_i(u)$ not in $Q$ may only
        be crossings intersecting the head of $\gamma_i(u)$ (recall that head is
        drawn on Figure~\ref{f:sentence_gadget}).  However, when trying to reach
        any such crossing from $x^{32}_i(u)$ by an arc emanating down, we have
        to meet $x^{31}_{i}(u)$ first again. This finishes the second subcase
        and therefore Case 1 of our analysis.
	
	\bigskip
	Case 2: $q = y(u)$ for some $u \in U$. This is almost trivial case: It
        is straightforward to check that $y(u)$ has no close neighbor $q'$ (as
        the figure is only very local; see Figure~\ref{f:sentence_gadget}).
	
	\bigskip
	Case 3: $q \in Z(u)$ for some $u \in U$. This means that $q = z_i^j(s,
        s')$ for some $s, s' \in S$ such that $R_i(s) = (T, s)$ with $s' \in T$
        and for some $j \in [2]$.  We want to consider the subcases according to
        whether $q = z_i^1(s, s')$ or $q = z_i^2(s, s')$. For both case, we
        follow Figure~\ref{f:interlacing_parallel} when describing the
        directions.

	So, let us first assume that $q = z_i^1(s, s')$. Because we need two
        arcs for close neighobors, we realize that the only candidates for $q'$
        are those that either intersect the head of $\gamma_1(s)$, or possibly
        $x^{32}_i(s)$ (the arc heading up is the only arc leaving the head among
        the arcs emanating from $z_i^1(s, s')$). In addition, the arc emanating
        up heads towards $\gamma_1(s')$ as requested by the construction; see
        Figure~\ref{f:interlacing_parallel}. In particular, it meets
        $x^{31}_0(s')$ before it may meet any of the candidate $q'$. Thus the
        arc emanating up from $z_i^1(s, s')$ cannot be one of the witnesses that
        $q$ and $q'$ are close neighbors. Therefore, the only crossing $q'$
        connected by two arcs to $z_i^1(s, s')$ without points of $Q$ in their
        interiors may be $z_i^2(s, s')$.  But in this case, the
        overpasses/underpasses do not fit the definition.
	
	Now, let us consider the case $q = z_i^2(s,s')$. The reasoning is
        similar. The directions left, up and right reveal that any candidate
        close neighbor $q'$ is in the head or it is the crossing
        $x^{31}_i(s)$. The bottom direction can be ruled out: It heads towards
        $M(s')$ and then to $\gamma_0(s')$. Before reaching the head of
        $\gamma_1(s)$ it thus has to reach $x^{32}_{\ell(s')}(s')$. Using the
        remaining directions, the only possible candidate is $q' = z_i^1(s,s')$.
        However, we have ruled out the pair $\{z_i^1(s,s'),z_i^2(s,s')\}$ as
        close neighbors already in the previous case.
	
	This finishes the case analysis thereby the proof of
        Proposition~\ref{p:verify_boosting} and the hardness part of
        Theorem~\ref{t:main}.
\end{proof}

\bibliographystyle{alpha}
\newcommand{\etalchar}[1]{$^{#1}$}

\end{document}